\documentclass[preprint,12pt]{elsarticle}

\usepackage{amsmath,amsfonts,amsbsy,amssymb,tabulary,graphicx,times,caption,fancyhdr}
\usepackage{amsthm}
\usepackage{enumitem}
\usepackage{stmaryrd}
\usepackage{tikz}

\newcommand{\td}{{\mathsf{td}}}
\newcommand{\px}{P\bm{x}}

\newcommand{\D}{\mathsf{D}}
\newcommand{\tom}{\bigcirc}

\newcommand{\LTD}{{\mathsf{LTD}}}
\newcommand{\LDTV}{{\mathsf{LDTV}}}
\newcommand{\LFD}{{\mathsf{LFD}}}
\newcommand{\FOL}{{\mathsf{FOL}}}
\newcommand{\ETL}{{\mathsf{ETL}}}
\newcommand{\PAL}{{\mathsf{PAL}}}
\newcommand{\term}{{\mathsf{Term}}}
\newcommand{\bv}{{\bm v}}
\newcommand{\bM}{{\mathbf M}}

\newcommand{\bS}{{\mathbf S}}
\newcommand{\ot}{\leftarrow}

\newtheorem{theorem}{Theorem}[section]
\newtheorem{fact}[theorem]{Fact}

\newtheorem{definition}[theorem]{Definition}
\newtheorem{example}[theorem]{Example}
\newtheorem{lemma}[theorem]{Lemma}

\newtheorem{remark}[theorem]{Remark}
\usepackage{bm}

\newcommand{\red}[1]{\textcolor{red}{#1}}

\allowdisplaybreaks[4]

\journal{arXiv}

\begin{document}

\begin{frontmatter}



\title{Dependence Logics in Temporal Settings}

\author[label1]{Alexandru Baltag}

 \author[label1,label2,label3]{Johan van Benthem}

  \author[label4,label5]{Dazhu Li}
 
 \affiliation[label1]{organization={Institute for Logic, Language and Computation, University of Amsterdam},
           city={Amsterdam},
           country={The Netherland}
           }

 \affiliation[label2]{organization={Department of Philosophy, Stanford University},
             city={Stanford},
             country={USA}
             }

 \affiliation[label3]{organization={The Tsinghua-UvA JRC for Logic, Department of Philosophy, Tsinghua University},
            city={Beijing},
           country={China}}

\affiliation[label4]{organization={Institute of Philosophy,
Chinese Academy of Sciences},
            city={Beijing},
           country={China}}

\affiliation[label5]{organization={Department of Philosophy, University of Chinese Academy of Sciences},
            city={Beijing},
           country={China}}



\begin{abstract}
Many forms of dependence manifest themselves over time,  with behavior of variables in dynamical systems as a paradigmatic example. This paper studies temporal dependence in dynamical systems from a logical perspective, by enriching a minimal modal base logic of static functional dependencies. We first introduce a logic for dynamical systems featuring temporalized variables,  provide a complete axiomatic proof calculus, and show that its satisfiability problem is decidable. Then, to capture explicit reasoning about dynamic transition functions, we enhance the framework with function symbols and term identity. Next we combine temporalized variables with a modality for next-time truth from standard temporal logic, where modal correspondence analysis reveals the principles needed for a complete and decidable logic of timed dynamical systems supporting reductions between the two ways of referring to time. Our final result is an axiomatization of a general decidable logic of dependencies in arbitrary dynamical systems. We conclude with a brief outlook on how the systems introduced here mesh with richer temporal logics of system behavior, and with dynamic topological logic.

\end{abstract}




\end{frontmatter}



\section{Introduction: dependencies over time}\label{sec:motivation}
Dependence is a basic phenomenon in many areas, from dynamical systems in the natural sciences to  databases, AI systems, or games. In physical systems, variables like distance typically depend on others, say, velocity and time, and the laws that govern the evolution of such systems reflect real  dependencies between  physical phenomena in the world. But dependence also underlies social behavior. Adopting the well-known strategy of \textit{Tit-for-Tat} in an iterated game means that you copy what your opponent did in the previous round \citep{games}, creating a strong one-step dependence. Or, an agent in a social network may adopt a given behavior or opinion in the next step depending on the proportion of friends who  currently have that behavior, and a global evolution of behaviors for the whole group then again unfolds in a dynamical system \citep{diffusion}. There are suggestive logical aspects to dependence of variables and the associated notion of independence, and by now,  logicians have come up with various frameworks such as \cite{dependence-logic,independence-friendly,AJ-Dependence}. These available frameworks
mostly focus on static dependencies between variables at the same moment in time. We now proceed to investigate the ubiquitous phenomenon of dependencies that occur over time. One might think that this is merely a matter of adding a new variable for time to earlier dependence logics, but that would not reveal much. Instead, we chart the new semantic notions and valid patterns of reasoning that arise when adding time in  more sensitive formalisms.

Our analysis will focus on discrete dynamical systems $\bS=(S,g)$, given by an abstract \emph{state space} $S$ together with a dynamic {transition function} $g$, that maps every state to a `next state'. However, in practice states may be structured (being typically represented as tuples of values in some multi-dimensional space), and our logics reflect this by considering a number of \emph{variables} that may take different values in various states. Thus, each state comes with an associated \emph{assignment} of values to variables. Such variable assignments are a crucial device in the semantics of first-order logic $\FOL$. But crucially, not all possible assignments are necessarily realized at states, and this is what leads to dependencies: a change in value for one variable may only be possible in the system by also changing the value of some other variable. This  interpretation of `assignment gaps' as modeling dependence and correlation is well-known from the logical literature \citep{ABN-1998}.

To talk about dynamic dependencies that manifest themselves over time, we will present a \emph{Logic of Dependence with Temporalized Variables} ($\mathsf{LDTV}$) that combines two components, (i) the modal logic $\mathsf{LFD}$ of static functional dependence from \cite{AJ-Dependence}, and (ii) basic vocabulary suggested by that of temporal logic. The  $\mathsf{LFD}$ component gives us dependence atoms $D_Xy$ expressing dependence of the current value of $y$ on the current values of the variables in $X$, plus modalities $\D_X\varphi$ that express which propositions are forced to be true by the current values of the variables in $X$. Our temporal component employs  temporal terms rather than propositional temporal modalities, e.g., $\tom x$ referring to the `next value' of  variable $x$, i.e., its value at the next state of the system. Combining these devices, formulas  can express how future values of variables depend on earlier ones -- the typical pattern in recursively defining a  transition function for a dynamical system. For instance, consider the recursion equation
\begin{equation*}
  y(t+1)=2x(t),  
\end{equation*}
which states that the value of $y$ in the next step equals two times the current value of $x$. Its dependence pattern can be captured by $D_{\{x\}}\tom y$.

 \begin{remark}\label{remark:different-from-causality}
The recursion  $y(t+1)=2x(t)$ also supports the `backward dependence' $D_{\{\tom y\}} x$. Unlike causality, temporal dependence has no unique  direction.
 \end{remark}

 A slightly more intricate example with interacting variables runs as follows:
 \begin{equation*}
\begin{cases}
x(t+1)=x(t)+y(t)\\
     y(t+1)=2x(t) 
  \end{cases}
\end{equation*}
 These equations give the dependence patterns $D_{\{x,y\}}\tom x$ and $D_{\{x\}}\tom y$. Also, from the two equations, we can infer that $$x(t+1)=x(t)+2x(t-1)$$  inducing the  pattern $D_{\{x,\tom x\}}\tom\tom x$. As we shall show later, the latter can  be formally derived from $D_{\{x,y\}}\tom x$ and $D_{\{x\}}\tom y$ using a  Hilbert-style proof system for the logic $\mathsf{LDTV}$. Over abstract state models for dynamical systems, this system  already expresses interesting basic facts about temporal dependence. 

In a next step, we extend the logic with function symbols and term identity, allowing for the specification of concrete laws defining  dynamical systems. We show that both $\LDTV$ and this functional extension have a complete axiomatization and  a decidable satisfiability problem. We prove these results by spelling out reductions to the dependence logic $\LFD$ which also yield the new insight that the addition of global (as opposed to local) term identities preserves decidability.

Next, moving closer to   ordinary temporal logics, we analyze a general intuition about propositions involving temporalized variables, viz. {\em future truth about variable values is the same as current truth about future values of these variables}. To make this precise we add a modality  $\tom$  to our language for truth at the next moment of system evolution. Via a modal correspondence analysis of the intuition, we then identify the structural conditions needed to make it true. As we shall see, this forces us to shift the perspective slightly from state spaces to  temporal models with histories of executions for dynamical systems. We axiomatize the resulting dynamic dependence logic of `timed models' plus its extension with function symbols and global term identity, and show  decidability and axiomatizability using a reduction algorithm to our previous system $\LDTV$ inspired by  dynamic-epistemic logic \citep{BMS,Johan-del,hans-del}. 

Finally, we reach the point where merely adding temporalized variables to static dependence logics  will no longer suffice in the analysis of dynamical systems. We introduce a more general temporal dependence logic $\LTD$ valid for arbitrary dynamical systems which does not support a reduction from the propositional temporal modality to non-modal statements about temporalized variables: the two ways of referring to time have  become mutually irreducible. We axiomatize this final logic as well and show its decidability: now no longer via reduction techniques, but with a modal filtration-based approach, a much more elaborate version of the one  used for analyzing the original dependence logic $\LFD$.

Our analysis in this article is just a start, and in a final part  we briefly discuss some extensions that make sense. First, one can view our systems as  fragments of richer temporal dependence languages for dynamical systems which can say more about system evolution over time, in the spirit of multi-agent epistemic-temporal logics \citep{reason-about-knowledge}. Another natural extension is topological, since dynamical systems often come with topologies on the sets of values for variables, or on the state space itself. This line connects our systems with temporal-topological logic of dynamical systems \citep{handbook-dtl,artemov-1997}, but also with  richer logics of continuous dependence \citep{AJ-CD}. Finally, given the natural alternative interpretation of modal dependence logic as a  logic of information and knowledge, we point out some interesting connections with epistemic-temporal logics of agency \citep{reason-about-knowledge, Benthem2006TheTO,etl-1989}.

\vspace{3mm}

\noindent{\bf Structure of the paper}\; Section \ref{sec:o-term-logic} introduces the logic of dependence with temporalized variables $\LDTV$ and offers various results, including decidability and a complete Hilbert-style  proof system. Section \ref{sec:global-equality} defines an extended system $\LDTV^{\mathsf{f},\equiv}$  with function symbols  and a  global equality. We show that  decidability and axiomatizability remain. Next, connecting to standard  temporal logics,  Section \ref{sec:next-time-modality} analyzes a way of reducing future truth to current truth and determines the class of dynamical systems for which this holds. The resulting logic $\LTD^{\mathsf{t}, \mathsf{f}, \equiv}$ is axiomatized and shown decidable in Section \ref{sec:next-time-modality-axiomatization-timed-semantics}. Finally, in  Section \ref{sec:next-time-modality-non-timed-semantics}, we define a general temporal dependence logic $\LTD$ for arbitrary dynamical systems models, and show its completeness and decidability in the most complex proof in this paper. Wrapping up, Section \ref{sec:further-directions} closes the article with further directions deserving to be explored in the future. A number of detailed technical proofs for results stated in the preceding sections have been collected in a series of Appendices.

\section{Logic of dependence with temporalized variables}\label{sec:o-term-logic}

 In this section, we introduce the system $\mathsf{LDTV}$, the Logic of Dependence with Temporalized Variables. We will  show that its satisfiability problem is decidable and provide a complete Hilbert-style proof system. But first, here are the basics.

\subsection{Language and semantics of $\LDTV$}
The language of $\mathsf{LDTV}$ is based on a  \emph{vocabulary} $\mathsf{V}=(V, Pred, ar)$, where $V=\{v_1, \ldots, v_N\}$  is a finite set  of symbols denoting \emph{basic variables};  $Pred$ is a set of \emph{predicate} symbols, and $ar: Pred\to \mathbb{N}$ is an \emph{arity} function assigning to each  $P\in Pred$ some arity $ar(P)\in
\mathbb{N}$.

\begin{definition}\label{def:language}
Given a vocabulary $\mathsf{V}=(V, Pred, ar)$, we construct \emph{terms} $x$ denoting `dynamical variables', as well as \emph{formulas}
$\varphi\in\mathcal{L}$ via the following:
\begin{center}
 $x ::= v \mid \tom x$\\
$\varphi ::= P(x_1,\ldots , x_k)\mid\neg\varphi\mid \varphi\land\varphi\mid  \D_X\varphi\mid D_X x$
\end{center}

\noindent where $v\in V$ is any basic variable symbol, $x, x_1, \ldots, x_k$ are terms, $P$ is a predicate symbol of some arity $k$, and $X$ is any non-empty finite set of terms.\footnote{In the dependence logic $\LFD$ sets of variables can also be empty, which allows us to define a universal modality in that setting. For concrete reasoning about  dynamical systems, non-empty sets make more sense, and there are also some technical issues with temporal dependence logics containing a universal modality which we will not consider in this paper.}  We use the notation $\term$ for {\em the set of all terms} and $\widehat{\D}_X\varphi$ for {\em the dual of} $\D_X\varphi$.
\end{definition}

Abbreviations $\top,\bot,\lor,\to,\leftrightarrow$ are as usual. The dynamical term $\tom x$ captures the \emph{next value of $x$}, i.e., the value that $x$ will have at the next moment in time. The readings of atomic formulas $\px$ and Boolean connectives $\neg,\land$ are also as usual. Formula $\D_X\varphi$ states that {\em the current values of the variables in $X$ fix the truth of $\varphi$} and $D_Xy$ says that {\em the current values of the variables in $X$ fix the value of $y$}.

Next, for any non-empty $X\subseteq \term$ and  $n\in\mathbb{N}$, we define $\tom X := \{\tom x: x\in X\}$, $\tom^0 X := X$, and $\tom^{n+1} X:= \tom \tom^n X$. For  arbitrary formulas $\varphi$, we use $\varphi[V/ \tom^n V]$ for the formula obtained by replacing each $v\in V$ in $\varphi$ with $\tom^n v$. We measure nestings of $\tom$ in a term or a set of terms using the following  notion:
 
 \begin{definition}\label{def:temporal-depth-terms-LTD}
 For any $v\in V$, $x\in \term$ and finite $X\subseteq \term$, we define the {\em temporal depth} of terms and finite sets $X$ of terms   in the following:
 \begin{center}
   $\mathsf{td}(v):=0$, \quad $\mathsf{td}(\tom x):=1+\mathsf{td}(x)$, \quad $\td(X):={\rm{max}}\{\td(x): x\in X\}$.  
 \end{center}
 \end{definition}

Having presented the basic syntax, we now move to the semantics. 

 \vspace{2mm}

\par\noindent\textbf{Dynamical dependence models}\; A \emph{typed $\FOL$-model} is a multi-typed structure $M=(\mathbb{D}_v, I)_{v\in V}$, indexed by variables $v\in V$ (each thought as having its own distinct type), where: for each $v\in V$, $\mathbb{D}_v$ is the range of values of variable $v$; and $I$ is an interpretation function, mapping each $P\in Pred$ of arity $n$ into some $n$-ary relation on the union of all values $\bigcup_{v\in V} \mathbb{D}_v$.

Given a vocabulary $\mathsf{V}=(V, Pred, ar)$, a \emph{dynamical dependence model} for $\mathsf{V}$ (or `{\em dynamical model}', for short) is a structure $\bM=(M, \bS, \bm{v})_{v\in V}$, consisting of: a dynamical system $\bS=(S,g)$; a typed $\FOL$-model $M=(\mathbb{D}_v, I)_{v\in V}$; and, for each variable $v\in V$, a corresponding dynamical variable, i.e., a map $\bm{v}:S\to \mathbb{D}_v$. Moreover, our dynamical systems will be fully specified in terms of the values of the basic
variables. This means that {\em dynamical states are uniquely determined by the values of all the basic variables}, i.e., they satisfy the following condition:
\begin{center}
    if $\bm{v}(s)=\bm{v}(t)$ for all $v\in V$, then $s=t$.
\end{center}
Then, each state $s\in S$ corresponds to a unique variable assignment $\overline{s}$, assigning to each basic variable
$v\in V$ its value $\overline{s}(v):=\bm{v}(s)$. In other words, states can be identified with basic variable assignments.

\vspace{2mm}

\par\noindent\textbf{Semantics: interpretation of (sets of) terms}\; Given a dynamical model $\bM=(M, \bS, \bm{v})_{v\in V}$, we can extend the interpretation $\bm{v}$ of basic variables in $V$ to arbitrary terms $x$, as well as
finite sets of terms $X$. These are interpreted as dynamical variables $\bm{x}: S\to \mathbb{D}_x$, $\bm{X}:S\to \mathbb{D}_X$, as
follows: for basic variable symbols $v\in V$, $\mathbb{D}_v$ and $\bm{v}$ are already given by the dynamical model; for other single terms $\tom x$, we recursively put
$$\mathbb{D}_{\tom x}:=\mathbb{D}_x,  \,\, \, (\tom \bm{x})(s):=\bm{x}(g(s));$$
and for sets of terms $X$, we take
\begin{center}
 $\mathbb{D}_{X}:=\Pi_{x\in X} \mathbb{D}_x, \,\,\, \bm{X}(s):= ({\bm x}(s))_{x\in X}.$
\end{center}

\par\noindent\textbf{Value  agreement}\; Let $\bS=(S,g)$ be a dynamical system and $X$ be a finite set of terms. We can introduce an
equivalence relation $=_X$ on $S$, called \emph{$X$-value agreement}, by putting
\begin{center}
  $w =_X w'$ iff   $\bm{x}(w)=\bm{x}(w')$ for all $x\in X$.
\end{center}
In particular, for singletons $\{x\}$, we use $=_x$ as an abbreviation for $=_{\{x\}}$.

\vspace{2mm}

\noindent\textbf{Semantics: interpretation of formulas}\; The semantics of $\mathsf{LDTV}$ is given by the same clauses as in the semantics of $\LFD$ \cite{AJ-Dependence}, but applied to the richer language $\mathcal{L}$:

\begin{definition}\label{def-semantics-O-term-logic} Given a dynamical model $\bM$, we define the \emph{truth of a formula} $\varphi\in\mathcal{L}$ \emph{in} $\bM$ \emph{at} a state $s\in S$, written  $s\vDash_{\bM}\varphi$, in the following recursive format:
\begin{center}
\begin{tabular}{r c l }
$s\vDash_{\bM} P(x_1,\ldots , x_n)$ & \quad iff  &\quad $(\bm{x_1}(s),\ldots,\bm{x_n}(s))\in I(P)$\\
 $s\vDash_{\bM}\neg\varphi$    &\quad  iff &\quad  not $ s \vDash_{\bM}\varphi$\\
 $s\vDash_{\bM} \varphi\land\psi $&\quad  iff &\quad $s\vDash_{\bM}\varphi \, \,and \, \, s\vDash_{\bM}\psi$\\
 $s\vDash_{\bM} D_X y$&\quad  iff &\quad  for every $w\in S$, $s=_X w$ implies $s=_y w$\\
 $s\vDash_{\bM} \D_X\varphi$&\quad iff &\quad for every $w\in S$, $s=_X w$ implies $w\vDash_{\bM}\varphi$
\end{tabular}
\end{center}
\end{definition}

To ease notational clutter, we often omit the subscript on the truth relation denoting the model when the latter is clear from the context. Also, notions of {\em satisfiability}, {\em validity} and {\em logical consequence} are defined as usual. 

Since the system $\LDTV$ is a direct extension of $\LFD$, the validities of the latter are evidently still valid in the new setting, including
\begin{itemize}
\item[$\bullet$] $ D_XY\land D_ZU\to D_{X\cup Z}(Y\cup U)$ \hfill  (Additivity of Dependence)
\item[$\bullet$] $  D_Xy \to D_Zy$, \quad for $X\subseteq Z$ \hfill  (Monotonicity of Dynamic Dependence)
 \item[$\bullet$] $ \D_X\varphi\to \D_Y\varphi$,\quad for $X\subseteq Y$ \hfill (Monotonicity of Dependence Quantifiers)
\end{itemize}
  But it is important to notice that these are now  about complex terms with temporalized variables. Besides the above, the formulas below are also valid in $\LDTV$:
\begin{itemize}
    \item $D_{\tom^n V}\tom^{n+m} v$, where $V$ is the whole set of variables and $m,n\in\mathbb{N}$.
    \item  $\D_V\varphi\leftrightarrow \varphi$, where $V$ is the whole set of variables.
\end{itemize}

\subsection{Decidability of  the satisfiability problem for
 $\mathsf{LDTV}$}\label{sec:o-term-logic-decidability}
To prove decidability of $\mathsf{LDTV}$, we present a reduction to static dependence logic which may have some independent interest as a general method. We will embed $\mathsf{LDTV}$ into {\em $\LFD^{\mathsf{f}}$ that extends $\mathsf{LFD}$ with function symbols}, which is decidable \cite{AJ-Dependence}. More precisely, we can faithfully translate $\mathsf{LDTV}$ into an $\LFD$-language with  $N$ functional symbols $\{f_v: v\in V\}$ of arity $N$, where  $V=\{v_1, \ldots, v_N\}$ is the set of all variables in the vocabulary of $\mathcal{L}$.  

 \vspace{2mm}
 
 \noindent{\bf Translation from $\mathcal{L}$-formulas into $\LFD^{\mathsf{f}}$-formulas}\; First, we associate to each $x\in \term$
 a translation $\textbf{Tr}(x)$ as a functional term in $\LFD^{\mathsf{f}}$, and associate to each $\varphi\in\mathcal{L}$ a
 translation $\textbf{Tr}(\varphi)$ into $\LFD^{\mathsf{f}}$, by the following inductive clauses:
\begin{center}
$\textbf{Tr}(v):= v$, \quad $\textbf{Tr}(\tom^{n+1} v):= f_v (\textbf{Tr}(\tom^nv_1), \ldots, \textbf{Tr}(\tom^nv_N))$  \\
\vspace{2mm}
$\textbf{Tr}(P x_1,\ldots, x_n):= P (\textbf{Tr}(x_1),\ldots, \textbf{Tr}(t_n))$, \quad $\textbf{Tr}(D_X y): =
D_{\textbf{Tr}(X)} \textbf{Tr}(y)$,\\
\vspace{1mm}
$\textbf{Tr}(\neg \varphi) :=\neg \textbf{Tr}(\varphi)$, \quad $\textbf{Tr}(\varphi\wedge \psi):=\textbf{Tr}(\varphi) \wedge
\textbf{Tr}(\psi)$,\\
\vspace{1mm}
$\textbf{Tr}(\D_X\varphi):=\D_{\textbf{Tr}(X)}\textbf{Tr}(\varphi)$,
\end{center}
\noindent where $\textbf{Tr}(X)=\{\textbf{Tr}(x): x\in X\}$. 

\vspace{2mm}
 
\noindent{\bf From $\LFD^{\mathsf{f}}$-models to  dynamical models}\; Let $\bM=(\mathcal{O},I,A)$ be a model for
the
$\LFD^{\mathsf{f}}$-language described above. We can convert it into a  
dynamical model
$\bM^{\downarrow}=(\mathbb{D}_v,I',S,g,{\bm v})_{v\in V}$ over variables $V$, by taking:
\begin{itemize}
\item[$\bullet$] For all $v\in V$, $\mathbb{D}_{v}:=\mathcal{O}\times \mathbb{N}$.
\item[$\bullet$] $I'$ is the interpretation map for predicate symbols, and for each $P$ of arity $n$ and $o_1,\dots,o_n\in \mathcal{O}$,
\begin{center}
    $((o_1,i_1),\dots, (o_n,i_n))\in I'(P)$ iff $(o_1,\dots, o_n)\in I(P)$
\end{center}
where $i_1,\ldots,i_n\in\mathbb{N}$.
\item[$\bullet$] $S:=\{(s, i): i\in\mathbb{N}\;\textit{and}\; s\in A  \}$.
\item[$\bullet$] $g(s, i):= (s,i+1)$.
\item[$\bullet$] For variables $v\in V$,  ${\bm v}(s, i)$ are recursively defined by putting:
\begin{align*}
 {\bm v}(s,0)   & :=({\bm v}(s),0)\\
 {\bm v}(s,n+1) & :=(I(f_v)([{\bm v_1}(s,n)]^-,\dots,[{\bm v_N}(s,n)]^-),n+1)
\end{align*}
where for each $v_m$, $[{\bm v_m}(s,n)]^-$ is the operation that removes the natural number from the tuple ${\bm v_m}(s,n)$ (e.g., $[{\bm v}(s,0)]^-=[({\bm v}(s),0)]^-={\bm v}(s)$), and $I(f_v)$ is the interpretation of $f_v$ as an actual $N$-ary function in $\bM$.
\end{itemize}

By construction, for any $(s,n),(t,m)\in S$ and term $x$, $(s,n)=_x(t,m)$ implies $n=m$. Also, we have the following observations on this construction:

\begin{fact}\label{fact:timed-LTD-model-transition-1}
In the resulting $\bM^{\downarrow}$, for all $x\in \term$ and $(s,n)\in S$, it holds that: $${\bm x}(s,n)=({\bf Tr}(\tom^n  { \bm x})(s), n+\mathsf{td}(x) ).$$
\end{fact}
\begin{proof}
    See  \ref{appendix:proof-for-fact:timed-LTD-model-transition-1}.
\end{proof}

\begin{fact}\label{fact:timed-LTD-model-transition-2}
For all sets of terms $X$, assignments $s,w\in A$ and  $n,m\in\mathbb{N}$,
\begin{center}
$(s,n)=_X(w,m)$ \quad iff \quad $n=m$ and $s=_{{\bf Tr}(\tom^n X)} w$.  
\end{center}
\end{fact}

\begin{proof}
It holds by the definition of $\bM^{\downarrow}$ and Fact \ref{fact:timed-LTD-model-transition-1}.
\end{proof}

In what follows, for any formula $\varphi$, we denote by $\varphi[v_1/x_1,\ldots,v_{N}/x_N]$ the formula obtained by replacing every basic variable $v_i$ occurring in $\varphi$ with $x_i$.     Now, it is crucial to see that this construction is truth-preserving for $\mathsf{LDTV}$:

\begin{fact}\label{fact:LTD-t-lfd-1}
For all formulas $\varphi\in\mathcal{L}$ and all $\LFD^{\mathsf{f}}$-models $\bM$ of the appropriate type and all pairs $(s,n)\in S$:
\begin{center}
$s\vDash_{\bM} {\bf Tr}(\varphi[V/\tom^n V])$\; iff\; $(s,n)\vDash_{\bM^{\downarrow}} \varphi$.
\end{center}
\end{fact}
\begin{proof}
    See  \ref{appendix:proof-for-fact:LTD-t-lfd-1}.
\end{proof}
 
\noindent{\bf From dynamical dependence models to $\LFD^{\mathsf{f}}$-models}\; Conversely, given a dynamical model
$\mathcal{M}=(\mathbb{D}_v,I,S,g,{\bm v})_{v\in V}$, we can convert it into an $\LFD^{\mathsf{f}}$-model $\mathcal{M}^+=(\mathcal{O},I^+,A)$
for a language with $N$ variables $V$ and $N$ $N$-ary functions $\{F_v: v\in V\}$ as follows:
\begin{itemize}
\item[$\bullet$] $\mathcal{O}:= \{\infty\}\cup\bigcup_{v\in V}\mathbb{D}_v$, where $\infty$ is a fresh object.
\item[$\bullet$] $A:=\{s^+: s\in S\}$ s.t. ${\bm v}(s^+):= {\bm v}(s)$ for each $v\in V$.
\item[$\bullet$] For all $P$, $I^+(P):=I(P)$.

For all $F_v$, if there is $s\in S$ such that ${\bm v_1}(s)=o_1, \ldots, {\bm v_N}(s)=o_N$, then $I^+(F_v) (o_1, \ldots, o_N):= {\bm v}(g(s))$, and otherwise, if no such $s$ exists, putting $I^+(F_v) (o_1, \ldots, o_N):= \infty$.
\end{itemize}

For any $F_v$, the interpretation $I^+(F_v)$ is well-defined: to see this, one just needs to notice that dynamical states are uniquely
determined by the values of all the basic variables. Again, this construction is truth-preserving:

\begin{fact}\label{fact:LTD-t-lfd-2}
For all $\varphi\in \mathcal{L}$ and all  dynamical models $\mathcal{M}$ of the appropriate type and all admissible assignments $s$ of $\mathcal{M}^+$, we have:
\begin{center}
  $s\vDash_{\mathcal{M}} \varphi$ \,  iff  \, $ s^+\vDash_{\mathcal{M}^+}  {\bf Tr}(\varphi)$.
\end{center}
\end{fact}

\begin{proof}
The proof goes by induction on formulas $\varphi\in\mathcal{L}$. The cases for atoms $\px$ and Boolean connectives $\land,\neg$ are
straightforward. Here are the other cases:

(1). Formula $\varphi$ is of the form $D_Xy$. To see this, one just needs to notice that for any set $X$ of $\mathsf{LDTV}$-terms,  $s^+=_{{\bf{Tr}}(X)} t^+$ iff $s=_Xt$.

(2). Formula $\varphi$ is of the form $\D_X\psi$. Using the observation in (1) above plus the inductive hypothesis, we  obtain the required equivalence immediately.
\end{proof}

Facts \ref{fact:LTD-t-lfd-1} and \ref{fact:LTD-t-lfd-2} show that $\mathsf{LDTV}$ can be thought of as a fragment of $\LFD^{\mathsf{f}}$,
and also that $\mathsf{LDTV}$-satisfiability of a formula is the same as $\LFD^{\mathsf{f}}$-satisfiability of its translation. Immediately, this gives us the following:

\begin{theorem}\label{theorem:decidability-LTD-t}
The satisfiability problem for $\mathsf{LDTV}$ is decidable.
\end{theorem}

\subsection{Axiomatization of $\mathsf{LDTV}$}\label{sec:o-term-logic-axiomatization}

In addition to its  decidability, the system $\mathsf{LDTV}$ is also finitely axiomatizable, and a proof system $\mathbf{LDTV}$ for the logic is as follows:  

\vspace{1.5mm}

\noindent{\bf Proof system $\mathbf{LDTV}$}\; The details of the calculus are given in Table \ref{table:axiomatization-LTD}. It extends the calculus  $\mathbf{LFD}$ for $\LFD$ \cite{AJ-Dependence} with an axiom of {\em Determinism}, stating that {\em whenever we fix a state, we can fix the values of all variables at the next stage}, i.e., transitions only depend on global system states, not on when these states occur. Also, the axiom {\em $\D$-Distribution} is standard for normal modalities and the dependence quantifiers $\D_X\varphi$ satisfy \textbf{S5}-axioms.
The notions of \emph{syntactical derivation} and \emph{provability}
are  as usual.

 \begin{table}
\caption{The proof system {\bf LDTV}}
\label{table:axiomatization-LTD}
\centering
\begin{tabular}{lll}
\cline{1-3}
\textbf{I} &\quad& \textbf{Axioms and Rules of  Classical Propositional Logic} \\
\cline{1-3}
\textbf{II}&\quad& \textbf{Axioms and Rules for $\D_X\varphi$} \\
\cline{1-3}
 $\D$-Distribution  &\quad& $\D_X(\varphi\to\psi)\to(\D_X\varphi\to \D_X\psi)$\vspace{1mm} \\
 $\D$-Introduction$_1$ &\quad& $P(x_1,\ldots, x_n)\to \D_{\{x_1,\ldots,x_n\}}P(x_1,\ldots, x_n)$\vspace{1mm} \\
 $\D$-Introduction$_2$ &\quad& $D_Xy\to \D_XD_Xy$\vspace{1mm} \\
$\D$-T &\quad& $\D_X\varphi\to\varphi$\vspace{1mm} \\
$\D$-4 &\quad& $\D_X\varphi\to \D_X\D_X\varphi$\vspace{1mm} \\
$\D$-5 &\quad& $\neg \D_X\varphi\to \D_X\neg \D_X\varphi$\vspace{1mm} \\
$\D$-Necessitation &\quad& From $\varphi$, infer $\D_X\varphi$\vspace{1mm} \\
\cline{1-3}
\textbf{III}&\quad& \textbf{Axioms for $D_Xy$} \\
\cline{1-3}
Dep-Ref &\quad& $D_Xx$ \, for all $x\in X$\vspace{1mm} \\
Dep-Trans &\quad& $D_XY\land D_YZ\to D_XZ$\vspace{1mm} \\
Determinism  &\quad& $D_{V}\tom v$,\, for all variables $v\in V$ \vspace{1mm} \\
\cline{1-3}
\textbf{IV}&\quad& \textbf{Interaction Axiom} \\
\cline{1-3}
Transfer  &\quad&$D_XY\land\D_Y\varphi\to \D_X\varphi$\vspace{1mm} \\
\cline{1-3}
\textbf{V}&\quad&\textbf{Substitution for Temporalized Variables} \\
\cline{1-3}
Substitution$_{\tom}$  &\quad&  From $\varphi$, infer $\varphi[V/\tom V]$\vspace{1mm} \\
\cline{1-3}
\end{tabular}
\end{table}

\vspace{1.5mm}

At a first glance, one may think that the given proof system misses some important axioms, including $\neg P(x_1,\ldots, x_n)\to \D_{\{x_1,\ldots,x_n\}}\neg P(x_1,\ldots, x_n)$. However,  the latter principle is provable from the system, and so are all  earlier-mentioned semantic validities. The latter can then be used  to show that the formula $D_{\{x,\tom x\}}\tom\tom x$ in Section \ref{sec:motivation} can be proved  from the premises $D_{\{x\}}\tom y$ and $D_{\{x,y\}}\tom x$: 

(1). by Dep-Ref we have $D_{\{\tom x\}}\tom x$; (2). by Additivity of Dependence, $D_{\{x\}}\tom y\land D_{\{\tom x\}}\tom x \to D_{\{x,\tom x\}} \{\tom y,\tom x\}$; (3). using propositional logic, we get $D_{\{x,\tom x\}} \{\tom y,\tom x\}$; (4). applying Substitution$_{\tom}$ to $D_{\{x,y\}}\tom x$ gives  $D_{\{\tom x,\tom y\}}\tom\tom x$; finally,  (5). using propositional logic, we  obtain $D_{\{x,\tom x\}} \tom\tom x$ with Transfer and the formulas derived in steps (4) and (5), as desired.

\smallskip
The soundness of  {\bf LDTV} is easy to see:

\begin{fact}\label{fact:soundness-LTD}
   The proof system  {\bf LDTV} is sound w.r.t. dynamical dependence models.
\end{fact}

The calculus is also complete w.r.t. dynamical dependence models. To show this, we will make use of $\LFD^{\mathsf{f}}$ again. As was shown in \cite{AJ-Dependence}, a complete proof system ${\bf LFD^{f}}$ for this logic can obtained by extending the proof system
{\bf LFD} with the following axiom, for all terms $x_1, \ldots, x_n$ and $n$-ary function symbols $f$:\footnote{The  axiomatization of   $\LFD^{\mathsf{f}}$  in \citep{AJ-Dependence} also contains a Term Substitution Rule `\emph{from  $\varphi$, infer $\varphi[x/t]$}', where $\varphi[x/t]$ is the formula obtained by replacing $x$ in $\varphi$ with an arbitrary term $t$. As we already stated the axioms of $\bf{LFD}$ with arbitrary terms (see Table \ref{table:axiomatization-LTD}), we can omit this rule.}
\begin{align*}
 &D_{\{x_1,\ldots, x_n\}}f(x_1, \ldots, x_n).
\end{align*}

\noindent By simply inspecting the proof of completeness in \citep{AJ-Dependence}, and noting that all steps go through if we restrict both the language and the axioms to terms in some sub-term-closed set of terms $T$, we obtain:

\begin{fact}\label{fact:completeness-lfd}
Let $T$ be any (finite or infinite) set of $\LFD^{\mathsf{f}}$-terms that is closed under sub-terms. The logic
$\LFD^{\mathsf{f}}_{\mathsf{T}}$ with functional terms restricted to $T$ is completely axiomatized by the system ${\bf LFD^f_T}$ obtained by
restricting ${\bf LFD^f}$ to formulas that use only terms in the set $T$.
\end{fact}

Now we can show the completeness of the calculus {\bf LDTV} w.r.t. dynamical dependence models. A first step towards this is as follows:

\begin{theorem}\label{theorem:theorem-preserving}
   The translation ${\bf Tr}$ from formulas of $\LDTV$ to formulas of $\LFD$  is theorem-preserving, i.e.,  $\varphi$ is a theorem in {\bf LDTV} iff ${\bf Tr}(\varphi)$ is a theorem in ${\bf LFD^f}$.
\end{theorem}

\begin{proof}
We put $T:=\{{\bf Tr}(x): x\in\term\}$, which is closed under sub-terms. The earlier-defined translation map {\bf Tr} from $\term$ to $T$ is bijective. Moreover, the map {\bf Tr} from $\LDTV$-formulas to its range is bijective.  

\vspace{1mm}

Now, we suppose that $\not\vdash_{{\bf LFD^f_T}} {\bf Tr}(\varphi)$. Then, from the completeness of ${\bf LFD^f_T}$ (Fact
\ref{fact:completeness-lfd}), we know that ${\bf Tr}(\neg\varphi)$ is satisfiable. By Fact \ref{fact:LTD-t-lfd-1}, $\neg\varphi$ is satisfiable. By the soundness of ${\bf LDTV}$, it holds
that $\not\vdash_{{\bf LDTV}}\varphi$.

\vspace{1mm}

For the converse direction, note that for every axiom $AX$ of ${\bf LFD^f_T}$,  ${\bf Tr}^{-1}(AX)$ is an axiom or a theorem of ${\bf
LDTV}$; and every correct application  of a rule of ${\bf LFD^f_T}$ is  mapped by ${\bf Tr}^{-1}$ into a correct
application of a rule in ${\bf LDTV}$. It follows that, if $\vdash_{{\bf LFD^f_T}}{\bf Tr}(\varphi)$, then $\vdash_{{\bf LDTV}}
{\bf Tr}^{-1}({\bf Tr}(\varphi))$, i.e., $\vdash_{{\bf LDTV}}\varphi$.
\end{proof}

As an immediate corollary, we have what we want:

\begin{theorem}\label{theorem:completeness-LTD-t}
{\bf LDTV} is complete w.r.t. dynamical dependence models. 
\end{theorem}







\section{Global equality and functions}\label{sec:global-equality}

\subsection{Introducing the system} As we have seen in our motivating examples, dynamical systems are driven by laws that determine the transition function, and these laws often have an equational format.   We can explicitly express such functional dynamic laws without losing decidability, and this section will show how. We will extend $\LDTV$ with \emph{function symbols} $f(x_1, \ldots, x_n)$ (and in particular constants $c$), as well as \emph{term identity}, i.e., global equality formulas $x\equiv y$  (stating that $x$ and $y$ have the same value in \emph{all} states of the system).\footnote{We thus depart from standard $\LFD$ methodology, by adding global statements $\D_\emptyset \, x = y$
without their local versions $x=y$. The reason for this departure is that it is known that the extension of $\LFD$ with explicit equality $v=v'$, interpreted locally at states, leads to
undecidability \citep{Aachen1}. Since $\LFD$ can be embedded in our logic $\LDTV$, the same applies to
its extension with explicit equality. In contrast, adding global equality formulas $x\equiv y$ between terms $x,y$ is an innocuous move, that preserves decidability and axiomatizability.} We will write $\LDTV^{\mathsf{f},\equiv}$ for {\em the resulting logic} and $\mathcal{L}^{f, \equiv}$ for {\em its language}.\footnote{A vocabulary  $\mathsf{V}=(V, Pred, ar,
Funct)$ for $\mathcal{L}^{f, \equiv}$ extends that for $\mathcal{L}$ with a set $Funct$ of function symbols, where the arity function $ar$, besides applying to predicate symbols, also assigns each function symbol $f$ an arity $ar(f)\in\mathbb{N}$.} Also,  the notion of {\em temporal depth} for $\LDTV^{\mathsf{f},\equiv}$-terms extends that of $\LDTV$-terms (Definition \ref{def:temporal-depth-terms-LTD}) with the following clause:
\begin{center}
   $ \mathsf{td}(f(x_1,\dots,x_n)):=\mathsf{min}\{\mathsf{td}(x_1),\dots,\mathsf{td}(x_n)\}$
\end{center}

\noindent\textbf{Semantics}\;
We interpret $\mathcal{L}^{f, \equiv}$ on appropriate \emph{dynamical dependence models with functions}: essentially, the only change is that in the
notion of typed $\FOL$-model $M=(\mathbb{D}_v, I)_{v\in V}$, the interpretation map $I$ has to be extended to function symbols, mapping each
such symbol $f$ of arity $n$ into some $n$-ary function $I(f): (\bigcup_{v\in V}\mathbb{D}_v)^n \to \bigcup_{v\in V}\mathbb{D}_v$. Given this, the semantics of
$\LDTV^{\mathsf{f},\equiv}$ is
obtained by extending the semantics of $\LDTV$ with the following: for terms $f(x_1, \ldots, x_n)$, 
\begin{center}
 $\mathbb{D}_{f(x_1, \ldots, x_n)}:= \bigcup_{v\in V}\mathbb{D}_v,  \,\, \,\,\, \, \bm{f(x_1, \ldots, x_n)} (s) := I(f) (\bm{x_1}(s),
\ldots, \bm{x_n} (s)),$
\end{center}
while for term-identity formulas we put
\begin{center}
 $s\vDash_\bM x\equiv y \,\, \mbox{ iff } \,\, \bm{x}(w)=\bm{y}(w) \mbox{ for all $w\in S$}.$
\end{center}

\medskip

\par\noindent\textbf{Application: specifying dynamical laws}\;  The system  $\LDTV^{\mathsf{f},\equiv}$ can specify the
one-step dynamical laws governing a given dynamical system. For example, the equations $x(t+1)=x(t)+y(t)$ and 
 $y(t+1)=2x(t)$ can be defined with, e.g. $\tom x \equiv f_1(x,y)$ and $\tom y=f_2(x)$, respectively, where $f_1$  stands for the sum of two terms and $f_2$ for multiplication by 2. The equation $x(t+1)=x(t)+2x(t-1)$ following from the above two equations is then   captured by $\tom x\equiv f_1(\tom x, f_2(x))$. In the remainder of the section, we will  provide a complete Hilbert-style proof system for the logic of this sort of reasoning, and in particular,  $\tom x\equiv f_1(\tom x, f_2(x))$ can then be proved formally from  $\tom x \equiv f_1(x,y)$ and $\tom y=f_2(x)$.

\subsection{Axiomatization and decidability}\label{sec:global-equality-axiomatization-decidability}

In what follows, we prove that  $\LDTV^{\mathsf{f},\equiv}$ is decidable and completely axiomatizable by extending the  reduction method  for $\LDTV$. We first consider the logics $\LFD^{\equiv}$ and $\LFD^{\mathsf{f}, \equiv}$ obtained by adding to the  language of $\LFD$
identity atoms $x\equiv y$ to obtain $\LFD^{\equiv}$, and then further adding
function symbols to obtain $\LFD^{\mathsf{f}, \equiv}$. The corresponding proof systems $\textbf{LFD}^{\equiv}$ and $\bf{LFD^{f, \equiv}}$
are obtained by adding to $\textbf{LFD}$ the principles {\em Reflexivity of Identity} and {\em Substitution of Identicals} in Table \ref{table:aximotization-LTD-equality} in the
first case, and  also adding a {\em Functionality Axiom} in the  second case. The soundness of the two proof systems is obvious. We now
establish  decidability and  completeness  for $\LFD^{\mathsf{f}, \equiv}$ and $\LFD^{\equiv}$, starting with the latter system.

\begin{fact}\label{fact:decidability-lfd-equiv}
The satisfiability problem for $\LFD^{\equiv}$ is decidable. Moreover, the proof system ${\bf LFD^{\equiv}}$ is complete for this logic.
\end{fact}

\begin{proof}
The proof uses a translation back into $\LFD$. For the decidability of $\LFD^{\equiv}$, note first that all terms of this logic are
basic variables in
$V$. For any equivalence relation $\mathcal{E}\subseteq V\times V$ on variables in $V$, we put
\begin{center}
$\chi_{\mathcal{E}} \, :=\, \bigwedge_{(v,v')\in \mathcal{E}} (v\equiv v') \wedge \bigwedge_{(v,v')\not\in \mathcal{E}} \neg (v\equiv v')$.
\end{center}

Let $\varphi$ be a formula of $\LFD^{\equiv}$. Assuming a canonical enumeration without repetitions of all  variables in $V$, we can
define, for each such family $\mathcal{E}\subseteq V\times V$, a `translation' $T_{\mathcal{E}} (\varphi)$ into $\LFD$, by: (a) replacing
every
occurrence of any variable $v$ in our formula by the first variable $v'$ (according to our enumeration) such that $(v, v')\in \mathcal{E}$; then
(b) substituting in the resulting formula every subformula of the form $v\equiv v'$ with either $\top$ when $v=v'$, or with $\bot$ when $v\neq
v'$.
Using the global nature of equality $\equiv$, we can
easily see that, for every $\mathcal{E}$, the equivalence
\begin{center}
  $(\varphi \wedge \chi_{\mathcal{E}})\leftrightarrow (T_{\mathcal{E}} (\varphi)  \wedge \chi_{\mathcal{E}})$  
\end{center}
is provable in  $\textbf{LFD}^{\equiv}$, and thus valid. Combining this with the equivalence
\begin{center}
 $\varphi \, \leftrightarrow \, \bigvee \{\varphi\wedge  \chi_{\mathcal{E}}: {\mathcal{E}} \textit{~is an equivalence relation on~} V\}$   
\end{center}
(which is also provable in the calculus  $\textbf{LFD}^{\equiv}$), we obtain the following:

\vspace{2mm}

\noindent\textbf{Claim 1:}\; $\vdash_{\mathbf{LFD}^{\equiv}} \,  \varphi \leftrightarrow  \bigvee \{T_{\mathcal{E}}(\varphi)\wedge
\chi_{\mathcal{E}}:
{\mathcal{E}} \textit{~is an equivalence relation on}~ V\}$

\vspace{2mm}

On the other hand, we also have, for every equivalence relation $\mathcal{E}\subseteq V\times V$:

\vspace{2mm}

\noindent\textbf{Claim 2:}\; $T_{\mathcal{E}} (\varphi)$ is satisfiable iff $T_{\mathcal{E}} (\varphi)\wedge \chi_{\mathcal{E}}$ is
satisfiable.

\vspace{2mm}

The right-to-left implication is obvious. For the left-to-right direction, assume that $T_{\mathcal{E}} (\varphi)$ is satisfiable. By a known
property of $\LFD$, it must then be satisfiable in some state $s$ of some `distinguished' dependence model $\bM$: one in which the range of values of every two distinct variables occurring in  $T_{\mathcal{E}} (\varphi)$ are mutually disjoint \citep{AJ-Dependence}. Change the values of all the variables $v\in V$, by assigning them the same value as the first chosen representative, i.e., $\bm{v}(s):= \bm{v'} (s)$, where $v'$ is the first variable $v'$ (according to our enumeration) s.t. $(v, v')\in \mathcal{E}$.  Call $\bM'$ the modified model. The definition of $T_{\mathcal{E}} (\varphi)$ ensures that all the variables actually occurring in $T_{\mathcal{E}} (\varphi)$ kept their old values. By a known property of $\LFD$ (its `Locality'), the change will not affect the truth value of $T_{\mathcal{E}} (\varphi)$, which thus remains true at $s$ in the changed model $\bM'$. Furthermore, it is easy to see our change makes the formula $\chi_{\mathcal{E}}$ true at all states in the resulting model $\bM'$. Putting these observations together, the formula $T_{\mathcal{E}}(\varphi) \wedge \chi_{\mathcal{E}}$ is satisfied at $s$ in $\bM'$.

\vspace{2mm}

We gather now all these `translations' into one formula as follows: 
\begin{center}
$T(\varphi) \, :=\, \bigvee \{ T_{\mathcal{E}} (\varphi): \mathcal{E}
\textit{~is an equivalence relation on~} V\}$.    
\end{center}
Using Claims 1 and 2 above, we immediately obtain:

\vspace{2mm}

\noindent\textbf{Claim 3:}\; $\varphi$ is satisfiable iff $T(\varphi)$ is satisfiable.

\vspace{2mm}

We have reduced the satisfiability problem for $\LFD^{\equiv}$ to the corresponding problem for $\LFD$, thus proving its decidability.

\vspace{2mm}

As for completeness, suppose that $\varphi$ is consistent w.r.t.  ${\mathbf{LFD}}^{\equiv}$. By Claim 1, there must exist some
equivalence relation $\mathcal{E}\subseteq V\times V$ s.t. $T_{\mathcal{E}}(\varphi)\wedge  \chi_{\mathcal{E}}$ is consistent w.r.t.  ${\mathbf{LFD}}^{\equiv}$, and hence $T_{\mathcal{E}}(\varphi)$ is also consistent w.r.t. $\mathbf{LFD}$ (since its
axioms and rules are included among those of ${\mathbf{LFD}}^{\equiv}$). By the completeness of   $\mathbf{LFD}$, 
$T_{\mathcal{E}}(\varphi)$ is satisfiable. Using the definition of $T(\varphi)$, it follows that $T(\varphi)$ is satisfiable, and so by Claim 3 we conclude that $\varphi$ is satisfiable, as desired.
\end{proof}

\begin{fact}\label{prop:decidability-lfd-equiv-funct}
The satisfiability problem for $\LFD^{\mathsf{f}, \equiv}$ is decidable. Moreover, the proof system $\bf{LFD^{f, \equiv}}$ is complete for
this logic.
Finally, if $T$ is any set of $\LFD^{\mathsf{f}, \equiv}$-terms closed under sub-terms, then  $\LFD^{\mathsf{f},\equiv}_\mathsf{T}$ with
function terms restricted to $T$ is completely axiomatized by $\bf{LFD^{f, \equiv}_T}$, which is obtained by restricting all the axioms and
derivation rules of $\bf{LFD^{f,
\equiv}}$ to instances that use only terms in $T$.
\end{fact}
\begin{proof}
    Here is the key idea, using our earlier results. The proof of decidability for $\LFD^{\mathsf{f}, \equiv}$ goes via reduction to $\LFD^{\equiv}$, and it uses the decidability of
$\LFD^{\equiv}$ together with the same translation method  used for the reduction of $\LFD^{\mathsf{f}}$ to $\LFD$: complex functional
terms are recursively replaced by new variables. The proof of completeness uses Fact \ref{fact:decidability-lfd-equiv}, following the same
pattern as the proof of Fact \ref{fact:completeness-lfd} on the completeness of fragments of $\bf{LFD^f}$.
\end{proof}

\par\noindent\textbf{Decidability of $\LDTV^{\mathsf{f},\equiv}$}\; We now proceed to use these facts to establish our main results on $\LDTV^{\mathsf{f},\equiv}$, starting with its decidability:

\begin{theorem}\label{theorem:decidable-LTD-equiv-funct}
The satisfiability problem for $\LDTV^{\mathsf{f},\equiv}$ is decidable. 
\end{theorem}

\begin{proof}
The proof proceeds via a translation into $\LFD^{\mathsf{f}, \equiv}$, following the same lines as the proof of the corresponding results for $\LDTV$. We  go from the fragment of the language of $\LDTV^{\mathsf{f},\equiv}$  to the language of
$\LFD^{\mathsf{f},\equiv}$. This is done by extending the translation  $\textbf{Tr}(\varphi)$ from Section
\ref{sec:o-term-logic-decidability} with two additional inductive clauses: for functional terms we add
\begin{center}
$\textbf{Tr}(\tom^m f(x_1, \ldots, x_n):= f(\textbf{Tr}(\tom^m x_1), \ldots, \textbf{Tr}(\tom^m x_n))$;
\end{center}
\noindent for formulas we add
\begin{center}
$\textbf{Tr}(x\equiv y):= (\textbf{Tr}(x)\equiv \textbf{Tr}(y))$.
\end{center}

\noindent We can then go back and forth between  dynamical models with functional symbols and dependence models with functional symbols, by
extending the maps $\bM \mapsto \bM^{\downarrow}$ and $\mathcal{M}\mapsto \mathcal{M}^+$ to cover  functional symbols. Here are some  hints:
\begin{itemize}
    \item[$\bullet$] For the model
$\bM^{\downarrow}$, we extend the interpretation function $I'$ for functional symbols as follows:
\begin{center}
$I'(f)((o_1,i_1),\dots, (o_n,i_n))= (I(f)(o_1,\dots,o_n), \mathsf{min}\{i_1,\dots,i_n\})$,
\end{center}
where $f$ is a functional symbol of arity $n$, $i_1,\ldots,i_n\in\mathbb{N}$ and $o_1,\dots,o_n\in \mathcal{O}$. 
Based on the new construction, we can extend the proofs of Fact \ref{fact:timed-LTD-model-transition-1} and Fact \ref{fact:LTD-t-lfd-1} to cover functional terms and identity.  

   \item[$\bullet$] For the converse, we start with a dynamical dependence model with functions $\mathcal{M}$, and modify the definition of  $\mathcal{M}^+$ to cover these functions, putting $$I^+(f)(o_1, \ldots, o_m):=I(f)(o_1, \ldots,o_m)$$
whenever the
later is defined, and $$I^+(f)(o_1, \ldots, o_m):=\infty$$ otherwise, then again extend the proof of Fact \ref{fact:LTD-t-lfd-2} to cover the
language with functional terms and identity. 
\end{itemize}
 This establishes decidability.
\end{proof}

\par\noindent\textbf{Proof system $\mathbf{LDTV^{f,\equiv}}$}\; Next, we present a proof system  $\mathbf{LDTV^{f,\equiv}}$  for the logic, which consists of the axioms and rules of the proof system
$\mathbf{LDTV}$ (restated for all formulas and terms of the extended logic $\LDTV^{\mathsf{f},\equiv}$) together with the new
axioms given in Table \ref{table:aximotization-LTD-equality}. The following principles are provable in the system:
\begin{description}
\item[Symmetry of Identity:] $x\equiv y \to y\equiv x$
\item[Functional Substitution:] $\bigwedge_{1\le i\le n} (x_i\equiv y_i)\to f(x_1, \ldots, x_n) \equiv f(y_1,\ldots, y_n)$
\item[Global Identity:] $x\equiv y \to \D_X (x\equiv y)$, \,\,\,\, \,  $x\equiv y \to \tom (x\equiv y)$
\end{description}

\begin{table}
    \caption{New axioms for $\mathbf{LDTV^{f,\equiv}}$}
    \label{table:aximotization-LTD-equality}
\centering
\begin{tabular}{r l}
\cline{1-2}
Functionality Axiom  &   $D_{\{x_1, \ldots, x_n\}} f(x_1, \ldots, x_n)$, where $f\in Funct$ is a\\ &functional symbol with $ar(f)=n$.
\vspace{1mm}\\
Reflexivity of Identity &   $x\equiv x$ \vspace{1mm}\\
Substitution of Identicals  &   $x\equiv y \to (\varphi \to \psi)$, where formula $\psi$ is obtained \\
 & by  substituting some occurrence of $x$ in $\varphi$ by $y$.
\vspace{1mm}\\
Term Reduction   &   $\tom f(x_1,\ldots,x_n)\equiv f(\tom x_1,\ldots,\tom x_n)$\\
\cline{1-2}
\end{tabular}
\end{table}

Again, the soundness of ${\bf LDTV^{f,\equiv}}$ is easy to see:

\begin{fact}\label{facf:timed-LTD-ident-soundness}
${\bf LDTV^{f,\equiv}}$ is sound w.r.t.  dynamical dependence models with functions.
\end{fact}

Also, we have the following  result for its completeness:
\begin{theorem}\label{theorem:completeness-LTD-equiv-funct}
$\bf{LDTV^{f,\equiv}}$ is complete for dynamical  models with functions.
\end{theorem}

\begin{proof}
The completeness proof uses the completeness result in Fact \ref{prop:decidability-lfd-equiv-funct}, following the same steps as the
proof of completeness for $\LDTV$ (Theorem \ref{theorem:completeness-LTD-t}).
\end{proof}

\section{Adding a next-time modality}\label{sec:next-time-modality}
Going beyond the dependence logics studied so far featuring temporalized variables, let us now consider standard temporal reasoning about a dynamical process making explicit assertions $\tom \varphi$ about the next stage after the current one. Is this a true expressive extension, or could it be a redundant device just for convenience? Here is a `reductive' intuition that seems plausible: one may think that
 \begin{quote}
\qquad\quad   {\em future truth about values of variables is the same as 

\qquad\quad current truth about future values of these variables}.  
\end{quote}   

\noindent This is in fact how one can read the  recursion equations of a transition function. 

To analyze this intuition precisely, we first add $\tom$ as a {\em temporal modality},  writing $\mathcal{L}^{\dagger,f,\equiv}$ {\em for the resulting language obtained by adding to $\mathcal{L}^{f, \equiv}$  formulas of the form $\tom \varphi$}, which are interpreted as follows:
\begin{center}
    $s\vDash_{\bM} \tom\varphi$\quad {\it iff} \quad $g(s)\vDash_{\bM} \varphi$.
\end{center}
Now,  the intuition can be characterized by the following equivalences:
\begin{center}
\begin{tabular}{lcl}
Atomic-Reduction  & \qquad  &  $\tom P(x_1, \ldots, x_k)  \leftrightarrow P(\tom x_1, \ldots, \tom x_k)$ \vspace{1mm} \\

Next-Time$_1$  & \qquad  &   $\tom \D_X\varphi\leftrightarrow \D_{\tom X}\tom\varphi $ \vspace{1mm} \\

Next-Time$_2$  &\qquad  &  $\tom D_Xy\leftrightarrow D_{\tom X}\tom y$ \\
\end{tabular}
\end{center}

\vspace{-1mm}

However, although Atomic-Reduction is a validity of $\LDTV$ and $\LDTV^{\mathsf{f},\equiv}$,  the reduction formulas Next-Time$_1$ and Next-Time$_2$ are not. More precisely,  only  the following weaker implications are valid:
\begin{center}
\begin{tabular}{lcl}
w-Next-Time$_1$  & \qquad  &   $\tom \D_X\varphi\to\D_{\tom X}\tom\varphi $ \vspace{1mm} \\

w-Next-Time$_2$  &\qquad  &  $\tom D_Xy\to D_{\tom X}\tom y$ \\
\end{tabular}
\end{center}
 \vspace{-1mm}
 
Here are counterexamples to the converses.

\begin{fact}\label{fact:next-time-not-valid}
Neither Next-Time$_1$ nor Next-Time$_2$ is valid in $\LDTV$. 
\end{fact}

\begin{proof}
Let us consider instances $D_{\tom x} \tom y\to \tom D_xy$ and $D_{\tom x} \tom Py\to \tom D_xPy$. We define a model
$\bM=(\mathbb{D}_v,I,S,g,{\bm v})_{v\in V}$ such that:
\begin{itemize}
    \item $\mathbb{D}_x=\mathbb{D}_y=\{0,1\}$,
    \item $I(P)=\{1\}$,
    \item   $S=\{s_1,s_2,s_3\}$,
    \item $g(s_1)=g(s_2)=g(s_3)=s_3$,
    \item ${\bm x}(s_1)=1$, ${\bm x}(s_2)=0$, ${\bm x}(s_3)=1$, \\
    ${\bm y}(s_1)=0$, ${\bm y}(s_2)=1$, and ${\bm y}(s_3)=1$.
\end{itemize}
By construction,  
$s_1\not\vDash D_{\tom x} \tom y\to \tom D_xy$ and $s_1\not\vDash D_{\tom x} \tom Py\to \tom D_xPy$.
\end{proof}

To see just when the intuition does apply, we need to see what temporal properties of dynamical systems are required by the full Next-Time$_1$ and Next-Time$_2$. This can be made clear  via the following modal-style correspondence analysis which involves frame truth under all valuations for the `dynamical frames' underlying our dynamical models:

\begin{fact}\label{fact:correspondence-dependence-atoms}
For each dynamical frame $\bM=(\mathbb{D}_v,I,S,g,{\bm v})_{v\in V}$, Next-Time$_1$ and Next-Time$_2$ are valid (true everywhere under all valuations) iff for any $s,u\in S$, $g(s)=_Xu$ implies that there is some $w\in S$ with $g(w)=u$ and $s=_{\tom X} w$.
\end{fact}

\begin{proof}
Let us consider for Next-Time$_2$. The reasoning for Next-Time$_1$  is similar. 

\vspace{1mm}

(1). First, assume that for any $s,u\in S$, $g(s)=_Xu$ implies that there is some $w\in S$ with $g(w)=u$ and $s=_{\tom X} w$. Let $\bM$ be an arbitrary model over the frame and $s$  a state such that $s\vDash D_{\tom X} \tom y$. Also, let $u$ be a state such that $g(s)=_Xu$. By assumption, there is some $w\in S$ with $g(w)=u$ and $s=_{\tom X} w$. Then, we get that $s=_{\tom y} w$, which gives us $g(s)=_yu$ immediately.

\vspace{1mm}

(2). Next, consider for the other direction. Assume that there are two states $s,u\in S$ such that  $g(s)=_Xu$ and that there is no $w\in S$ with $g(w)=u$ and $s=_{\tom X} w$. Now, we can define a model  in which $\tom y$ has the same value at all those states $t$ such that $s=_{\tom X} t$, but $y$ has different values at $g(s)$ and $u$. Clearly, this makes $D_{\tom X} \tom y\to \tom D_Xy$ false at $s$ in the model.
\end{proof}

\begin{remark}\label{remark:correspondence}
The analysis in the proof above essentially uses quantification over relations $=_X$, which is not in monadic second-order logic (as is the case with standard modal correspondence). Also, we interpret $=_{\tom X}$ concretely as $=_X$, i.e., $s=_{\tom X} t$ iff $g(s)=_{X} g(t)$. But in addition to this interpretation, one can also consider various further abstractions, including reading $=_{\tom X}$ as an arbitrary binary relation on states, or considering  dynamical systems in which dynamic transitions can be binary relations other than functions. It would be interesting to explore how standard modal correspondence theory generalizes to such settings. E.g., the above axioms still have intuitive Sahlqvist form, but we now need to do the correspondence analysis in the presence of relational atoms like $D_Xy$.\footnote{For a modal perspective on our logics, cf. Section \ref{sec:modal-approach}.}
\end{remark}

In what follows, we first study the above logic over dynamical models that satisfy the  restriction matching the intuition (Section \ref{sec:next-time-modality-axiomatization-timed-semantics}), and then generalize to a setting without such restrictions (Section \ref{sec:next-time-modality-non-timed-semantics}). For both kinds of temporal dependence logic, we find a complete  Hilbert-style proof system and show its decidability.

\section{Logic of temporal dependence with timed semantics}\label{sec:next-time-modality-axiomatization-timed-semantics}

Dynamical systems in models so far   were abstract state spaces with a transition function. In this section, we move a bit closer to a temporal view of the executions of unfoldings of a dynamical system, and present a `timed semantics', under which the truth of the intuition spelled out in Section \ref{sec:next-time-modality} is ensured. We will prove that the resulting logic is decidable and completely axiomatizable. The method for showing this, though reductive, is  different from the reductions used in the preceding sections, based on reduction axioms as used in dynamic-epistemic logics \citep{BMS,Johan-del,hans-del}, plus the axiomatization of the logic $\LFD^{f}$ with function terms over objects found in \cite{AJ-Dependence}. The system obtained in this way may be said to encode the laws governing `synchronous dynamic dependence'.

\subsection{Basics of timed semantics}\label{sec:basics-timed-semantics}
We first introduce some basic semantic notions for the new setting:

\begin{definition}\label{def:timed-system}
A dynamical system $\bS=(S,g)$ is \emph{timed} if there is a map $\tau:S\to \mathbb{N}\cup\{\infty\}$, associating to each state $s\in S$ a
finite or infinite `time' $\tau_s$, satisfying  the following two conditions:
\begin{itemize}
\item[{\bf (a).}] $g(s)=0$ for every initial state $s$ such that $s\not=g(w)$ for any $w\in S$, and
\item[{\bf (b).}] $\tau_{g(s)}=\tau_s+1$.
\end{itemize}
(Here, we use the convention that $\infty+1=\infty$.)  Any map $\tau$ satisfying these conditions is called a \emph{timing map}. The relation
$=^{\tau}$, defined on $S$ by putting:
\begin{center}
    $s=^{\tau}w$ iff $\tau_s=\tau_w$
\end{center}
is called the \emph{synchronicity relation}.
\end{definition}

It should be noted that \emph{not every dynamical system is timed}. To be precise, we define the following notion of `$g$-history' of states
in dynamical systems:

\vspace{1.5mm}

\noindent{\bf $g$-history}\; Given a dynamical system, a \emph{$g$-history} of a state $s$ is a finite or infinite backward-transition chain
$(s_0 = s, s_1, s_2,\ldots)$, with $s_n = g(s_{n+1})$ for all $n <m$, where $m$ is the total number of states in the chain (called the
\emph{length} of the history). A $g$-history is \emph{maximal} if it is infinite or it cannot be extended to the right to a $g$-history of
greater length.

\vspace{1.5mm}

\noindent Essentially, timed dynamical systems are exactly the ones in which all maximal $g$-histories of any state have the same (finite or infinite) length. Fact \ref{facf:timed-system} below collects some properties of these notions. Now, to get a better feel for the features of timed dynamical systems, we  provide some examples and counterexamples:

\begin{example}\label{example:timed-dynamical-systems}
Figure \ref{fig:timed-dynamical-system} presents some timed dynamical systems, while Figure \ref{fig:untimed-dynamical-system} shows some
dynamical systems that are not timed.

In $\bS_1$, we have $\tau_{w_0}=\tau_{k_0}=0$, $\tau_{w_1}=\tau_{k_1}=1$, and $\tau_{w_2}=2$. Different from $\bS_1$ in which all
maximal $g$-history of a state are finite, for all states $u$ in $\bS_2$, we have $\tau_{u}=\infty$. It is instructive to notice that the
cycle in the system can only be an `end' (i.e., states in the cycle do not have $g$-successors that are not in the cycle), as otherwise   dynamical transitions would not be deterministic. But the length of a cycle in a dynamical system may be
bigger than 1: say,  $\bS_3$ contains a cycle of length $4$. Additionally, any disjoint union of timed dynamical systems are
still timed, and so the dynamical system $\bS_4$ is: it can be treated as a disjoint union of the system $\bS_2$ and a linear structure that
is timed obviously. Moreover, the system $\bS_4$ shows that both the maximal $g$-histories and the `$g$-future' of a state may contain
infinitely many different states.

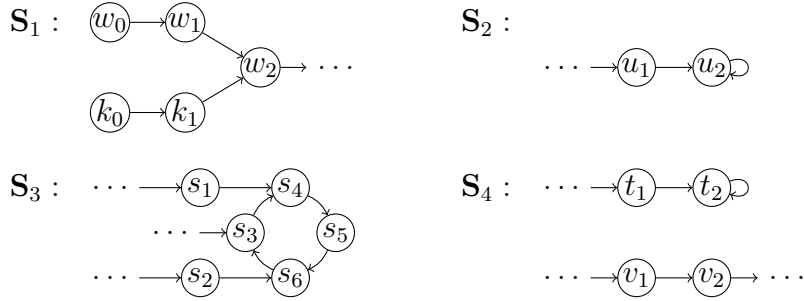
\begin{figure}
\centering
\begin{tikzpicture}
\node(w00)[] at (-1,0) {${\bf S}_1:$};
\node(w0)[circle,draw,inner sep=0pt,minimum size=5mm] at (0,0) {$w_0$};
\node(w1)[circle,draw,inner sep=0pt,minimum size=5mm] at (1,0) {$w_1$};
\node(w2)[circle,draw,inner sep=0pt,minimum size=5mm] at (2,-.6) {$w_2$};
\node(w3)[] at (3,-.6){$\ldots$};
\draw[->](w0) to (w1);
\draw[->](w1) to (w2);
\draw[->](w2) to (w3);

\node(w01)[circle,draw,inner sep=0pt,minimum size=5mm] at (0,-1.2) {$k_0$};
\node(w11)[circle,draw,inner sep=0pt,minimum size=5mm] at (1,-1.2) {$k_1$};
\draw[->](w01) to (w11);
\draw[->](w11) to (w2);

\node(u00)[] at (5,0) {${\bf S}_2:$};
\node(u0)[] at (6,-.6) {$\ldots$};
\node(u1)[circle,draw,inner sep=0pt,minimum size=5mm] at (7,-.6) {$u_1$};
\node(u2)[circle,draw,inner sep=0pt,minimum size=5mm] at (8,-.6) {$u_2$};
\draw[->](u0) to (u1);
\draw[->](u1) to (u2);
\draw[->](u2) to  [in=-20, out=20,looseness=5] (u2);

\node(s)[] at (-1,-2.2) {${\bf S}_3:$};
\node(s0)[] at (0,-2.2) {$\ldots$};
\node(s00)[] at (0,-3.4) {$\ldots$};
\node(s1)[circle,draw,inner sep=0pt,minimum size=5mm] at (1.2,-2.2) {$s_1$};
\node(s2)[circle,draw,inner sep=0pt,minimum size=5mm] at (1.2,-3.4) {$s_2$};
\node(s4)[circle,draw,inner sep=0pt,minimum size=5mm] at (2.4,-2.2) {$s_4$};
\node(s3)[circle,draw,inner sep=0pt,minimum size=5mm] at (1.8,-2.8) {$s_3$};
\node(s7)[] at (0.8,-2.8) {$\ldots$};
\node(s6)[circle,draw,inner sep=0pt,minimum size=5mm] at (2.4,-3.4) {$s_6$};
\node(s5)[circle,draw,inner sep=0pt,minimum size=5mm] at (3,-2.8) {$s_5$};
\draw[->](s0) to (s1);
\draw[->](s00) to (s2);
\draw[->](s1) to (s4);
\draw[->](s2) to (s6);
\draw[->](s7) to (s3);
\draw[->](s3) to [bend left=20] (s4);
\draw[->](s4) to [bend left=20] (s5);
\draw[->](s5) to [bend left=20] (s6);
\draw[->](s6) to [bend left=20] (s3);

\node(t00)[] at (5,-2.2) {${\bf S}_4:$};
\node(t0)[] at (6,-2.2) {$\ldots$};
\node(t1)[circle,draw,inner sep=0pt,minimum size=5mm] at (7,-2.2) {$t_1$};
\node(t2)[circle,draw,inner sep=0pt,minimum size=5mm] at (8,-2.2) {$t_2$};
\draw[->](t0) to (t1);
\draw[->](t1) to (t2);
\draw[->](t2) to [in=-20, out=20,looseness=5] (t2);

\node(v0)[] at (6,-3.4) {$\ldots$};
\node(v1)[circle,draw,inner sep=0pt,minimum size=5mm] at (7,-3.4) {$v_1$};
\node(v2)[circle,draw,inner sep=0pt,minimum size=5mm] at (8,-3.4) {$v_2$};
\node(v3) at (9,-3.4) {$\ldots$};
\draw[->](v0) to (v1);
\draw[->](v1) to (v2);
\draw[->](v2) to (v3);
\end{tikzpicture}
\caption{Timed dynamical systems, with arrows for dynamical transitions and ellipses for states.}
\label{fig:timed-dynamical-system}
\end{figure}

Although $\bS'_1$ and $\bS'_2$
look similar to $\bS_1$ and $\bS_2$ respectively, neither of them is timed: each of them
contains states having maximal $g$-histories of unequal lengths. For instance, state $i_3$ in $\bS'_1$ has two maximal $g$-histories of
lengths $2$ and $3$, and state $j_1$ in    $\bS'_2$ has $g$-maximal histories of, e.g., $1$ and $\infty$.

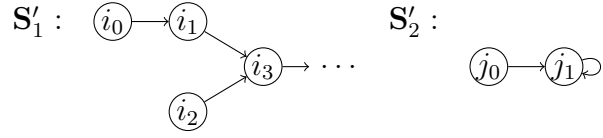
\begin{figure}
\centering
\begin{tikzpicture}
\node(t00)[] at (-1,0){${\bf S}'_1:$};
\node(t0)[circle,draw,inner sep=0pt,minimum size=5mm] at (0,0) {$i_0$};
\node(t1)[circle,draw,inner sep=0pt,minimum size=5mm] at (1,0) {$i_1$};
\node(t2)[circle,draw,inner sep=0pt,minimum size=5mm] at (2,-.6) {$i_3$};
\draw[->](t0) to (t1);
\draw[->](t1) to (t2);

\node(t11)[circle,draw,inner sep=0pt,minimum size=5mm] at (1,-1.2) {$i_2$};
\node(t31)[] at (3,-.6){$\ldots$};
\draw[->](t11) to (t2);
\draw[->](t2) to (t31);

\node(s00)[] at (4,0) {${\bf S}'_2:$};
\node(s0)[circle,draw,inner sep=0pt,minimum size=5mm] at (5,-.6) {$j_0$};
\node(s1)[circle,draw,inner sep=0pt,minimum size=5mm] at (6,-.6) {$j_1$};
\draw[->](s0) to (s1);
\draw[->](s1) to  [in=-20, out=20,looseness=5] (s1);
\end{tikzpicture}
\caption{Some dynamical systems that are not timed.}
\label{fig:untimed-dynamical-system}
\end{figure}
\end{example}

Let us explore more features of timed dynamical systems. Clearly, the synchronicity relation $=^{\tau}$ is an equivalence relation and it satisfies the following two \emph{synchronicity conditions}:
\begin{itemize}
\item[(1).]  $s=^{\tau}w$ iff $g(s)=^{\tau}g(w)$
\item[(2).] If $s=^{\tau}g(w)$, then $s=g(w')$ for some $w'\in S$.
 \end{itemize}

 Also, whenever a dynamical system is timed, the timing map $\tau$ is uniquely determined by the dynamical transition function:

 \begin{fact}\label{fact:timed-system-timing-map}
 If $\bS=(S,g)$ is a timed system, then the unique map $\tau$ satisfying the above timing conditions
(a) and (b) is given by:
\begin{center}
 $\tau_s:={\rm{supremum}}\{n\in\mathbb{N}: g^n(w)=s \;\textit{for some}\; w\in S\}.$
\end{center}
 \end{fact}
 
  \begin{proof}
 Let $\bS$ be a timed dynamical system. Then, it is not hard to see that $\tau$ described above satisfies the condition (a). Also, for any $s,w\in\bS$ and $n\in \mathbb{N}\cup\{\infty\}$, $g^n(w)=s$ iff $g^{n+1}(w)=g(s)$. So, by the definition of $\tau$, we have $\tau_s+1=\tau_{g(s)}$, i.e., the function satisfies the condition (b). So, $\tau$ is a timing map.

 Next, suppose  there is another timing map $\tau'$ different from $\tau$, i.e., $\tau_s\not=\tau'_s$ for some $s\in\bS$. From condition (a), we know that $s$ cannot be an initial state, as otherwise $\tau_s=\tau'_s=0$. Without loss of generality, assume $\tau_s>\tau'_s$. Then $\tau'_s$ must be in $\mathbb{N}$, which implies that $s=g^n(s_0)$ for some initial state $s_0$. Since $\tau$ satisfies (b) and $\bS$ is timed, it holds that $\tau_s=n$, a contradiction.
 \end{proof}

Although in any dynamical system the map $\tau_s$, defined as Fact \ref{fact:timed-system-timing-map}, always satisfies the first timing
condition (a), it may fail to satisfy the condition (b).

It is thus important to characterize timed systems directly in terms of the dynamical transition structure, as well as in terms of the synchronicity relation:
 
 \begin{fact}\label{facf:timed-system}
For every dynamical system $\bS=(S,g)$, the following are equivalent:
\begin{itemize}
\item[(1).] $\bS$ is timed.
\item[(2).] For any $n\in\mathbb{N}\cup\{\infty\}$, the predecessors of any $(n+1)$-successor are $n$-successors: if $g(s)=g^{n+1}(w)$, then $s=g^{n}(w')$ for some $w'\in S$.
\item[(3).] All maximal $g$-histories of the same state have the same length.
\item[(4).] There is an equivalence relation satisfying the two synchronicity conditions.
\end{itemize}
\end{fact}

\begin{proof}
    See \ref{appendix:proof-for-facf:timed-system}.
\end{proof}

Timed dynamical systems are of interest in their own right, and they include   model classes used  in computer science such as linear-time temporal models.

 \vspace{1.5mm}

\noindent{\bf Temporal dynamical systems}\;  A dynamical system $\bS= (S, g)$ is \emph{temporal} if every state has at most one predecessor: $g(s) = g(w)$ implies $s = w$. This is equivalent to requiring that every state has a unique maximal $g$-history (`unique past'). It is easy to see that temporal systems are timed, but the converse is false: none of the timed dynamical systems in Figure \ref{fig:timed-dynamical-system} is temporal. A typical kind of temporal dynamical systems are finite cycles, e.g., the system $\bS_5$ in Figure \ref{fig:temporal-dynamical-system}, and another kind of temporal dynamical systems are lines with either finite or infinite past, and infinite future. Again, temporal dynamical systems are closed under disjoint union (e.g., $\bS_6$ in Figure \ref{fig:temporal-dynamical-system}).

\begin{figure}
    \centering
\begin{tikzpicture}
\node(t00)[] at (-.5,0) [label=left:${\bf S}_6:$]{};
\node(t0)[circle,draw,inner sep=0pt,minimum size=5mm]  at (0,0) []{$s_0$};
\node(t1)[circle,draw,inner sep=0pt,minimum size=5mm] at (1,0) []{$s_1$};
\node(t2)[circle,draw,inner sep=0pt,minimum size=5mm] at (2,0) []{$s_2$};
\node(t3) at (3,0) {$\ldots$};
\draw[->](t0) to (t1);
\draw[->](t1) to (t2);
\draw[->](t2) to (t3);

\node(s0) at (0,-1) {$\ldots$};
\node(s1)[circle,draw,inner sep=0pt,minimum size=5mm] at (1,-1) []{$t_0$};
\node(s2)[circle,draw,inner sep=0pt,minimum size=5mm] at (2,-1) []{$t_1$};
\node(s3) at (3,-1) {$\ldots$};
\draw[->](s0) to (s1);
\draw[->](s1) to (s2);
\draw[->](s2) to (s3);

\node(u00)[] at (-5.5,0) [label=left:${\bf S}_5:$]{};
\node(u0)[circle,draw,inner sep=0pt,minimum size=5mm]  at (-5,0) []{$u_0$};
\node(u3)[circle,draw,inner sep=0pt,minimum size=5mm] at (-5,-1) []{$u_3$};
\node(u1)[circle,draw,inner sep=0pt,minimum size=5mm] at (-3.5,0) []{$u_1$};
\node(u2)[circle,draw,inner sep=0pt,minimum size=5mm] at (-3.5,-1) []{$u_2$};
\draw[->](u0) to (u1);
\draw[->](u1) to (u2);
\draw[->](u2) to (u3);
\draw[->](u3) to (u0);
\end{tikzpicture}
 \caption{Temporal dynamical systems.}
 \label{fig:temporal-dynamical-system}
\end{figure}
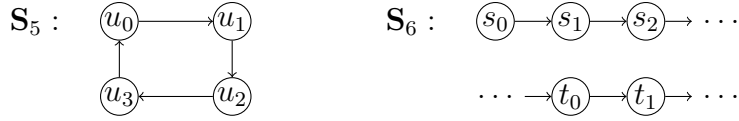

\medskip

\noindent{\bf Linear-time systems} \; A \emph{linear-time system} is a  temporal dynamical system $\bS= (S, g)$ that is `acyclic': i.e.,
$g^n(s)\neq s$ for all $s\in S$ and $n\geq 1$. Essentially, linear-time systems are disjoint unions of infinite-future `lines' (i.e., infinite
upward chains with or without a first point). We can restrict even further, looking at \emph{linear-time systems with finite past}: these are
defined as linear-time systems $\bS= (S, g)$ with the property that every state has a finite history: for every
$s\in S$ there exists some $s_0\in S-g(S)$ s.t. $s=g^n(s_0)$ for some natural number $n\geq 0$, where 
$g(S):=\{t\in S: \textit{~there is some~} s\in S \textit{~such that~} t=g(s) \}$.

\medskip

We thus obtained a descending series of ever-more-restricted forms of dynamical systems: from the general ones to the special case of timed systems, then to the even more special case of temporal systems, then to the subclass of linear-time systems, and finally to the most special case of linear-time systems with finite past. We end with a sketch of a classification of these dynamical systems:

 \begin{remark}\label{remark:comments-on-pictures}
Every dynamical system can be viewed as just a function on a set. Even so, the resulting patterns can be  complicated. However, the special classes considered here effect drastic reductions in complexity. Timed systems are disjoint unions of `lines' (i.e., infinite upward chains with or without a first point), finite cycles and figures obtained by appending to any of the points of a finite cycle one or more infinite linear past-chains. Temporal systems are even fewer: being just disjoint unions of such lines and finite cycles. Finally, linear-time systems are merely disjoint unions of lines (with infinite future, with or without a first point).  
\end{remark}

 \subsection{Axiomatization and decidability}\label{sec:dfd-timed-semantics}
 
A dynamical model is \emph{timed} (or \emph{temporal}, \emph{linear-time}, or has \emph{finite past}) if the underlying dynamical system is
timed (or respectively, temporal, linear-time, etc).

\vspace{1.5mm}

\noindent{\bf Timed semantics}\; Given a timed dynamical model with timing map $\tau$ and corresponding synchronicity relation $=^{\tau}$, the \emph{timed semantics} is the same as before, except that the interpretations for $\D_X\varphi$ and $D_Xy$ are defined using \emph{synchronized value agreement} $=_X^{\tau}$ instead of $=_X$ in the semantic clauses for $\D_X\varphi$ and $D_Xy$; while synchronized value agreement is given by putting:
\begin{center}
$s=_X^{\tau}w$ \; iff \; both ${\bm X}(s)={\bm X}(w)$ and $s=^{\tau} w$.
\end{center}
Therefore, relation $=^{\tau}_X$ states that the values of $X$ given by the \emph{synchronous} states are the same. The result of endowing the
language $\mathcal{L}^{\dag}$ with this synchronous semantics will be called \emph{the logic of temporal dependence with timed semantics}, and will be denoted by
$\LTD^{\mathsf{t},\mathsf{f},\equiv}$. 

Note that the \emph{meaning} of our operators is subtly different in the timed semantics from the old dynamic semantics  in earlier sections. For instance, $D_Xy$ now means  that the value of $y$ is uniquely determined by the current value(s) of $X$ \emph{plus the current time}: i.e., all states
synchronous with the current state that agree with it on the value(s) of $X$ also agree on the value of $y$. In what follows, we explore a complete proof system for this logic and show that its satisfiability problem is decidable. We start with the following results for the logic {\em $\LDTV^{\mathsf{t},\mathsf{f},\equiv}$ employing the language $\mathcal{L}^{\mathsf{f},\equiv}$ and the timed semantics}:

\begin{fact}\label{fact:satisfiability-timed-without-O-modality}
The satisfiability problems for $\LDTV^{\mathsf{t},\mathsf{f},\equiv}$ on timed models, temporal models, linear-time models, and linear-time models with finite
past  are all equivalent, and moreover they are decidable.
\end{fact}
\begin{proof}
See \ref{appendix:proof-for-fact:satisfiability-timed-without-O-modality}.
\end{proof}

\begin{fact}\label{fact:completeness-timed-without-O-modality}
$\bf{LDTV^{f,\equiv}}$ is a complete calculus for  $\LDTV^{\mathsf{t},\mathsf{f},\equiv}$ on timed models, temporal models, linear-time models, and linear-time models with finite
past.
\end{fact}

\begin{proof}
    See  \ref{appendix:proof-for-fact-completeness-timed-without-O-modality}.
\end{proof}

Now let us return to $\LTD^{\mathsf{t}, \mathsf{f},\equiv}$ and introduce a calculus for it:

\medskip

\par\noindent\textbf{Proof system}\; A proof system  $\mathbf{LTD^{t,f,\equiv}}$ can be obtained by  extending $\mathbf{LDTV^{f,\equiv}}$ with the axioms and rule in Table \ref{table:axiomatization-LTD-o-modality}. But notice that the rule Substitution$_{\tom}$ displayed in Table \ref{table:axiomatization-LTD} can be dropped, since it is derivable in the new setting.  Here, {\em $\bigcirc$-Distribution} shows that $\bigcirc$ is a  normal modality, and {\em Functionality}  ensures that dynamical transitions between states are functional.  
The axioms {\em Atomic-Reduction}, {\em Next-Time$_1$} and {\em Next-Time$_2$} illustrate the interactions between dynamical transitions and the  truth involving atoms, dependence quantifiers and dependence formulas $D_Xy$ respectively. It is simple to check the soundness of the system:

\begin{fact}\label{fact:soundness-timed-modality}
$\mathbf{LTD^{t,f,\equiv}}$ is sound w.r.t. timed dynamical models (and hence also w.r.t. any of the above subclasses: temporal models, linear-time models etc).
\end{fact}

\begin{table}
    \caption{New axioms and rule for  $\mathbf{LTD^{t,f,\equiv}}$}
    \label{table:axiomatization-LTD-o-modality}
    \centering
\begin{tabular}{lll}
\cline{1-3}
\textbf{Axioms and Rule for $\bigcirc$} &\qquad&\qquad \vspace{1mm} \\
\cline{1-3}
$\bigcirc$-Distribution &\qquad&\qquad $\bigcirc(\varphi\to\psi)\to(\bigcirc\varphi\to\bigcirc\psi)$\vspace{1mm} \\
Functionality&\qquad&\qquad $\bigcirc\neg\varphi\leftrightarrow\neg\bigcirc\varphi$\vspace{1mm} \\
$\tom$-Necessitation &\qquad&\qquad  From $\varphi$, infer $\tom\varphi$\vspace{1mm} \\
\cline{1-3}
\textbf{Interaction Axioms}&\qquad&\qquad \vspace{1mm}\\
\cline{1-3}
Atomic-Reduction &\qquad& \qquad $\tom P(x_1, \ldots, x_k) \leftrightarrow P(\tom x_1, \ldots, \tom x_k)$ \vspace{1mm} \\
Next-Time$_1$  &\qquad & \qquad $\tom \D_X\varphi\leftrightarrow \D_{\tom X}\tom\varphi $ \vspace{1mm} \\
Next-Time$_2$  &\qquad & \qquad
$\tom D_Xy\leftrightarrow D_{\tom X}\tom y$\vspace{1mm} \\
\cline{1-3}
\end{tabular}
\end{table}

Next, we show  the completeness of the calculus: 

\begin{theorem}\label{theorem:complete-timed-modality}
The system $\mathbf{LTD^{t,f,\equiv}}$ is complete w.r.t. the timed semantics on each of the above classes: timed dynamical models, temporal models, linear-time models, and linear-time models with finite past.
\end{theorem}

\begin{proof} 
 We prove this by showing that each $\mathcal{L}^{\dagger,f,\equiv}$-formula can be reduced to an equivalent $\mathcal{L}^{f,\equiv}$-formula under the timed semantics. More precisely, we define a map $\mathfrak{T}:\mathcal{L}^{\dagger,f,\equiv}\to\mathcal{L}^{f,\equiv}$  that keeps $\px$ and $D_Xy$ the same, permutes with Boolean connectives and modalities $\D_X$, and sets
 \begin{center}
  $\mathfrak{T}(\tom \varphi): = \mathfrak{T}(\varphi)[V/\tom V]$.   
 \end{center}
It now suffices to show that the following are provable in $\mathbf{LTD^{t,f,\equiv}}$:
\begin{itemize}
 \item[(a)] $\vdash_{\mathbf{LTD^{t,f,\equiv}}}\tom\varphi\leftrightarrow \varphi[V/\tom V]$
\item[(b)] $\vdash_{\mathbf{LTD^{t,f,\equiv}}}\varphi\leftrightarrow \mathfrak{T}(\varphi)$
\end{itemize}
which together with the completeness of $\mathbf{LDTV^{f,\equiv}}$ w.r.t. the timed dynamical models (Fact \ref{fact:completeness-timed-without-O-modality}) can give us the completeness of $\mathbf{LTD^{t,f,\equiv}}$.

\vspace{1.5mm}

(1). We first prove part (a), by induction on formulas $\varphi\in\mathcal{L}^{\dagger,f,\equiv}$. The cases for $\px$ and $D_Xy$ are given by the basic axioms
Atomic-Reduction and Next-Time$_2$ respectively. We now consider the other cases.

\vspace{1mm}

(1.1).  $\varphi$ is $\neg\psi$. By the inductive hypothesis,  
 $\vdash_{\mathbf{LTD^{t,f,\equiv}}}\tom\psi\leftrightarrow \psi[V/\tom V]$.   
 So, 
$\vdash_{\mathbf{LTD^{t,f,\equiv}}}\neg\tom\psi\leftrightarrow \neg\psi[V/\tom V]$.    
Using the axiom of Functionality, we have 
 $\vdash_{\mathbf{LTD^{t,f,\equiv}}}\tom\neg\psi\leftrightarrow \neg\psi[V/\tom V]$.   
Notice that  $\neg\psi[V/\tom V]$ is exactly $\varphi [V/\tom V]$.

\vspace{1mm}

(1.2). Formula $\varphi$ is $\psi_1\land \psi_2$. Clearly,
  $\vdash_{\mathbf{LTD^{t,f,\equiv}}}\tom(\psi_1\land \psi_2)\leftrightarrow
\tom\psi_1\land \tom \psi_2$.  
By the inductive hypothesis, for each $i\in\{1,2\}$,  
$\vdash_{\mathbf{LTD^{t,f,\equiv}}}\tom\psi_i\leftrightarrow \psi_i[V/\tom V]$.  
Also,  $\varphi[V/\tom V]$ is the same as  $(\psi_1\land \psi_2)[V/\tom V]$. Hence,
$\vdash_{\mathbf{LTD^{t,f,\equiv}}}\tom\varphi\leftrightarrow \varphi[V/\tom V]$.

(1.3). When $\varphi$ is $\D_X\psi$, using the axiom Next-Time$_1$ and the inductive hypothesis, one  proves the equivalence directly.

\vspace{1mm}

(1.4). $\varphi$ is $\tom\psi$. By the inductive hypothesis,  
  $\vdash_{\mathbf{LTD^{t,f,\equiv}}}\varphi\leftrightarrow \psi[V/\tom V]$,  
which can give us 
 $\vdash_{\mathbf{LTD^{t,f,\equiv}}}\tom\varphi\leftrightarrow \tom\psi[V/\tom V]$.   

\vspace{1.5mm}

(2). Part (b)  follows immediately from part (a): we just need to use the fact that
 $\vdash_{\mathbf{LTD^{t,f,\equiv}}}\psi\leftrightarrow\chi$ implies
$\vdash_{\mathbf{LTD^{t,f,\equiv}}}\varphi[\psi/\chi]$,    
where $\varphi[\psi/\chi]$ results from substituting all occurrences of $\psi$ in $\varphi$ with
the formula $\chi$, and  applying (a) repeatedly to make all subformulas $\tom\psi$ of $\varphi$ disappear.  
\end{proof}

The proof above also indicates that we can  reduce the satisfiability problem for $\LTD^{\mathsf{t},\mathsf{f},\equiv}$ to that for $\LDTV^{\mathsf{t},\mathsf{f},\equiv}$, and since the later is decidable (see Fact \ref{fact:satisfiability-timed-without-O-modality}), it follows straightforwardly that: 

\begin{theorem}\label{theorem:decidable-timed-modality}
The satisfiability problems of $\mathcal{L}^{\dagger,f,\equiv}$-formulas on timed models, temporal models, linear-time models, and linear-time models with finite
past (all considered with the timed semantics) are all equivalent, and  decidable.
\end{theorem}

\section{A general logic of temporal dependence}\label{sec:next-time-modality-non-timed-semantics}  

In this final part of our exploration of temporal logics for dynamic dependence, we return to the general setting without  timing restrictions, and study the combined logic with temporalized variables and  a temporal next-time modality over  the class of all dynamical dependence models. 

\subsection{Introducing the system $\LTD$}
Merely for simplicity of exposition, we will not add  function symbols and global equality in what follows, but just consider the language {\em $\mathcal{L}^{\dagger}$ that extends the language $\mathcal{L}$ of $\LDTV$ with formulas of the form $\tom \varphi$}. The resulting {\em Logic of Temporal Dependence}, called $\LTD$, has a semantics   given by the same recursive clauses already stated for the earlier logic $\LDTV$ (but applied to the richer language of $\LTD$), together with the standard semantic clause for a next-time modality. Let us first provide a  Hilbert-style calculus $\mathbf{LTD}$ for the logic.

\vspace{2mm}

\noindent{\bf Proof system $\mathbf{LTD}$}\; The proof system consists of the following: 
\begin{itemize}
    \item Axioms and rules in Table \ref{table:axiomatization-LTD}.
    \item Axioms and rule for $\tom$ and Atomic-Reduction in Table \ref{table:axiomatization-LTD-o-modality}.
    \item $\tom\D_X \varphi\to\D_{\tom X}\tom\varphi$ (w-Next-Time$_1$) and\\
$\tom D_Xy\to D_{\tom X}\tom y$ (w-Next-Time$_2$).
\end{itemize}

So, instead of Next-Time$_1$ and Next-Time$_2$ that are contained in $\mathbf{LTD^{t,f,\equiv}}$, $\mathbf{LTD}$ includes their weaker versions. Below are some derivable principles: 
 \begin{itemize}
 \item[$\bullet$] $\vdash_{\mathbf{LTD}} D_X\tom ^nY\land\bigcirc^n\D_Y\varphi\to \D_X\bigcirc^n\varphi$ \hfill (Dyn-Transfer)
\item[$\bullet$] $\vdash_{\mathbf{LTD}} D_X\tom ^nY\land\bigcirc^nD_Y\tom^mZ\to D_X\tom^{m+n}Z$ \hfill (Dyn-Trans)
\item[$\bullet$] $\vdash_{\mathbf{LTD}} \bigcirc\D_{V}\varphi\to \D_{V}\bigcirc\varphi$\hfill ($\bigcirc$-$\D_{V}$-Commutation)
\end{itemize}


As we shall see, the details of the proofs for the completeness of $\mathbf{LTD}$ and the decidability of $\LTD$ will very different from those for  $\LTD^{\mathsf{t},\mathsf{f},\equiv}$, since the behaviors of the two systems are quite different: for instance, unlike the axioms Next-Time$_1$ and Next-Time$_2$ for $\LTD^{\mathsf{t},\mathsf{f},\equiv}$, the new axioms w-Next-Time$_1$ and w-Next-Time$_2$ do not give us recursion axioms. To establish the desired results then, in what follows we first provide an abstract modal perspective on the logic, which is equivalent to the  semantics w.r.t. dynamical dependence models described above (Section \ref{sec:modal-approach}), in that they capture the same class of validities in  $\LTD$. Subsequently, we generalize the modal approach to a `looser' setting (Section \ref{sec:general-relational-models}), which will be useful in showing the completeness of $\mathbf{LTD}$  w.r.t. both the classes of modal models and  dynamical dependence models (Section \ref{sec:completeness-DFD}). Finally, with the help of those semantic proposals for $\LTD$, we will show  decidability  in Section \ref{sec:decidability-DFD-non-empty}.

\subsection{Changing to an equivalent modal semantics}\label{sec:modal-approach}
Now we switch to a slightly more abstract modal perspective on $\LTD$ and its models, which will be our main vehicle in what follows. 

\begin{definition}\label{def-model} A \emph{standard relational model} is a tuple $\bM=(W, g, \sim_v, \|\bullet\|)_{v\in V}$  such that
\begin{itemize}
\item[$\bullet$] $W$ is a non-empty set of abstract states or `possible worlds'.
\item[$\bullet$] For each basic variable $v\in V$, $\sim_v\subseteq W\times W$ is an equivalence relation. \\ We extend this notation to \emph{terms} $x$, by putting recursively
$$s\sim_{\tom x} w \, \, \mbox{ iff } \, \, g(s)\sim_x g(w),$$
and finally we extend it to non-empty sets of terms $X$ by taking intersections:
$$s\sim_X w  \, \, \mbox{ iff } \, \, s\sim_x w \mbox{ for all $x\in X$.}$$
So, $\sim_V$ is the intersection $\bigcap_{v\in V}\sim_v$.
\item[$\bullet$] $g: W\to W$ preserves $\sim_V$, i.e., $s\sim_{V} w$ implies $g(s)\sim_{V}g(w)$.
\item[$\bullet$] $\|\bullet\|$ is a valuation map from atoms $P(x_1,\ldots, x_n)$ to the power set $\mathcal{P}(W)$ of $W$ such that  whenever $s \sim_X w$ and $s \in \|P(x_1,\ldots, x_n)\|$ for some $x_1,\ldots,x_n\in X$, then $w \in \|P(x_1,\ldots, x_n)\|$.
\end{itemize}
\noindent A pair of a model and a world $(\bM, s)$ (for short, $\bM, s$) is called a \emph{pointed model}.
\end{definition}

In a standard relational model, states $s,t$ are called `$X$-equivalent' if
$s\sim_Xt$. Now let us turn to our semantics on standard relational models.

\begin{definition}\label{def-semantics-standard-relational-model}
Given a pointed model $(\bM,s)$ and a formula $\varphi\in\mathcal{L}^{\dag}$, the following recursion defines when $\varphi$ \emph{is true in} $\bM$
\emph{at} $s$, written $s\vDash_{\bM}\varphi$ (where we suppress the truth clauses for operators that read exactly as the standard interpretation):
\begin{center}
\begin{tabular}{rcl}
$s\vDash P(x_1,\ldots, x_n)$   &\quad iff  &\quad  $s\in \|P(x_1,\ldots, x_n)\|$\\
 $s\vDash D_Xy$  &\quad  iff &\quad  for each $t\in W$, $s\sim_Xt$ implies $s\sim_y t$ \\
$s\vDash \D_X\varphi$     &\quad  iff &\quad  for each $t\in W$, $s\sim_Xt$ implies $t\vDash\varphi$
\end{tabular}
\end{center}
\end{definition}

We now relate the modal semantics and the semantics w.r.t. dynamical dependence models by two model transformations.

\begin{definition}\label{def:dependenc-to-relational}
For each $\bM=(\mathbb{D}_v,I,S,g,\bv)_{v\in V}$, the \emph{induced standard relational model} is $\bM'=(W,G,\sim_v,\|\bullet\|)_{v\in V}$ with $W:=S$, $G:=g$, and
\begin{itemize}
\item[$\bullet$] for all $s,t\in W$ and $v\in V$, we put: $s\sim_vt$ iff $s=_vt$;
\item[$\bullet$] for each $\px$, we put $\| \px \|:=\{s\in S: s\vDash_{\bM} \px \}$.
\end{itemize}
\end{definition}

Here the function $G$ and the valuation $\|\bullet\|$ satisfy the two special conditions imposed in Definition \ref{def-model}. Now,
a simple induction on formulas $\varphi$ suffices to show that the semantics w.r.t. dynamical models agrees with the modal semantics:

\begin{fact}\label{proposition:dependence-to-relational}
For each dynamical model $\bM$ and $\LTD$-formula $\varphi$,
\begin{equation*}
s\vDash_{\bM}\varphi\Leftrightarrow s\vDash_{\bM'}\varphi \,\,\,\,\,\,\,\, \, \mbox{ (for all states $s\in \bS$)}.
\end{equation*}
\end{fact}

Here is the, less obvious, transformation in the opposite direction.

\begin{definition}\label{def:relational-to-dependence} For each standard relational model $\bM=(W,g,\sim_v,\|\bullet\|)_{v\in V}$ the \emph{induced
dynamical model} $\bM^\sim=(\mathbb{D}_v,I,S,G,\bv)_{v\in V}$, where:
\begin{itemize}
\item[$\bullet$] For each $v\in V$, we put $\mathbb{D}_v:=\{\sim_v(s): s\in W\}$, where $\sim_v(s)=\{t\in W: s\sim_v t\}$ is the equivalence class of state $s$ modulo $\sim_v$;
\item[$\bullet$] For each $n$-ary predicate symbol $P$, $I(P):=\{(\sim_{x_1}(s),\ldots,\sim_{x_n}(s)): s\in \|P(x_1,\ldots, x_n)\| \}$.
\item[$\bullet$] $S:=\{\sim_V(s): s\in W\}$, where $\sim_V(s)=\{t\in W: s\sim_V t\}$ is the equivalence class of state $s$ modulo the $V$-equivalence relation $\sim_V=\bigcap_{v\in V}\sim_v$ defined in Section \ref{sec:modal-approach};
\item[$\bullet$] $G(\sim_V(s)):= \sim_V(g(s))$;
\item[$\bullet$] For each $v\in V$, ${\bm v}(\sim_V(s)):=\sim_v(s)$.
\end{itemize}
\end{definition}

It is easy to check that these are well-defined (independent on the choice of representatives for a given equivalence class), and that the
resulting model does satisfy the restriction imposed on dynamical dependence models. Moreover, the modal semantics agrees with the dynamical
model semantics:

\begin{fact}\label{proposition:relational-to-dependence}
For each relational $\bM=(W,g,\sim_v,\|\bullet\|)_{v\in V}$ and formula $\varphi$, we have:
\begin{equation*}
s\vDash_{\bM}\varphi\, \, \Leftrightarrow \,\, \sim_V(s)\vDash_{\bM^\sim}\varphi \,\,\,\,\,\,\,\, \, \mbox{ (for all states $s\in W$)}.
\end{equation*}
\end{fact}

\begin{proof}
The proof is by induction on the formula $\varphi$. The cases for atoms and Boolean connectives are routine. The equivalence for atoms $D_Xy$
holds by the semantics and the following fact:
\begin{center}
For all $s,t\in W$ and terms $X$,  $s\sim_X t$ in $\bM$ \, iff \,$\sim_V(s)=_X \sim_V(t)$ in $\bM^\sim$.
\end{center}

\noindent  The inductive cases for $\D_X\varphi$ and $\bigcirc\varphi$ are straightforward.
\end{proof}

Facts \ref{proposition:dependence-to-relational} and \ref{proposition:relational-to-dependence} immediately imply a validity-reduction in both
ways:

\begin{fact}\label{proposition:equivalence-of-dependence-and-standard}
The same $\LTD$-formulas are valid on their dynamical models and on standard relational models. 
\end{fact}

Both perspectives on $\LTD$ are interesting, but we will mainly work with the modal view, which allows us to use notions such as \emph{generated submodels}, \emph{bisimulations} and \emph{$p$-morphisms}, and techniques such as \emph{unraveling} into tree form \citep{blackburn2004modal}.

\subsection{General relational models}\label{sec:general-relational-models}

Our eventual aim is to show that the system $\mathbf{LTD}$ is complete with respect to standard relational models (and hence also w.r.t. our
original dynamical models). To achieve this, we take several steps of separate interest:
\begin{description}
\item[Step 1.] We introduce a new notion of `general relational models' and interpret the formulas of our language in this broader setting.
\item[Step 2.] We prove  completeness of the system $\mathbf{LTD}$ w.r.t. the new models.
\item[Step 3.] We show a representation result for general relational models as $p$-morphic images of standard relational models, which implies that $\mathbf{LTD}$ is also complete w.r.t. standard relational models.
\end{description}

This subsection is concerned with Step 1.

\begin{definition}\label{def-generalrelationalmodel} A \emph{general relational model} is a tuple $\bM=(W,g,=_X,\|\bullet\|)_{X}$ with the following
four components:
\begin{itemize}
\item[$\bullet$] $W$ is a non-empty set of abstract states or `possible worlds'.
\item[$\bullet$] $g: W\to W$ is a function.
\item[$\bullet$] For each non-empty finite set $X$ of terms, $=_X\subseteq W\times W$ is a binary relation.
\item[$\bullet$] $\|\bullet\|$ is an `extended' valuation, mapping `atoms'  $P(x_1, \ldots, x_n)$ or $D_Xy$ to
    $\mathcal{P}(W)$.
    For simplicity, we will typically write $s\vDash P(x_1, \ldots, x_n)$ instead of $s\in \|P(x_1, \ldots, x_n)\|$, and similarly write
    $s\vDash D_Xy$
    instead of $s\in \|D_Xy\|$.
\end{itemize}

\noindent These ingredients are required to satisfy the following conditions:
\begin{itemize}
\item[\emph{C1}.] All $=_X$ are equivalence relations on $W$.
\item[\emph{C2}.] All $\|D_Xy\|$ satisfy the following three properties:
\begin{itemize}
\item[$\bullet$] Dep-Reflexivity: For $x\in X$, we have $s\vDash D_Xx$.
 \item[$\bullet$] Dep-Transitivity: If $s\vDash D_XY$ and $s\vDash D_YZ$, then $s\vDash D_XZ$.
 \item[$\bullet$] Determinism: For all terms $x$, $s\vDash D_{V} \tom x$.
\end{itemize}
\item[\emph{C3}.] If $s=_X w$ and $s\vDash D_XY$, then $w\vDash D_XY$ and $s =_Y w$.
\item[\emph{C4}.] If $s=_X w$ and $s\vDash P(x_1, \ldots, x_n)$ with $x_1, \ldots, x_n\in X$, then $w \vDash P(x_1, \ldots, x_n)$.
\item[\emph{C5}.] $s\vDash P(\tom x_1,\ldots,\tom x_n)$ iff $g(s)\vDash P(x_1,\ldots,x_n)$.
\item[\emph{C6}.] If $s=_{\tom X} w$, then $g(s)=_X g(w)$.
\item[\emph{C7}.] If $g(s)\vDash D_XY $, then $s\vDash D_{\tom X} \tom Y$.
\end{itemize}
\end{definition}

It is useful to compare this new notion with the standard relational models of Definition \ref{def-model}. Each $=_X$ for a non-empty finite $X\subseteq \term$
is now a primitive relation, not necessarily the intersection of the individual $=_{x\in X}$. Also, formulas $D_Xy$ are treated as atoms now,
with truth values given directly by valuation functions. 

The  truth definition for $\LTD$-formulas in general relational models reads exactly as that in Definition
\ref{def-semantics-standard-relational-model} for standard relational models, though with the new understanding of relations and atoms as just
explained.

\begin{fact}\label{fact:sound-DFD-general-relational-model}
The proof system $\mathbf{LTD}$ is sound for  general relational models.
\end{fact}

\begin{proof}
The key reasons for the validity of the non-trivial principles are as follows.
As $g$ is a function, the axiom Functionality is valid. 
Condition C1 ensures the validity of $\D$-T,$\D$-4 and $\D$-5. Condition C2 gives us Dep-Ref, Dep-Trans and Determinism. Condition C3 ensures both $\D$-Introduction$_2$ and Transfer.
Conditions C4 and C5 ensure $\D$-Introduction$_1$ and Atomic-Reduction respectively. Finally, C6 and C7 ensure the validity of
w-Next-Time$_1$ and w-Next-Time$_2$ respectively.
\end{proof}

\subsection{Completeness of the system $\mathbf{LTD}$}\label{sec:completeness-DFD} 

We now come to Step 2 and Step 3 mentioned in Section \ref{sec:general-relational-models}. In what follows, we first show that $\mathbf{LTD}$ is complete w.r.t.
general relational models, which is crucial to establish the completeness $\mathbf{LTD}$ w.r.t. standard relational models. 

\begin{definition}\label{def-canonicalmodel} The \emph{canonical model for $\mathbf{LTD}$} is $\bM^c = (W^c,g^c,=^c_X,\|\bullet\|^c)$, where
\begin{itemize}
\item[$\bullet$] $W^c$ is the class of all maximal $\mathbf{LTD}$-consistent sets.
\item[$\bullet$] For all $s\in W^c$, $g^c(s)=\{\varphi: \bigcirc\varphi\in s\}$.
\item[$\bullet$] For all $s, t\in W^c$ and terms $X$, $s=^c_Xt$ iff $\D_Xs\subseteq t$.
\item[$\bullet$] $s\in\|D_Xy\|^c$ iff $D_Xy\in s$, and $s\in\|\px\|^c$ iff $\px\in s$.
\end{itemize}
For all states $s\in W^c$, $\D_Xs$ denotes the set of formulas $\{\varphi: \D_X\varphi\in s\}$.
\end{definition}

The $g^c$ above defines a function
on $W^c$, which follows from the following: 

\begin{fact}\label{fact:canonical-function} 
In the canonical model  $\bM^c=(W^c,g^c,=^c_X,\|\bullet\|^c)$, $g^c(s)\in W^c$.
\end{fact}

\begin{proof}
Since the Functionality axiom of $\mathbf{LTD}$ has syntactic Sahlqvist form, the canonical model will satisfy its corresponding semantic frame condition of functionality by a standard argument about maximally consistent sets \citep{blackburn2004modal}.
\end{proof}

So, the similarity type of the model introduced here fits.  It remains to check the conditions on general relational models listed in Definition \ref{def-generalrelationalmodel}, and we will prove that the calculus $\mathbf{LTD}$ is complete for general relational models.

\begin{fact}\label{fact:canonical-general-model}
The canonical model for $\mathbf{LTD}$ is a general relational model of $\LTD$.
\end{fact}

\begin{proof}
    See \ref{appendix:proof-for-fact:canonical-general-model}.
\end{proof}

Next, a standard argument proves the following Existence Lemma:

\begin{lemma}\label{lemma-existence}
Let $\bM^c$ be the canonical model and $s\in W^c$. Then we have:
\begin{center}
If $\widehat{\D}_X\varphi\in s$, then there exists $t\in W^c$ such that $s=^c_Xt$ and $\varphi\in t$.
\end{center}
\end{lemma}

Now we are able to prove the following key \emph{Truth Lemma}:

\begin{lemma}\label{lemma-truth}
Let $\bM^c$ be the canonical model, $s\in W^c$ and $\varphi\in\mathcal{L}$. Then
\begin{equation*}
s\vDash_{\bM^c}\varphi\Leftrightarrow\varphi\in s.
\end{equation*}
\end{lemma}

\begin{proof}
The proof is by induction on $\varphi$. We only show two cases.

\vspace{1mm}

(1). Formula $\varphi$ is $\bigcirc\psi$. Then, $s\vDash_{\bM^c}\varphi$ iff $g^c(s)\vDash_{\bM^c}\psi$. Then, by the inductive hypothesis,
$g^c(s)\vDash_{\bM^c}\psi$ iff $\psi\in g^c(s)$. From the definition of $g^c$, we know that $\psi\in g^c(s)$ iff $\varphi\in s$.

\vspace{1mm}

(2). Formula $\varphi$ is $\widehat{\D}_X\psi$. From left to right, assume that $s\vDash_{\bM^c}\widehat{\D}_X\psi$. Then, there exists $t\in
W^c$ with $s=^c_Xt$ and $t\vDash_{\bM^c}\psi$. By the inductive hypothesis, $\psi\in t$. From the definition of $=^c_X$, we have
$\widehat{\D}_X\psi\in s$. Conversely, suppose that $\widehat{\D}_X\psi\in s$. Then, by Lemma \ref{lemma-existence}, there is a $t\in W^c$
with $s=^c_Xt$ and $\psi\in t$. By the inductive hypothesis, $t\vDash_{\bM^c}\psi$. So, $s\vDash_{\bM^c}\widehat{\D}_X\psi$.
\end{proof}

Immediately, we have the following:

\begin{theorem}\label{theorem-completeness-general}
The proof system $\mathbf{LTD}$ is complete for  general relational models.
\end{theorem}

Now Step 2 has been completed. The remainder of the part is concerned with Step 3, for which we introduce the following  Representation theorem:

\begin{theorem}\label{theorem:Repr}
Every general relational model of $\LTD$ is a $p$-morphic image of some standard relational model of $\LTD$. 
\end{theorem}
\begin{proof}
    See \ref{appendix:proof-for-theorem:Repr}.
\end{proof}



As truth of modal formulas is preserved under surjective $p$-morphisms \citep{blackburn2004modal}, and for  $\LTD$, standard relational models are general relational models, it follows that the same formulas of the logic are valid on its general relational models and on its standard relational models. Combining this with the earlier representation results of Section \ref{sec:modal-approach}, we have shown the completeness:

\begin{theorem}\label{theorem-completeness-standard}
The proof system $\mathbf{LTD}$ is sound and complete w.r.t. both standard relational models and dynamical models. 
\end{theorem}

\subsection{Decidability of logic $\LTD$}\label{sec:decidability-DFD-non-empty}

In this part, we consider the decidability of  $\LTD$, using filtration techniques that work for  $\mathsf{LFD}$  \cite{AJ-Dependence} plus ideas from \cite{handbook-dtl} on the topic of dynamic topological logics. More precisely, we are going to prove that the logic has the finite model property w.r.t. its general relational models.

Before introducing  details of the proof, let us first describe briefly what will be going on.  We first generalize the notion of  temporal depth for terms given in Definition \ref{def:temporal-depth-terms-LTD} into formulas, which is used to measure the nesting depth of the operator $\tom$. Then, the notion will be used as a parameter to construct a finite \emph{closure} of a given finite set of formulas. Given a formula $\varphi$ that is satisfied by some general relational model $\bM$ (that might be infinite), using both the notions of closure and temporal depth we define the \emph{$i$-type} of a state in the model, which consists of all formulas with temporal depth no bigger than $i$ that belong to the closure of $\varphi$ and are true at the state. Finally, all those $i$-types of $\bM$ will be used as states to construct a new general relational model $\bM^\dag$ satisfying $\varphi$. As the closure of $\varphi$ is finite, $\bM^\dag$ is finite as well. The method is inspired by \cite{handbook-dtl} that proved the decidability for dynamical topological logics, but is much simpler.\footnote{The techniques  in this section may help  simplify the decidability proof in the cited paper.}  

\begin{definition}\label{def:temporal-depth}
We generalize the notion  of temporal depth on terms and finite sets of terms, given in Definition \ref{def:temporal-depth-terms-LTD}, to formulas:
\begin{center}
$\td(P(x_1,\ldots, x_n))={\rm{max}}\{\td(x_1),\ldots,\td(x_n)\}\qquad $\\\vspace{1mm}
$\td(\neg \varphi)=\td( \varphi)\qquad \td( \varphi\land \psi)={\rm{max}}\{\td( \varphi),\td( \psi)\}\qquad \td(\bigcirc
\varphi)=\td(\varphi)+1$\\\vspace{1mm}
$\td(\D_X\varphi)={\rm{max}}\{\td( \varphi),\td(X)\}\qquad \td(D_Xy)={\rm{max}}\{\td(x): x\in X\cup\{y\}\}$
\end{center}
Also, for a set $\Phi$ of formulas, $\td(\Phi)={\rm{max}}\{\td( \varphi): \varphi\in\Phi\}$.
\end{definition}

One can check that $\td(D_{\tom X}\tom y)=\td(\tom D_Xy)$, $\td(\tom \D_X\varphi)=\td(\D_{\tom X}\tom \varphi)$ and $\td(P(\tom x_1,\ldots,\tom x_n))=\td (\tom P(x_1,\ldots,x_n))$.

\vspace{2mm}

\noindent{\bf Generalized subformulas}\; For any $\varphi\in\mathcal{L}^{\dagger}$, we denote by $\varphi^{-\tom}$ \emph{the formula
resulting from removing all occurrences of $\tom$ from $\varphi$}. We say $\varphi$ is a `\emph{generalized subformula}' of $\psi$ if $\varphi^{-\tom}$ is a subformula of $\psi^{-\tom}$ (in the usual sense).

 \vspace{2mm}

Moreover, given a finite set $\Phi$ of formulas with $\td(\Phi)=k$, we employ $\mathbb{V}_{\Phi}$ for the set of variables occurring in
$\Phi$, and $\mathbb{T}_{\Phi}$ for the set of terms $\tom^n v$ such that $v\in\mathbb{V}_{\Phi}$ and $n\le k$. Notice that
$\mathbb{V}_{\Phi}\subseteq\mathbb{T}_{\Phi}$ and both of them are finite. We now proceed to introduce the following:

\begin{definition}\label{def:closed-set}
Let $\Phi\subseteq\mathcal{L}^{\dagger}$ be  finite and $\td(\Phi)=k$. We say $\Phi$ is \emph{closed}, if
\begin{itemize}
\item[\emph{P1}.] For all non-empty $X, Y\subseteq\mathbb{T}_{\Phi}$ and $y\in \mathbb{T}_{\Phi}$, $\D_YD_Xy\in \Phi$.
\item[\emph{P2}.] If $\varphi\in\Phi$ is not of the form $\neg\psi$, then $\neg\varphi\in\Phi$.
\item[\emph{P3}.] If $\varphi$ is a generalized subformula of $\psi\in\Phi$ and $\td(\varphi)\le k$, then $\varphi\in \Phi$.
\item[\emph{P4}.] If $\varphi\in\Phi$ is $\tom\psi$ or $\px$, then for any non-empty $Y\subseteq\mathbb{T}_{\Phi}$, $\D_Y\varphi\in\Phi$.
\end{itemize}
For a set $\Psi$ of formulas, its \emph{closure} is the smallest closed set that contains $\Psi$.
\end{definition}

The closure of a finite set $\Psi$ is also finite, and to see this, it is useful to observe that: (1). every clause of P1-P4 can only give us finitely many formulas, and (2). the clauses do not have infinite interactions that produce infinitely many formulas.

Let $\Phi$ be a set of formulas s.t. $\td(\Phi)=k$. For each $i\le k+1$, we define its $i$-layer $\Phi_i:=\{\varphi\in\Phi:
\td(\varphi)<i\}$, consisting of formulas $\varphi\in \Phi$ with $\td(\varphi)< i$.

\begin{definition}\label{def:types}
Let $\bM=(W,g,=_X,\|\bullet\|)$ be a general relational model and $\Phi$ a closed set with $\td(\Phi)=k$. For each $i\le k+1$, the
\emph{$\Phi_{i}$-type} of a state $s\in W$ is defined as 
  $i$-type$(s):=\{\varphi\in \Phi_i: s\vDash \varphi\}$.
\end{definition}

Thus, a $i$-type of a state in a model is a maximal consistent subset of $\Phi_i$. Also, the empty set of formulas $\emptyset$ is the $0$-type
of any state.

In the remainder of this section, we shall work with a fixed general relational model $\bM$ and a finite closed set $\Phi$. We will construct
a finite general relational model $\bM^{\dag}=(W^{\dag},G,\approx_X,\|\bullet\|^{\dag})$ of $\LTD$ satisfying all
$\Phi_i$-types in $\bM$. First of all, definitions of $W^{\dag}$ and $G$ are simple:
\begin{itemize}
\item[$\bullet$] $W^{\dag}=\{\alpha: \alpha \;\textit{is some}\; i\text{-type}(s) \;\textit{with}\; s\in W\;\textit{and}\; i\le k+1\}$.
\item[$\bullet$] $G(\alpha)=\{\psi: \bigcirc\psi\in \alpha\}$.
\end{itemize}

One can check that $W^{\dag}$ is finite. For the transition function $G$, we have:

\begin{fact}\label{fact:function-well-defined}
For any $i$-type $\alpha$ of a state $s$, $G(\alpha)$ is the ${\rm max}\{0, i-1\}$-type $\beta$ of $G(s)$. So, $G(\alpha)\in W^{\dag}$.
\end{fact}

\begin{proof}
The case that $G(\alpha)=\emptyset$ is trivial, and we merely consider $G(\alpha)\not=\emptyset$. Also, it is simple to see that
$G(\alpha)\subseteq \beta$. For the other direction, assume $\psi\in\beta$. Then, in the original model $\bM$, we have $s\vDash \tom \psi$. As
$\td(\psi)\le i-1$, it is easy to see $\tom\psi \in \Phi_i$. Immediately, it then holds that $\tom\psi\in\alpha$, and therefore, $\psi\in
G(\alpha)$.
\end{proof}

Now, let us proceed to introduce the equivalence relations $\approx_X$ on $W^{\dag}$:

\begin{definition}\label{def:equivalence-relations-decidability}
For all non-empty sets $X\subseteq\term$ and all $\alpha,\beta\in W^{\dag}$, we write $\alpha\approx_X\beta$ if one of the following two
cases holds:
\begin{itemize}
\item[\emph{E1.}]  $\td(\alpha)=\td(\beta)\ge \td(X)$, and
\begin{itemize}
\item[\emph{E1.1}.] For all $D_XY$, $D_XY\in\alpha\Leftrightarrow D_XY\in\beta$, and
\item[\emph{E1.2}.] When $D_XY\in\alpha$ (or equivalently, $D_XY\in\beta$),  $\D_Y\varphi\in\alpha\Leftrightarrow \D_Y\varphi\in\beta$.
\end{itemize}
\item[\emph{E2.}] $\td(\alpha)=\td(\beta)< \td(X)$, and there is some $m\le {\rm{min}}(\{\td(\alpha)\}\cup \{\td(x): x\in X\})$ such that $\alpha\approx_{\tom^m\mathbb{V}_{\Phi}}\beta$ holds in the sense of \emph{E1}.
\end{itemize}
\end{definition}

So, for all non-empty $X\subseteq \term$, if $\alpha\approx_X\beta$, then $\alpha=\emptyset$ iff $\beta=\emptyset$. With this construction in
place, here is an observation on $\approx_X$:

\begin{fact}\label{fact:observation-on-relations-of-phi-model}
When $\alpha\approx_X\beta$ and $D_XY\in\alpha$, we have $\alpha\approx_Y\beta$.
\end{fact}

\begin{proof}
Assume that $\alpha\approx_X\beta$ and $D_XY\in\alpha$. Let $\td(\alpha)=i$. Then, we know that $\td(X)\le i$ and $\td(Y)\le i$.

Let $D_YZ\in\alpha$. Using P1, we cab obtain $\D_YD_YZ\in \Phi_i$. Then, by P2, it holds that $\neg \D_YD_YZ\in\Phi_i$. With Definition \ref{def:types} and
axiom $\D$-Introduction$_2$, we have $\D_YD_YZ\in\alpha$. Now, from E1.1, we know $\D_YD_YZ\in\beta$. So, $D_YZ\in\beta$.

Moreover, it is simple to see that $D_XZ\in \alpha\cap\beta$. Using E1.2, we have $\D_Z\varphi\in\alpha$ iff $\D_Z\varphi\in\beta$. Finally,
with the clauses of E1, we can check that $\alpha\approx_Y\beta$.
\end{proof}

The fact will be useful to simplify our proofs below. Now, let us continue to show that all relations $\approx_X$ are equivalence
relations:

\begin{fact}\label{fact:phi-model-c1}
Let $X$ be non-empty. The relation $\approx_X$ is an equivalence relation.
\end{fact}

\begin{proof}
It is easy to see that the relation is \emph{reflexive} and \emph{symmetric}. Now we prove that it is \emph{transitive}. Let
$\alpha\approx_X\beta$ and $\beta\approx_X\gamma$.
Then, $\td(\alpha)=\td(\beta)=\td(\gamma)$. The case that $\alpha=\emptyset$ is trivial, as it implies $\beta=\gamma=\emptyset$. We now
consider $\alpha\not=\emptyset$.

\vspace{1mm}

(1). First, consider $\td(\alpha)\ge\td(X)$. Then, for any $D_XY$, it is simple to see that  $D_XY\in\alpha$ iff $D_XY\in\beta$ iff $ D_XY\in\gamma$. Also, when $D_XY\in\alpha$, it is a matter of direct checking that $\D_Y\varphi\in \alpha$ iff $\D_Y\varphi\in \beta$ iff
$\D_Y\varphi\in \gamma$.

\vspace{1mm}

(2). Next, consider $\td(\alpha)<\td(X)$. Then, there are $m_1$ and $m_2$ s.t. $\alpha\approx_{\tom^{m_1}\mathbb{V}_{\Phi}}\beta$,
$\beta\approx_{\tom^{m_2}\mathbb{V}_{\Phi}}\gamma$, and for $i\in\{m_1,m_2\}$, $i\le {\rm{min}}(\{\td(\alpha)\}\cup\{\td(x): x\in X\})$.

When $m_1=m_2$, by the same reasoning as the case above, but now using $\tom^{m_1}\mathbb{V}_{\Phi}$ in place of $X$, we have
$\alpha\approx_{\tom^{m_1}\mathbb{V}_{\Phi}}\gamma$. Thus, it still holds that $\alpha\approx_X\gamma$.

Now, let us consider $m_1\not=m_2$. Without loss of generality, we assume that $m_1< m_2$, and it suffices to show
$\alpha=_{\tom^{m_2}\mathbb{V}_{\Phi}}\beta$: to see this, one just need to notice that (a).
$D_{\tom^{m_1}\mathbb{V}_{\Phi}}\tom^{m_2}\mathbb{V}_{\Phi}\in \alpha$, (b). $\alpha=_{\tom^{m_1}\mathbb{V}_{\Phi}}\beta$, and (c). Fact
\ref{fact:observation-on-relations-of-phi-model}.
\end{proof}

To complete the definition of $\bM^{\dag}$, it remains to define $\|\bullet\|^{\dag}$:

\begin{definition}\label{def:valuation-decidability}
For all $\px$, we put $\|\px\|^{\dag}:=\{\alpha: \px\in\alpha\}$. For each $D_Xy$, we define $\|D_Xy\|^{\dag}$ as the \emph{the smallest
subset of $W^{\dag}$} satisfying the following:
\begin{itemize}
\item[\emph{V1}.] If $\tom^m\mathbb{V}_{\Phi}\subseteq X$ for some $m\le \td(y)$, then $\alpha\in \|D_Xy\|^{\dag}$ for all $\alpha\in W^{\dag}$.

\item[\emph{V2}.] For all $y\in X$ and $\alpha\in W^{\dag}$, $\alpha\vDash D_Xy$.

\item[\emph{V3}.] If $D_{X'}y\in\alpha$ for some $X'\subseteq X$, then $\alpha\vDash D_{X}y$.

\item[\emph{V4}.] If $D_{X'}\tom^m{\mathbb{V}_{\Phi}}\in\alpha$ for some $X'\subseteq X$ and $m\le \td(y)$, then $\alpha\vDash D_{X}y$.
\end{itemize}
\end{definition}

By construction, for all $D_Xy$ with $\td(D_Xy)\le \td(\alpha)$, it holds that
\begin{center}
 $\alpha\in \|D_Xy\|^{\dag}$ iff $D_Xy\in\alpha$.    
\end{center}

\begin{fact}\label{fact:phi-model-c2}
The valuation $\|\bullet\|^{\dag}$ satisfies `Dep-Reflexivity', `Dep-Transitivity' and `Determinism'.
\end{fact}

\begin{proof}
    See \ref{appendix:proof-for-fact:phi-model-c2}.
\end{proof}

Also, for $\|\px\|^{\dag}$, we have the following:

\begin{fact}\label{fact:phi-model-c5}
For all $\alpha\in W^{\dag}$, $\alpha\vDash P(\tom x_1,\ldots,\tom x_n)$ iff $G(\alpha)\vDash P(x_1,\ldots,x_n)$.
\end{fact}

\begin{proof}
Let $\alpha\in W^{\dag}$ such that $\td(\alpha)=i$.

For the direction \emph{from left to right}, from $\alpha\vDash P(\tom x_1,\ldots,\tom x_n)$ we know that $P(\tom x_1,\ldots,\tom x_n)\in\alpha$. So, $P(\tom x_1,\ldots,\tom x_n) \in \Phi_i$. From P3, it is not hard to see that $\tom P(x_1,\ldots,x_n)\in\Phi_i$. Then, by $P(\tom x_1,\ldots,\tom x_n)\in\alpha$, it holds that $\tom P(x_1,\ldots,x_n)\in\alpha$. So, $P(x_1,\ldots,x_n)\in G(\alpha)$. Thus, $G(\alpha)\vDash P(x_1,\ldots,x_n)$.

\vspace{1mm}

For the direction from \emph{right to left}, from $G(\alpha)\vDash P(x_1,\ldots,x_n)$ we know that $\tom P(x_1,\ldots,x_n)\in \alpha$. Then, it is easy to see that $P(\tom x_1,\ldots,\tom x_n)\in \Phi_i$. Now, using the fact that $\tom P(x_1,\ldots,x_n)\in \alpha$, we have $P(\tom x_1,\ldots,\tom x_n)\in \alpha$, and so it holds that $\alpha\vDash P(\tom x_1,\ldots,\tom x_n)$.
\end{proof}

Now, it is important to show that $\bM^{\dag}$ is a general relational model for $\LTD$, i.e., its components satisfy the
 conditions C1-C7 in Definition \ref{def-generalrelationalmodel}.
We have already proven some of them, and let
us move to the others.

\begin{fact}\label{fact:phi-model-c3}
For any $\alpha,\beta\in W^{\dag}$, if $\alpha\approx_X\beta$ and $\alpha\vDash D_XY$, then it holds that $\alpha\approx_Y\beta$ and $\beta\vDash D_XY$.
\end{fact}

\begin{proof}
    See \ref{appendix:proof-for-fact:phi-model-c3}.
\end{proof}

\begin{fact}\label{fact:phi-model-c4}
For all $\alpha,\beta\in W^{\dag}$, if $\alpha\approx_X \beta$ and $\alpha\vDash\px$ (where $X$ is the set of terms occurring in the sequence
${\bm x}$), then $\beta\vDash\px$.
\end{fact}

\begin{proof}
Assume that $\alpha\approx_X \beta$ and $\alpha\vDash\px$. Now, we have $\px\in\alpha$, which implies $\td(X)\le \td(\alpha)$. So, $\alpha\approx_X
\beta$ must hold by E1. Obviously, $D_XX\in\alpha$. Also, using P4 and the definition of $\alpha$, one can check that $\D_X\px\in\alpha$. With
the definition of $\approx_X$, we can infer $\D_X\px\in\beta$. So, $\px\in\beta$. Consequently, $\beta\vDash\px$.
\end{proof}

\begin{fact}\label{fact:phi-model-c6}
For all $\alpha,\beta\in W^{\dag}$, if $\alpha\approx_{\tom X}\beta$, then $G(\alpha)\approx_X G(\beta)$.
\end{fact}

\begin{proof}
    See \ref{appendix:proof-for-fact:phi-model-c6}.
\end{proof}

\begin{fact}\label{fact:phi-model-c7}
For any $\alpha\in W^{\dag}$, if $G(\alpha)\vDash D_XY$, then $\alpha\vDash D_{\tom X}\tom Y$.
\end{fact}

 \begin{proof}
We consider two cases: (1). $G(\alpha)=\emptyset$ and (2). $G(\alpha)\not=\emptyset$. Let us begin.

\vspace{1.5mm}

(1). $G(\alpha)=\emptyset$. Then $G(\alpha)\vDash D_XY$ can only hold by V1 or V2.

\vspace{1mm}

(1.1). When it holds by V1,  $\tom^m\mathbb{V}_{\Phi}\subseteq X$ for some $m\le {\rm{min}}\{\td(y): y\in Y\}$. For this $m$,
we have $\tom^{m+1}\mathbb{V}_{\Phi}\subseteq \tom X$ and $m+1\le {\rm{min}}\{\td(\tom y): y\in Y\}$. Therefore,
$\alpha\vDash D_{\tom X}\tom Y$.

\vspace{1mm}

(1.2). When it holds by V2,  $Y\subseteq X$. So, $\tom Y\subseteq\tom X$. Thus, $\alpha\vDash D_{\tom X}\tom Y$.

\vspace{1.5mm}

(2). $G(\alpha)\not=\emptyset$. The proofs for the cases that $G(\alpha)\vDash D_XY$ holds by V1 and V2 are the same as (1.1) and (1.2) above
respectively. We move to others.

\vspace{1mm}

(2.1). It holds by V3. Then, $D_{X'}Y\in G(\alpha)$ for some $X'\subseteq X$. It is easy to check that $D_{\tom X'}\tom Y\in \alpha$. 
Using $\tom X'\subseteq \tom X$ and V3, we have $\alpha\vDash D_{\tom X}\tom Y$.

\vspace{1mm}

(2.2). It holds by V4. Then, $D_{X'}\tom^m\mathbb{V}_{\Phi}\in G(\alpha)$ for some $X'\subseteq X$ and $m\le {\rm{min}}\{\td(y): y\in Y\}$.
Now, $\tom D_{X'}\tom^m\mathbb{V}_{\Phi}\in \alpha$. Then, one can infer $D_{\tom X'}\tom^{m+1}\mathbb{V}_{\Phi}\in \alpha$. Notice that $m+1\le {\rm{min}}\{\td(\tom y): \tom y\in \tom Y\}$ and  $\tom
X'\subseteq \tom X$. Thus, from V4 it follows that $\alpha\vDash D_{\tom X}\tom Y$.
\end{proof}

\begin{theorem}\label{theorem:phi-is-general-model}
$\bM^{\dag}=(W^{\dag},G,\approx_X,\|\bullet\|^{\dag})$ is a general relational model of $\LTD$.
\end{theorem}

\begin{proof}
Facts \ref{fact:phi-model-c1},  \ref{fact:phi-model-c2}, \ref{fact:phi-model-c3}, \ref{fact:phi-model-c4}, \ref{fact:phi-model-c5},
\ref{fact:phi-model-c6} and \ref{fact:phi-model-c7} show that conditions C1-C7 are satisfied.
\end{proof}

\begin{theorem}\label{theorem:another-truth-lemma-decidability}
Let $\alpha\in W^{\dag}$ be the $i$-type of a state $s$ of the original model. For all $\varphi\in \Phi_i$, 
$\alpha\vDash_{\bM^{\dag}}\varphi$ iff $\varphi\in \alpha$.
\end{theorem}

\begin{proof}
The proof goes by induction on $\varphi\in\mathcal{L}^{\dagger}$. Boolean cases are routine.

\vspace{1.5mm}

(1). $\varphi$ is $\bigcirc\psi$.  $\alpha\vDash\varphi$ iff $G(\alpha)\vDash\psi$. By assumption, $\psi\in \Phi_{i-1}$. Then, by the
inductive hypothesis, $G(\alpha)\vDash\psi$ iff $\psi\in G(\alpha)$. Now, $\psi\in G(\alpha)$ iff $\bigcirc\psi\in \alpha$.

\vspace{1.5mm}

(2). $\varphi$ is $\D_X\psi$. We consider the two directions separately.

\vspace{1mm}

(2.1). Assume that $\D_X\psi\in\alpha$. So, $\td(X)\le \td(\alpha)$. Let $\beta\in W^{\dag}$ with $\alpha\approx_X\beta$. Notice that
$\alpha\approx_X\beta$ holds by E1. Obviously, $D_XX\in \alpha$. By E1.2, it holds that $\D_X\psi\in\beta$, which can give us $\psi\in\beta$.
By the inductive hypothesis, it follows that $\beta\vDash\psi$. Consequently, $\alpha\vDash \D_X\psi$.

\vspace{1mm}

(2.2). Suppose $\D_X\psi\not\in\alpha$. Then there is some $u$ of $\bM$ s.t. $s=_X u$ and $u\not\vDash\psi$. We now consider the
$i$-type $\beta$ of $u$. Now, $\neg\psi\in \beta$. By the inductive hypothesis, it holds that $\beta\not\vDash\psi$. Also, with E1, we can check that $\alpha\approx_X\beta$. Thus, $\alpha\not\vDash\D_X\psi$.
\end{proof}

Now we can conclude that:

\begin{theorem}\label{theorem:dd-dfd-f.m.p.}
$\LTD$ has the finite model property w.r.t. general relational models.
\end{theorem}

\begin{proof}
Let $\bM$ be a general relational model  and $\varphi\in\mathcal{L}^{\dagger}$ a formula such that
$s\vDash_{\bM}\varphi$. Assume that $\td(\varphi)=i$. Denote by $\Phi$ the closure of $\varphi$ and $\alpha$ the $\Phi_i$-type of $s$ in
$\bM$. From Theorem \ref{theorem:phi-is-general-model}, it follows that the corresponding model $\bM^{\dag}$ is finite. Now, from Theorem \ref{theorem:another-truth-lemma-decidability}, it follows that $\alpha\vDash\varphi$, as expected.
\end{proof}

As a consequence, it holds that:

 \begin{theorem}\label{theorem:decidable}
$\LTD$ is decidable.
\end{theorem}

\section{Further directions}\label{sec:further-directions}
Up to this point, we have explored a series of logics for dynamic dependence. We started with the system $\LDTV$ that features temporalized variables $\tom x$, and then enriched it with function symbols and global equality. Subsequently, motivated by ordinary temporal logics, we analyzed a general intuition reducing $\tom \varphi$-type statements about the future to current statements about  temporalized variables, and showed that this commits us to timed models, an interesting subclass of dynamical systems. Finally, we presented our most general logic of temporal dependence over arbitrary dynamcial systems. For all these logics, we found  nice properties, viz. decidability of the satisfiability problem and complete Hilbert-style axiomatizability of the validities. Along the way, we introduced various new reduction methods that may find other uses, and we extended existing methods (frame correspondence, geenralized modal models, and filtration)  which also enhance our understanding of the modal base logic $\LFD$ of static dependence.

Several technical open problems remain for the systems presented here. First, w.r.t. our overall design, the sets of terms $X$ in formulas $D_Xy$ and $\D_X\varphi$ were required to be non-empty. Although after dropping the restriction we can still find complete Hilbert-style proof systems for all our logics and show  decidability of $\LTD^{\mathsf{f}}$ with the timed semantics, the requirement of non-emptiness is pivotal in our decidability proof for the most general logic $\LTD$ with the un-timed semantics. Also left open is an explicit complete axiomatization for the logic extending $\LTD$ with functional symbols and term identify. Next, we have merely touched upon the frame correspondence analyses of modal languages with `relational atoms' $D_Xy$ in Section \ref{sec:next-time-modality}, but it remains to understand their correspondence theory comprehensively. A  useful thing to have would be a Gentzen-style proof-theoretic  treatment of our systems, as has already been given for the basic dependence logic $\LFD$. 

 It should also be noted that our emphasis in this paper may be just one half of the story. A natural counterpoint to our concerns in this paper would be the study of the logic of \emph{independent variables} in dynamical systems, on the analogy of notions of independence known in the setting of the base logic $\LFD$.
 
Next, here are some comparisons with other frameworks. First, compared to existing  \emph{temporal logics},  our languages are poor in expressive power, since they can only talk about next time steps, iterated up to specific depths.  One natural  addition would be a one-step past operator as an existential modality describing the previous, rather than the next stage.\footnote{This may also provide a new perspective on the completeness and decidability of $\mathsf{LTD}$, since we can then express the part of the crucial semantic condition `$s=_{\tom X}t$ iff $g(s)=_X g(t)$' which posed problems in our technical treatment, viz. as an axiom $\D_{\tom X}\neg \varphi\to \tom \D_X \neg \tom^{-1}\varphi$.}  More ambitious would be adding the unbounded future operators like $\tom^* x$ and $\tom^* \varphi$  which allow us to reason about the eventual long-term behavior of dynamical systems. 

 Also, our models of dynamical systems essentially  describe one function on a set of structured states carrying variable assignments. Thus, there are no genuine choices or branching futures, but only linear progression with possibly branching pasts. Allowing non-determinism might  even simplify axiomatizing some of the logics explored in the work. However, to us, the most urgent enrichment on the agenda  follows practice in using dynamical systems which often come with \emph{topologies} on the value ranges for variables, or on the state space itself. This would link our systems to {\em dynamic topological logics} in the style of \cite{artemov-1997,dtl,handbook-dtl}. But a topological setting also offers a  range of new notions and results in dependence logic, since we can now study continuous (and even uniformly continuous) dependence in empirical settings.  We refer to \cite{AJ-CD} for an extensive exploration of decidable and axiomatizable systems in this spirit, and to \cite{Dazhu2021}, which interprets $\LTD$ w.r.t. dynamic topological spaces and enriches the latter with a fine-structure of variables, for  first results on  topological dependence  in a temporal setting.

 Finally, at another interface, the combination of temporal operators and dependence modalities read epistemically \citep{AJ-Dependence}, suggests comparisons with  {\em epistemic-temporal logics} such as $\ETL$  \cite{reason-about-knowledge}. To see how this less obvious connection works, here are a few examples. Consider the setting that  each temporalized variable $x$ is treated as an agent; $\tom x$ refers to the next stage of $x$, who may know more (or less) than $x$; and $\D_x\varphi$ says that $x$ knows $\varphi$. With this understanding, our temporal dependence logics provide fresh perspectives on important principles in  $\ETL$. For instance, consider the following law in  $\ETL$,  reformulated with our language:
\begin{center}
{\em No Learning} (a.k.a. {\em No Miracles}):\qquad  $\tom \D_x\varphi \to \D_{x}\tom \varphi$     
\end{center}
 stating that {\em if at the next stage $x$ knows that $\varphi$, then $x$ now already knows that $\varphi$ will be true at the next stage}. This implication is not valid in our temporal dependence logics, but our systems do provide the following valid analog:
\begin{center}
 $\tom \D_x\varphi \to \D_{\tom x}\tom \varphi$   
\end{center}
which introduces an interesting new distinction into an agent logic like $\ETL$ between agents as they are now and as they will be at some later moment. \footnote{A similar analogy connects the well-known axiom for  {\em Perfect Recall} (i.e., $\D_{x}\tom \varphi \to \tom \D_x\varphi$) and the epistemic-temporal principle  $\D_{\tom x}\tom \varphi \to \tom \D_x\varphi$ with temporalized agents.} Despite these analogs, logical systems for  full epistemic-temporal  languages are usually  undecidable and non-axiomatizable when the expressive repertoire becomes  strong \citep{Benthem2006TheTO,etl-1989}, while our systems so far lie on the good side of this boundary.\footnote{We can even give the temporal operator $\tom$ a more concrete interpretation by  combining  paradigms of dynamic epistemic logics like the {\em public announcement logic} $\PAL$ \citep{BMS,Johan-del,hans-del}, and given an agent $x$, its next stage can be $!\varphi; i$, meaning the agent $x$ resulting from the public announcement $[!\varphi]$. It is  interesting to explore what $\PAL$ and our logics would look like if we have this kind of agents, which is also useful in enriching the standard view in $\ETL$ \citep{merging-del-etl}.}

 Still more structure would be needed to make good on the suggestion in \citep{AJ-Dependence} that the present style of  analysis might carry over to the dependencies over time that are found with causality and in games.

\section*{Acknowledgements} The
second author was supported by the Tsinghua University Initiative Scientific Research Program. The third author was supported by the National Social Science
Foundation of China [22CZX063].



 \bibliographystyle{elsarticle-num} 
 \bibliography{mybib}


\appendix

\section{Proof for Fact \ref{fact:timed-LTD-model-transition-1}}\label{appendix:proof-for-fact:timed-LTD-model-transition-1}

\begin{proof}
This can be proven by induction on the complexity of $x$.

\vspace{1.5mm}

(1). For the basic case $x=v\in V$, we use sub-induction on $n\in\mathbb{N}$.

\vspace{1mm}

(1.1). For $n=0$, we have ${\bm x}(s,0)=({\bm v}(s),0)=({\bf Tr}({\bm v})(s),0)$, as needed.

\vspace{1mm}

(1.2). For $n+1$, ${\bm v}(s,n+1)=(I(f_v)([{\bm v_1}(s,n)]^-,\ldots,[{\bm v_N}(s,n)]^-),n+1)$. Then, by the sub-inductive hypothesis, it holds that  
\begin{center}
\begin{tabular}{rcl}
${\bm
v}(s,n+1)$ & $=$  &$(I(f_v) ({\bf Tr}(\tom^n {\bm v_1})(s), \ldots, {\bf Tr}(\tom^n{\bm v_N})(s)),n+1)$\\
& $=$  & $({\bf Tr}(\tom^{n+1} {\bm v})(s),n+1)$
\end{tabular}
\end{center}

(2). For the inductive case $x:=\tom y$, we have 
\begin{center}
    ${\bm x}(s,n)=\tom {\bm y}(s,n)={\bm y}(g(s,n))={\bm y}(s,n+1)$.
\end{center}
Using the inductive hypothesis,
we obtain
\begin{center}
   ${\bm x}(s,n)={\bm y}(s,n+1)=({\bf Tr}(\tom^{n+1}{\bm y})(s), n+1+\mathsf{td}(y))=({\bf Tr}(\tom^{n}{\bm x})(s),n+\mathsf{td}(x))$, 
\end{center}
as desired.
\end{proof}

\section{Proof for Fact \ref{fact:LTD-t-lfd-1}}\label{appendix:proof-for-fact:LTD-t-lfd-1}

\begin{proof}
We use induction on the complexity of the formula $\varphi\in\mathcal{L}$. The cases for Boolean connectives $\neg,\land$ are straightforward, and we now consider other cases.

\vspace{1.5mm}

(1). Formula $\varphi$ is $P(x_1,\ldots,x_m)$. Then, ${\bf Tr}(\varphi[V/\tom^n V])$ is exactly the formula $P({\bf Tr}(\tom^n x_1),\ldots,{\bf Tr}(\tom^nx_m))$. So, we have:  
\begin{center}
\begin{tabular}{rcl}
&& $s\vDash_{\bM} P({\bf Tr}(\tom^n x_1),\ldots,{\bf Tr}(\tom^nx_m))$\\
& iff  & $({\bf Tr}(\tom^n {\bm x_1})(s),\ldots,{\bf Tr}(\tom^n{\bm x_m})(s))\in
 I(P)$\\
& iff  &  $(({\bf Tr}(\tom^n {\bm x_1})(s),n+\mathsf{td}(x_1)),\ldots,({\bf Tr}(\tom^n{\bm x_m})(s),n+\mathsf{td}(x_m)))\in
 I'(P)$\\
& iff  & $({\bm x_{1}}(s,n),\ldots,{\bm x_{m}}(s,n))\in I'(P)$\\
& iff  & $(s,n)\vDash_{\bM^{\downarrow}}\varphi$
\end{tabular}
\end{center}
Notice that the third equivalence holds by Fact \ref{fact:timed-LTD-model-transition-1}.

\vspace{1.5mm}

(2). $\varphi$ is $D_Xy$. ${\bf Tr}(\varphi[V/\tom^n V])$ is $D_{{\bf Tr}(\tom
^nX)}{\bf Tr}(\tom^ny)$. Now, we have:
\begin{center}
\begin{tabular}{rcl}
&& $s\vDash_{\bM} D_{{\bf Tr}(\tom ^nX)}{\bf Tr}(\tom^ny)$\\
 & iff  & for all $w\in A$,
$s=_{{\bf Tr}(\tom ^nX)}w$ implies $s=_{{\bf Tr}(\tom ^ny)}w$\\
  & iff  &for all $(w,m)\in S$, $(s,n)=_X (w,m)$ implies $(s,n)=_y (w,m)$\\
  &iff &$(s,n)\vDash_{\bM^{\downarrow}}\varphi$
\end{tabular}
\end{center}
\vspace{1.5mm}

(3). $\varphi$ is $\D_X\psi$. Formula ${\bf Tr}(\varphi[V/\tom^n V])$ is $\D_{{\bf Tr}(\tom^n
X)}{\bf Tr}(\psi[V/\tom^n V])$. Then, 
\begin{center}
\begin{tabular}{rcl}
&& $s\vDash_{\bM}\D_{{\bf Tr}(\tom^n X)}{\bf
Tr}(\psi[V/\tom^nV])$\\
 & iff  &  for all $w\in A$,  $s=_{{\bf Tr}(\tom ^nX)}w$ implies  $w\vDash_{\bM} {\bf
Tr}(\psi[V/\tom^nV])$\\
&iff& for all $w\in A$, $s=_{{\bf Tr}(\tom ^nX)}w$ implies
$(w,n)\vDash_{\bM^{\downarrow}} \psi$ \\
&iff&  for all $(w,m)\in A$, $(s,n)=_X (w,m)$ implies $(w,m)\vDash_{\bM^{\downarrow}} \psi$ \\
&iff& $(s,n)\vDash_{\bM^{\downarrow}}\varphi$. 
\end{tabular}
\end{center}
The proof is completed.
\end{proof}

\section{Proof for Fact \ref{facf:timed-system}}\label{appendix:proof-for-facf:timed-system}

\begin{proof}
The directions from (1) to (2) and from (2) to (3) are easy to check. To see that the direction from (3) to (4) holds, we just need to put:
$s=^{\tau}w$ iff the maximal $g$-histories of $s,w$ have the same length. Now it suffices to prove the direction from (4) to (1).

Let $=^{\tau}$ be an equivalence relation satisfying the synchronicity conditions. We define a map $\tau$ as follows: (i). if there is no $w$
with $s=^{\tau}g(w)$, then for all $t\in\bS$ such that $s=^{\tau}t$, $\tau_t:=0$; (ii). if $s=^{\tau}g(s)$, then for all $t\in\bS$ such that
$s=^{\tau}t$, $\tau_t:=\infty$; and (iii). for all $\tau_s\in \mathbb{N}\cup\{\infty\}$ and $t=^{\tau}g(s)$, $\tau_{t}:=\tau_s+1$. Notice that
the two synchronicity conditions satisfied by $=^{\tau}$ ensure that $\tau$ is well-defined on the system $\bS$: for every dynamical state
$s\in\bS$, we have a unique value $\tau_s$. Also, it is a matter of direct checking that the clause (iii) ensures the property (b) in
Definition \ref{def:timed-system}. Thus, to show that $\tau$ is a timing map, we only need to prove that condition (a) is satisfied.

Let $s$ be an initial state. Suppose  $\tau_s\not=0$. If $1\le\tau_s=n\in\mathbb{N}$, then the value $\tau_s$ must be given by clause (3),
which is impossible. Next,  consider the case that $\tau_s=\infty$. Then, by clause (ii), there is some $w$ such that $w=^{\tau}g(w)$ and
$w=^{\tau}s$. As $=^{\tau}$ is an equivalence relation, we have $s=^{\tau}g(w)$. Now, from the second condition of the synchronicity relation,
we know that there is some $w'$ such that $s=g(w')$, which contradicts the fact that $s$ is an initial state.
\end{proof}

\section{Proof for Fact \ref{fact:satisfiability-timed-without-O-modality}}\label{appendix:proof-for-fact:satisfiability-timed-without-O-modality}

\begin{proof}
It can be proved in a similar way to that for $\LDTV^{\mathsf{f},\equiv}$ (Theorem \ref{theorem:decidable-LTD-equiv-funct}), but we need to modify the constructions. First, we define a translation $\textbf{Tr}^{\mathsf{t}}$ that is the same as $\textbf{Tr}$ used in Theorem \ref{theorem:decidable-LTD-equiv-funct}, except that $\textbf{Tr}^{\mathsf{t}}(D_Xy):=D_{\textbf{Tr}^{\mathsf{t}}(X)\cup\{v_{N+1}\}}\textbf{Tr}^{\mathsf{t}}(y)$ and $\textbf{Tr}^{\mathsf{t}}(\D_X\varphi):=D_{\textbf{Tr}^{\mathsf{t}}(X)\cup\{v_{N+1}\}}\textbf{Tr}^{\mathsf{t}}(\varphi)$. The resulting language $\mathcal{L}^{f,\equiv}$ also contains an additional variable $v_{N+1}$ standing for `time'.

\vspace{1.5mm}

Now, let $\bM=(\mathcal{O},I,A)$  be a model of $\LFD^{\mathsf{f},\equiv}$ and $s_0\in A\subseteq \mathcal{O}^V$.  We  define a timed dynamical model
$\bM^{\downarrow}=(\mathbb{D}_v,I',S,g,{\bm v})_{v\in V}$ over variables $V$, by taking:
\begin{itemize}
\item[$\bullet$] For all $v\in V$, $\mathbb{D}_{v}:=\mathcal{O}$.
\item[$\bullet$] $I'$ is the restriction of $I$ to all predicate symbols and functional symbols $f$ (except those $f_v$).
\item[$\bullet$] $S:=\{(s, i): i\in\mathbb{N}\;\textit{and}\; s\in A\; \textit{with}\; {\bm v_{N+1}}(s)={\bm v_{N+1}}(s_0) \}$.
\item[$\bullet$] $g(s, i):= (s,i+1)$.
\item[$\bullet$] For each $v\in V$,  ${\bm v}(s, i)$ are recursively defined by putting:
\begin{center}
${\bm v}(s,0):=s(v)$, \quad ${\bm v}(s,n+1):=I(f_v)({\bm v_1}(s,n),\ldots,{\bm v_N}(s,n))$,
\end{center}
where $I(f_v)$ is the interpretation of $f_v$ as an actual $N$-ary function in $\bM$.
\end{itemize} 
By construction, one can check the following:
\begin{itemize}
    \item[$\bullet$] For any states $(s,i), (t,j)\in S$,   $(s,i)=^{\tau}_X(t,j)$ iff $i=j$ and $(s,i)=_X (t,j)$.
    \item[$\bullet$] For all $x\in \term$ and $(s,n)\in S$, it holds that ${\bm x}(s,n)={\bf Tr^{\mathsf{t}}}(\tom^n { \bm x})(s)$. 
    \item[$\bullet$]   $s\vDash_{\bM} {\bf Tr^{\mathsf{t}}}(\varphi[V/\tom^nV])$\; iff\; $(s,n)\vDash_{\bM^{\downarrow}} \varphi.$ 
\end{itemize}

\vspace{1.5mm}

Conversely, given a timed dynamical model
$\mathcal{M}=(\mathbb{D}_v,I,S,g,{\bf v})_{v\in V}$, we can convert it into an $\LFD^{\mathsf{f}}$-model $\mathcal{M}^+=(\mathcal{O},I^+,A)$
for a language (with $N+1$ variables $V\cup\{v_{N+1}\}$), which is defined in the same way as that in the proof for Theorem \ref{theorem:decidable-LTD-equiv-funct}, except the following:
\begin{itemize}
\item[$\bullet$] $\mathcal{O}:= \mathbb{N}\cup\{\infty\}\cup\bigcup_{v\in V}\mathbb{D}_v$, where $\infty$ is a fresh object.
\item[$\bullet$] $A:=\{s^+: s\in S\}$ s.t. ${\bm v}(s^+):= {\bm v}(s)$ for each $v\in V$ and ${\bm v_{N+1}}(s^+):= \tau_s$.
\end{itemize}
Now we can show the following:
\begin{center}
  $s\vDash_{\mathcal{M}} \varphi$ \,  iff  \, $ s^+\vDash_{\mathcal{M}^+}  {\bf Tr}^{\mathsf{t}}(\varphi)$.
\end{center}

\vspace{1.5mm}

Finally, notice that the model $\bM^{\downarrow}$ is a linear-time dynamical model with finite past, so we have the desired result.
\end{proof}

\section{Proof for Fact \ref{fact:completeness-timed-without-O-modality}}\label{appendix:proof-for-fact-completeness-timed-without-O-modality}
\begin{proof}
The proof is similar to that for Theorem \ref{theorem:theorem-preserving}, and we are also going to construct a theorem preserving function $\mathcal{T}$  from formulas of $\LDTV^{\mathsf{t},\mathsf{f},\equiv}$  to formulas of $\LFD^{\mathsf{f},\equiv}$  such that $\varphi$ is a theorem of $\bf{LDTV^{f,\equiv}}$ iff $\mathcal{T}(\varphi)$ is a theorem in $\bf{LFD^{f, \equiv}}$. But now, we will work with the translation ${\bf Tr}^{\mathsf{t}}$ defined in the proof for Fact \ref{fact:satisfiability-timed-without-O-modality}.

Let $T=\{{\bf Tr}^{\mathsf{t}}(x): x \textit{~is a term of~} \LFD^{\mathsf{f},\equiv}\}$. Also, we define $T^+=T\cup\{v_{N+1}\}$. Clearly, both $T$ and $T^+$ are closed under sub-terms, and the translation ${\bf Tr}^{\mathsf{t}}$ from terms of $\LFD^{\mathsf{f},\equiv}$ to $T$ is bijective.

Next, we define a map $\rho$  from $\LFD_{\mathsf{T}}^{\mathsf{f},\equiv}$-formulas to $\LFD_{\mathsf{T}^{\mathsf{+}}}^{\mathsf{f},\equiv}$-formulas that adds a new variable $v_{N+1}$ to the subscripts of all subformulas $\D_X\varphi$ and $D_Xy$. Clearly, the resulting map is injective, and its range coincides with the range of the translation ${\bf Tr}^{\mathsf{t}}$. Moreover, the map ${\bf Tr}^{\mathsf{t}}$ from formulas of $\LDTV^{\mathsf{t},\mathsf{f},\equiv}$  to the range of $\rho$ is bijective. Also, for the map $\rho$, we can prove the following:

\vspace{1.5mm}

\noindent{\bf Claim:}\; For any $\LFD_{\mathsf{T}}^{\mathsf{f},\equiv}$-formula $\varphi$ and $\LFD_{\mathsf{T}}^{\mathsf{f},\equiv}$-model $\bM$, we have
\begin{center}
 $s\vDash_{\bM} \varphi$\; iff \; $s^{v_{N+1}}\vDash_{\bM^{v_{N+1}}} \rho (\varphi)$,
\end{center}
\noindent where $\bM^{v_{N+1}}$ is a model of $\LFD_{\mathsf{T}^\mathsf{+}}^{\mathsf{f},\equiv}$ obtained by giving the new variable $v_{N+1}$ a constant value in $\bM$ and $s^{v_{N+1}}$ is the admissible assignment corresponding to $s$.

\begin{proof}
It is by induction on formulas of $\LFD_{\mathsf{T}}^{\mathsf{f},\equiv}$. Instead of showing details, we note just one key fact: since the new variable $v_{N+1}$ is a constant in $\bM^{v_{N+1}}$, for any terms $X$ in $\LFD_{\mathsf{T}}^{\mathsf{f}}$ and
$s,t\in \bM$, $s^{v_{N+1}}=_X t^{v_{N+1}}$ iff $s=_Xt$.
\end{proof}

Now, we can show that $\mathcal{T}:=\rho^{-1}\circ{\bf Tr}$ is a theorem-preserving translation.

\vspace{1.5mm}

Suppose that $\not\vdash_{{\bf LFD^{f,\equiv}_T}} \mathcal{T}(\varphi)$. Then, from the completeness of ${\bf LFD^{f,\equiv}_T}$ (recall Fact
\ref{prop:decidability-lfd-equiv-funct}), we know that $\mathcal{T}(\neg\varphi)$ is satisfiable. Using the claim above, ${\bf Tr}^{\mathsf{t}}(\neg\varphi)$ is
satisfiable, and so from the proof for Fact \ref{fact:satisfiability-timed-without-O-modality}, it follows that $\neg\varphi$ is satisfiable. By the soundness of $\bf{LDTV^{f,\equiv}}$, it holds
that $\not\vdash_{\bf{LDTV^{f,\equiv}}}\varphi$.

\vspace{1.5mm}

For the converse direction, note that for any axiom $AX$ of ${\bf LFD^{f,\equiv}_T}$,  $\mathcal{T}^{-1}(AX)$ is an axiom or a theorem of $\bf{LDTV^{f,\equiv}}$; and every correct application  of a rule of ${\bf LFD^{f,\equiv}_T}$ is mapped by $\mathcal{T}^{-1}$ into a correct
application of a rule in $\bf{LDTV^{f,\equiv}}$. It follows that, if $\vdash_{{\bf LFD^{f,\equiv}_T}}\mathcal{T}(\varphi)$, then $\vdash_{\bf{LDTV^{f,\equiv}}}
\mathcal{T}^{-1}(\mathcal{T}(\varphi))$, i.e., $\vdash_{\bf{LDTV^{f,\equiv}}}\varphi$.
\end{proof}

\section{Proof for Fact \ref{fact:canonical-general-model}}\label{appendix:proof-for-fact:canonical-general-model}
\begin{proof}
(1). Condition C1. Given the  \textbf{S5}-axioms for dependence quantifiers, all $=^c_X$ are equivalence relations by a standard  argument.

\vspace{1mm}

(2). Condition C2. The axioms for $D_Xy$ were precisely designed to ensure the truth of  conditions of `Dep-Reflexivity',
`Dep-Transitivity' and `Determinism'.

\vspace{1mm}

(3). Condition C3. Let $s=^c_Xt$ and $D_XY\in s$. Using $\D$-Introduction$_2$, we get $\D_XD_XY\in s$. Since $s=^c_Xt$, we have $D_XY\in t$.
Next, let $\D_Y\varphi\in s$. Then, using  Transfer, we get $\D_X\varphi \in s$, and hence $\varphi\in t$.

\vspace{1mm}

(4). Condition C4. Let $s=^c_Xt$ and $\px\in s$ (where $X$ is the set of terms occurring in ${\bm x}$).  Using $\D$-Introduction$_1$, we
have $\D_X\px\in s$, and hence $\px\in t$.

\vspace{1mm}

(5). Condition C5. We have: $P(\tom x_1,\ldots,\tom x_n) \in s$ iff $\tom P(x_1,\ldots, x_n) \in s$ iff $P(x_1,\ldots, x_n) \in
g^c(s)$. The first equivalence holds by axiom Atomic-Reduction, and the second follows from the definition of $g^c$.

\vspace{1mm}

(6). Condition C6. Assume that $s=^c_{\tom X} t$, and $\D_X\varphi\in g^c(s)$. From the latter, we have $\tom\D_X\varphi\in s$. Now, using
w-Next-Time$_1$, $\D_{\tom X}\tom \varphi\in s$. As $s=^c_{\tom X} t$, it holds immediately that $\tom \varphi\in t$, and thus $\varphi\in
g^c(t)$.

\vspace{1mm}

(7). Condition C7. Assume $D_XY\in g^c(s)$. It is simple to see $\tom D_XY\in s$. Using w-Next-Time$_2$, we have $D_{\tom X}\tom Y\in s$. This completes the proof.
\end{proof}

\section{Proof for Theorem \ref{theorem:Repr}}\label{appendix:proof-for-theorem:Repr}
\begin{proof}
Let $\mathcal{M}=(A,g,=_X,\|\bullet\|)$ be a general relational model of $\LTD$. We are going to construct a standard
relational model $\bM=(W, G, \sim_X, \|\bullet\|_\bM)$. 

\medskip

\par\noindent\textbf{Histories}\;
We fix a state $s_0\in A$.
A \textit{history} $h=(s_0,\alpha_0,s_1,\ldots,\alpha_{n-1},s_n)$ is a finite sequence, with states $s_0, \ldots, s_n\in A$ and
`transitions' $\alpha_0, \ldots, \alpha_{n-1}\in  \{g\}\cup \{X\subseteq \term: X \mbox{ is finite}\}$, subject to the following requirement,
for all $k < n$:  $s_k=_{\alpha_k} s_{k+1}$ whenever $\alpha_k \subseteq \term$, and $g(s_k)=s_{k+1}$ whenever $\alpha_k=g$. The set of
\textit{worlds} $W$ of our standard model will be the set of all histories. We denote by $last:W\to A$ the map sending a history $h$ to its
last state $last(h)$; i.e.,  $last(s_0,\alpha_0,s_1,\ldots,\alpha_{n-1},s_n):=s_n$. The \emph{immediate term-succession} relation $h\to h'$ is
defined on $W$ by putting: $h\to h'$ iff $h'=(h,Y, s')$ for some finite $Y\subseteq\term$ and some $s'\in A$; while the dynamical transition
relation $h \stackrel{G}{\to} h'$ is defined by putting: $h\to h'$ iff $h'=(h,g, g(last(h)))$. We denote by $\ot$ the converse of  $\to$, and similarly $ \stackrel{G}{\ot}$ denotes the converse of  $\stackrel{G}{\to}$.
Finally, the \emph{succession relation} $h\preceq h'$ is the reflexive-transitive closure of the union $\to\cup
\stackrel{G}{\to}$ of the immediate term-succession and dynamical transition relations.



\medskip

\par\noindent\textbf{Paths between histories}\; These histories from a natural tree-like structure, partially ordered by the succession
relation $h\preceq h'$, and having the `tree property': each history $h\neq (s_0)$ has a unique immediate predecessor, and moreover the set of
all its predecessors is totally ordered by  $\preceq$. Every two histories $h, h'$ have a unique \emph{greatest lower
bound} $inf_\preceq (h,h')$, given by their largest shared sub-history. Two histories $h, h'$ are \emph{neighbors} if either one is the
immediate predecessor of the other. A \emph{path} from a history $h$ to another history $h'$ is a chain of histories $(h_0, \ldots, h_n)$,
having $h_0 = h$ and $h_n = h'$ as its endpoints, and s.t. for every $k$, histories
$h_k$ and $h_{k+1}$ are neighbors. A path is \emph{non-redundant} if no history appears twice in the chain.
The tree-structure of $W$ ensures that there exists a \emph{unique non-redundant path} between any
two histories $h$ and $h'$. This path can be pictured in terms of first `going down' from $h$ to its predecessors until reaching the largest
shared sub-history $inf_\preceq (h,h')$, and then `going up' again to its successors until reaching the end of $h'$. This visual picture of
our tree-like model may help in understanding the arguments to follow.

\medskip

\noindent\textbf{Dynamical transitions function}\; We now define $G$ by putting:
$$G(h) \,\, := \,\, (h,g, g(last(h))).$$
In other words, $G(h)=h'$ is equivalent to $h  \stackrel{G}{\to} h'$.  

\medskip

\par\noindent\textbf{Observation 1:}\; If $G(h)= h'$, then $g(last(h))=last(h')$.

\medskip
\par\noindent\textbf{Valuation}\; We now  define the valuation $\|\bullet\|_\bM$ of our standard model $\bM$ for atoms $\px$ simply in terms
of truth at the last world in the history:
$$h\in \|P{\bm x}\|_\bM \, \, \mbox{ iff } \,\, last(h)\vDash P{\bm x}.$$

Next, we want to define the \emph{$X$-value equivalence relations} $\sim_X$ in  $\bM$. This has to be done in such a way that we
`improve' the given general relational model in three respects: (i) relations $\sim_X$ become intersections of the $\sim_x$ for $x \in X$,
(ii)   $h\sim_{\tom x} h'$ iff $G(h)\sim_x G(h')$, and (iii) atoms $D_Xy$ get their standard semantic interpretation at histories $h$ in a way
that matches with their truth in $\mathcal{M}$ at $last(h)$.

\medskip

\par\noindent\textbf{Auxiliary one-step relations}\;
Before defining $X$-equivalence, we first need to introduce some auxiliary relations $h\to_X h'$ between histories, one for each finite,
non-empty set $X\subseteq \term$: 
\begin{center}
\begin{tabular}{rcl}
 $h \to_X h'$ & iff & $h'=(h, Z, s')\, \& \, last(h)\vDash D_Z X$, \\
 && for some $Z\subseteq \term$ and some $s'\in A$.
\end{tabular}
\end{center}
\noindent Note that $\to_X$ is included in the immediate successor relation $\to$. We denote by $\ot_X$ the converse of $\to_X$.
With respectively the conditions C3, C6 of Definition \ref{def-generalrelationalmodel} and the equivalence of $D_Z (X\cup Y)$ and $D_Z X\land
D_Z Y$, it is easy to check that:

\vspace{1.5mm}

\par\noindent\textbf{Observation 2:}\; $h\to_{\tom^mX} h'$ implies $g^m(last(h))=_X g^m(last(h'))$.

\vspace{1.5mm}

\par\noindent\textbf{Observation 3:}\;
$h\to_{X\cup Y} h'$ holds iff both $h\to_X h'$ and $h\to_Y h'$ hold.

\vspace{1.5mm}

\par\noindent\textbf{$X$-chains} \; Given a finite, non-empty set $X\subseteq \term$, an \emph{$X$-chain} is a finite sequence of histories
of the form:

\vspace{-3ex}

\begin{align*}
&h=h_0 \ot_X h_1 \ldots \ot_X h_{m_1-1} \stackrel{G}{\ot} h_{m_1} \ot_{\tom X} \ldots \ot_{\tom X} h_{m_2-1}\\
&\stackrel{G}{\ot} h_{2} \ot_{\tom^2X}\ldots h_{m_N-1}\stackrel{G}{\ot}
h_{m_{N}}\ot_{\tom^NX} \ldots \ot_{\tom^NX} h_{m_{N}+k},
\end{align*}

\noindent for some $N, m_1, \ldots, m_N,k\in\mathbb{N}$. The smallest history $h_{m_{N}+k}$ in the $X$-chain is called the \emph{origin} of
the chain, while the largest history $h=h_0$ is  the \emph{end} of the $X$-chain. The number $m_N+k$ is called the \emph{total length} of the
$X$-chain, while $N$ (indicating the number of occurrences of $G$ in the $X$-chain) is called the \emph{$G$-length} of the $X$-chain. A
special case of $X$-chain is a \emph{zero $X$-chain}, i.e., one having total length $m_{N}+k=0$; such an $X$-chain just consists of a single
history $(h)$.

\vspace{1.5mm}

\par\noindent\textbf{$X$-paths}\quad
Given two $h$ and $h'$, a path from $h$ to $h'$ is called an \emph{$X$-path} if it includes two $X$-chains of some equal $G$-length $N$,
having respectively $h$ and $h'$ as their ends, such that 
the two chains have a common origin.
 Clearly, if an $X$-path from $h$ to $h'$ exists, then the non-redundant path from $h$ to $h'$ can give us such an $X$-path.

\vspace{1.5mm}

\par\noindent\textbf{$X$-equivalence relations}\quad We can now define the \emph{$X$-value equivalence relations $\sim_X$} of our intended
model $\bM$: for all finite $X\subseteq \term$ and $h, h'\in W$,
\begin{center}
  $h\sim_X h'$ \,\,  iff  \,\, there exists some $X$-path from $h$ to $h'$.
\end{center}

This completes our definition of the model $\bM$, which is well-defined:





\begin{lemma}\label{lemma:tree-proof-standard-model} $\bM=(W, G, \sim_X,\|\bullet\|_{\bM})$ is a standard relational model.
\end{lemma}

\begin{proof}
The first condition in the surplus of standard relational models over general relational models is that the relation $\sim_X$ equals the intersection $\bigcap_{x \in X} \sim_x$. This is a crucial feature of the above tree construction, which cannot be enforced routinely by means of the standard accessibility relations in the canonical model \citep{BML}. We therefore state it as a separate fact.

\vspace{1.5mm}

\par\noindent\textbf{Claim 1:}\; For any two histories $h$ and $h'$, $h\sim_Xh'$ iff $h\sim_xh'$ for all $x\in X$.

\begin{proof}
It suffices to prove the following:
\begin{center}
 $h\sim_{X\cup Y} h'$ holds iff both $h\sim_X h'$ and $h\sim_Y h'$ hold.
\end{center}
The direction \emph{from left to right} is obvious: by Observation 3, every $X\cup Y$-chain is both an $X$-chain and a $Y$-chain, and thus a
pair of $X\cup Y$-chains with the desired property (of forming an $X\cup Y$-path from $h$ to $h'$) is also a pair of $X$-chains as well as
$Y$-chains with the same property. 

\vspace{1.5mm}

For \emph{the right to left direction}, suppose $h\sim_X h'$ and $h\sim_Y h'$. 
Then the non-redundant path from $h$ to $h'$ has to contain a pair of $X$-chains as well as a pair of $Y$-chains with the desired properties (that make
it both an $X$-path and an $Y$-path). Moreover, we can treat $inf_\preceq (h,h')$ as the common origin of the two $X$-chains as well as the common origin of the two $Y$-chains. So, the $X$-chains are exactly the $Y$-chains. Also, we can show that the two chains are $X\cup Y$-chains: indeed, all the transition steps are both $X$-transitions $\to_{\tom^mX}$ and $Y$-transitions $\to_{\tom^mY}$, and
thus by Observation 3 they are $X\cup Y$-transitions $\to_{\tom^m(X\cup Y)}$. Then, we can  easily check that the two chains satisfy the required conditions for an $X\cup Y$-path.
\end{proof}

Next, we show that:

\vspace{1.5mm}

\par\noindent\textbf{Claim 2:}\; The relations $\sim_X$ are equivalence relations.

\begin{proof}
Reflexivity and symmetry are immediate. For transitivity, let us assume $h\sim_X h'\sim_X h''$, and we will show $h\sim_X h''$. Here we
already know that the non-redundant paths from $h$ to $h'$ and from $h'$ to $h''$ are $X$-paths, and we now have to prove the same assertion
about the non-redundant path from $h$ to $h''$.

Let $N$ be the common $G$-length of the two $X$-chains in the non-redundant $X$-path from $h$ to $h'$, and let $h_0$ and $h'_0$ be the origins
of these two chains; we know that  $h_0=h'_0=inf_\preceq (h,h')$. Similarly, let $M$ be the common $G$-length of
the two $X$-chains in the non-redundant path from $h'$ to $h''$, and let $h'_1$ and $h''_1$ be their origins; we know that  $h'_1=h''_1=inf_\preceq (h', h'')$.

Since both $inf_\preceq (h,h')$ and $inf_\preceq(h', h'')$ are predecessors of the history  $h'$, they must be comparable w.r.t. $\preceq$. Without loss of
generality,  let us assume that $inf_\preceq (h,h') \preceq inf_\preceq(h', h'')$. Then,  $inf_\preceq (h,h'')=inf_\preceq (h, h')$ and
$N\geq M$. Now we have $h'_0\preceq h'_1=inf(h', h'')$. The initial segment of the $X$-chain from $h'_0$ to $h'$ that
lies between $h'_0$ and $h'_1=inf(h', h'')$ consists of $N-M$ $G$-steps, and thus can then be concatenated with the $X$-chain of length $M$
from $h''_1=inf(h', h'')$ to $h''$, to form an $X$-chain of $G$-length $N$ from $h'_0$ to $h''$. So the path from $h$ to $h''$ starts with the
$X$-chain of $G$-length $N$ from $h$ to $h_0$, and end with an $X$-chain the same $G$-length $N$ from $h'_0$ to $h''$.  Recall that  $h_0=h'_0=inf_\preceq (h,h')=inf_\preceq (h, h'')$.  Hence, this path is an $X$-path, as needed.
\end{proof}

Next, we show that the transition function $G$ respects $\sim_{V}$:

\vspace{1.5mm}

\par\noindent\textbf{Claim 3:}\; If $h \sim_V h'$, then $G(h) \sim_V G(h')$.

\begin{proof}
The non-redundant path from $h$ to $h'$ can give us a $V$-path, say
\begin{center}
  $h\ot_{V} h_1 \ldots \stackrel{G}{\ot} \ot_{\tom V} \ldots \to_{V} h'.$
\end{center}

Now we are going to prove that there also exists a $V$-path from $G(h)$ to $G(h')$. A crucial observation here is that $h\to_{\tom^mV} h'$
implies $h\to_{\tom^{m+1}V} h'$, since it holds that $\vdash D_Y\tom^m {V} \to D_Y \tom^{m+1} {V}$. With this, we have:
\begin{center}
$G(h) \stackrel{G}{\ot} h \ot_{\tom V} h_1  \ldots \stackrel{G}{\ot} \ot_{\tom^2V} \ldots  \to_{\tom V} h'  \stackrel{G}{\to} G(h')$
\end{center}
from which we can obtain a $V$-path from $G(h)$ to $G(h')$.
\end{proof}

\par\noindent\textbf{Claim 4:}\; If $h\sim_X h'$, then $last(h)=_X last(h')$.

\begin{proof}
From $h\sim_X h'$ we know that the non-redundant path from $h$ to $h'$ is an $X$-path, and therefore, it is composed of two $X$ chains, one
descending and the second ascending, having the same $G$-length $N$. Let
{\small{
$$h \ot_X h_1 \ot_X\ldots  h_{m_1-1} \stackrel{G}{\ot} h_{m_1} \ot_{\tom X} \ldots h_{m_N-1}\stackrel{G}{\ot} h_{m_{N}}\ot_{\tom^NX} \ldots
\ot_{\tom^NX} h_{m_{N}+k}$$}}

\noindent be the first (`descending') $X$-chain of $G$-length $N$ on this non-redundant $X$-path from $h$ to $h'$.
Using Observations 1 and 2 repeatedly:
\begin{align*}
 &   last(h)=_X last(h_1)=_X \ldots =_X last(h_{m_1-1})
=g(last(h_{m_1}))=_X \ldots=_X  \\
&g(last(h_{m_2-1}))=g^2 (last(h_{m_2}))=_X \ldots =_X g^N(last(h_{m_{N}+k})).
\end{align*}
Therefore, we have:
\begin{center}
 $last(h)=_X g^N(last(h_{m_{N}+k})).$
\end{center}

\noindent Reasoning along the second (`ascending') $X$-chain of $G$-length $N$ on the same non-redundant $X$-path from $h$ to $h'$, we similarly obtain that
 \begin{center}
$last(h')=_X g^N(last(h_{m'_{N}+k'})),$
 \end{center}
 
 \noindent where $h_{m'_N+k'}$ is the origin of this second $X$-chain. Also, we have $h_{m_{N}+k}=
inf_\preceq (h,h')=h_{m'_N+k'}$, which can give us
$last(h)=_X g^N(last(h_{m_{N}+k})) = g^N( last(inf_\preceq (h,h'))) = g^N(last(h_{m'_{N}+k'})) =_X last(h')$, as desired.
\end{proof}
 




With the above result, we can show that the truth values of non-dependence atoms $\px$ are invariant in the way required by Definition
\ref{def-model}.

\vspace{1.5mm}

\par\noindent\textbf{Claim 5:}\; If $h\sim_X h'$ and $h\vDash P(x_1,\ldots, x_n)$ for some $x_1,\ldots, x_n\in X$, then it holds that $h'\vDash
P(x_1,\ldots, x_n)$.

\begin{proof}
 By Claim 4, it holds that $last(h)=_X last(h')$. Also, with the definition of $\|\bullet\|_{\bM}$, $h\vDash P(x_1,\ldots, x_n)$ gives us
 $last(h)\vDash P(x_1,\ldots, x_n)$. Then, by the condition C4 in Definition \ref{def-generalrelationalmodel}, we have $last(h')\vDash P(x_1,\ldots,
 x_n)$, and thus it holds that $h'\vDash P(x_1,\ldots, x_n)$.
\end{proof}

Additionally, we still need to prove the following:

\vspace{1.5mm}

\par\noindent\textbf{Claim 6:}\; $h\sim_{\tom X} h'$ iff $G(h)\sim_XG(h')$.

\begin{proof}
 From \emph{left to right}, suppose that $h\sim_{\tom X} h'$. Thus, the non-redundant path from history $h$ to $h'$ is a $\tom X$-path that it
 is composed of two $\tom X$-chains, one descending and the second ascending, having the same $G$-length $N$. Let
 \begin{align*}
  &h \ot_{\tom X} h_1 \ot_{\tom X}\ldots  h_{m_1-1} \stackrel{G}{\ot} h_{m_1}
\ot_{\tom^2 X} \ldots
h_{m_N-1}\stackrel{G}{\ot}\\
&h_{m_{N}}\ot_{\tom^{N+1} X} \ldots \ot_{\tom^{N+1} X} h_{m_{N}+k}   
 \end{align*}

\noindent be the first (`descending') $\tom X$-chain of $G$-length $N$ on this non-redundant $\tom X$-path running from $h$ to $h'$. Then
immediately, by adding $G(h)$ to the chain as a new end with the transition $\stackrel{G}{\ot}$, we obtain an $X$-chain of $G$-length $N+1$.
Similarly, we can also have an $X$-chain of the same $G$-length now having $G(h')$ as its end from the second ascending $\tom X$-chain.
Consequently, we have that $G(h)\sim_XG(h')$.

\vspace{1mm}

For the direction from \emph{right to left}, we assume that $G(h)\sim_XG(h')$. With the reasoning above, we just need to remove the ends (as
well as the corresponding transitions) of the two $X$-chains forming the $X$-path, and then we are certain to get two $\tom X$-chains that
give us a $\tom X$-path from $h$ to $h'$.
\end{proof}

This completes the proof of Lemma \ref{lemma:tree-proof-standard-model}.
\end{proof}

Finally, we show that the map $last$ preserves the truth values of $\LTD$-formulas:

 \begin{lemma}\label{lemma:p-morphism} 
$last:W\to A$ is a modal $p$-morphism from $\bM$ onto $\mathcal{M}$.
 \end{lemma}

 \begin{proof}
First, surjectivity is obvious, since each $s\in A$ equals $last(s)$. Next, and much less straightforwardly, we must check that the map $last$
satisfies the back-and-forth clauses of modal $p$-morphisms for the dependence relations and for the transition function, as well as the
`harmony' clause for the two kinds of atoms. We state these with their reasons.

\begin{itemize}
\item[$\bullet$] If $h\sim_Xh'$, then $last(h)=_X last(h')$.

This is exactly Claim 4.

\item[$\bullet$] If $last(h)=_Xs$, then there is a history $h'$ with $h\sim_Xh'$ and $last(h')=s$.

For $h'$, we can just take the history  $(h, X, s)$.

\item[$\bullet$] If $G(h)=h'$, then $g(last(h))=last(h')$.

This is our Observation 1.

\item[$\bullet$] If $g(last(h))=s$, then there is some $h'$ with $G(h)=h'$ and $last(h')=s$.

Here it suffices to let  $h'$ be the history $(h,g,s)$.
\end{itemize}

Next, we consider the valuation on atoms. For standard atoms $P{\bm x}$, histories $h$ in $\bM$ agree with their $last$-values $last(h)$
in $\mathcal{M}$ by the definition of $\|\bullet\|_{\bM}$. However, the more challenging case is that of dependence atoms $D_Xy$, since these
get their meaning through the semantics in the standard relational model $\bM$ rather then being imposed by the valuation. Thus, we need to
show the following:
\begin{itemize}
\item[$\bullet$]  \, $h \vDash_{\mathcal{M}} D_Xy \,\,\,\, \textrm{iff} \,\,\,  last(h) \vDash_{\bM} D_Xy$
\end{itemize}

We first make the auxiliary observation involving  dependence statements.

\vspace{1.5mm}

\par\noindent\textbf{Observation 5:}\; If $h\to_{\tom^iX} h'$ and $last(h')\vDash \bigcirc^i D_X Y$, then $last(h)\vDash \bigcirc^i D_X Y$.

\begin{proof}
From $last(h')\vDash \bigcirc^i D_X Y$ it follows that $g^i (last(h'))\vDash D_X Y$. Since $h\to_{\tom^iX} h'$, from Observation 2 it follows
that $g^i(last(h)) =_X g^i(last(h'))$. This, together with $g^i (last(h'))\vDash D_XY$, gives us that $g^i(last(h)) \vDash D_XY$ (again by C3), from which
we conclude $last(h)\vDash \bigcirc^i D_X Y$.  
\end{proof}

Also, it is simple to see that:

\vspace{1.5mm}

\par\noindent\textbf{Observation 6:}\; If $h\stackrel{G}{\to} h'$ and $last(h')\vDash \bigcirc^i D_X Y$, then
$last(h)\vDash \bigcirc^{i+1} D_X Y$.

\vspace{1.5mm}

Next, we have that:

\vspace{1.5mm}

\par\noindent\textbf{Observation 7:}\; If $h\to_{\tom^iX} h'$ and $last(h')\vDash \bigcirc^i D_XY$, then $h\to_{\tom^i Y} h'$.

\begin{proof}
Since $h\to_{\tom^iX} h'$, we have $h'=(h, Z, s')$ for some $Z$ and $s'$, where $last(h)=_Z s'$ and $last(h)\vDash D_Z\tom^i X$. From this, we
get  $s'\vDash D_Z\tom^i X$. Putting this together with $s'\vDash \bigcirc^i D_XY$ and using $\vdash D_X\tom^nY\land\bigcirc^nD_YZ\to
D_X\tom^nZ$ (Dyn-Trans), we can obtain $s'\vDash D_Z\tom^i Y$. It follows that $h\to_{\tom^iY} h'$.
\end{proof}

Next, we spell out the fact about dependence atoms that was needed above.

\vspace{1.5mm}

\par\noindent\textbf{Claim 7:}\;  The following two assertions are equivalent for histories $h$:
\begin{itemize}
\item[(a).] For all histories $h'$, $h\sim_X h'$ implies $h\sim_y h'$.
\item[(b).] $last(h)\vDash D_Xy$.
\end{itemize}

\begin{proof}
\emph{From (a) to (b)}. 
Let $s:=last(h)$. Denote by $h'$ the history $(h,X,s)$,
which is well-defined. Immediately, it holds that $h\sim_Xh'$. So, there exists a $y$-path between $h$ and $h'$, which includes the $X$-path
between $h$ and $h'$. In particular, the immediate variable-succession transition from $h$ to $h'$ is a link in the $y$-chain. More
precisely, the transition should be $\to_y$, which holds by the definition of $y$-chains. Therefore, $s\vDash D_X y$.

\vspace{1.5mm}

\emph{From (b) to (a)}. Assume $last(h)\vDash D_X y$, and let $h'\in W$ with $h\sim_X h'$. We need to show that $h\sim_y h'$. For this, we
look at the $X$-path from $h$ to $h'$, which must include two $X$-chains of some common $G$-length $N$:

\vspace{-4mm}

\begin{align*}
& h=h_0 \ot_X h_1 \ldots \ot_X h_{m_1-1} \stackrel{G}{\ot} h_{m_1}
\ot_{\tom X} \ldots \ot_{\tom X} h_{m_2-1}\\ &\stackrel{G}{\ot} h_{m_2} \ot_{\tom^2 X}\ldots h_{m_N-1}\stackrel{G}{\ot}
h_{m_{N}}\ot_{\tom^N X} \ldots \ot_{\tom^N X} h_{m_{N}+k}
\end{align*}
and
\begin{align*}
& h'=h'_0 \ot_X h'_1 \ldots \ot_X h'_{m'_1-1} \stackrel{G}{\ot} h_{m'_1}
\ot_{\tom X} \ldots \ot_{\tom X} h_{m'_2-1}\\ &\stackrel{G}{\ot} h_{m'_2} \ot_{\tom^2 X}\ldots
h_{m'_N-1}\stackrel{G}{\ot}h_{m'_{N}}\ot_{\tom^N X} \ldots \ot_{\tom^N X} h_{m'_{N}+k'}
\end{align*}
s.t. we have $h_{m_N+k}=h_{m'_N+ k'}=inf_\preceq (h,h')$.

\vspace{1mm}

Using the fact that $last(h)\vDash D_X y$ and repeatedly using Observations 5 and 6, we can show that the last state of every history in the first $X$-chain above satisfies $\bigcirc^i D_Xy$, where $i$ is the number of previous $G$-steps in the chain. Stated more formally: for every number of the form $m_i+q<m_{i+1}$, we have
\begin{center}
 $last(h_{m_i+q})\vDash  \bigcirc^i D_Xy.$
\end{center}

Similarly using the fact that $last(h')\vDash D_X y$ (which follows from  $last(h)\vDash D_X y$ plus the fact that $h\sim_X h'$, together with Claim 4 and condition C3 on Definition \ref{def-generalrelationalmodel}) and then employing Observations 5 and 6, we can prove the analogous fact for the second $X$-chain above: for every number of the form $m'_i+q'<m'_{i+1}$, $$last(h_{m'_i+q'})\vDash \bigcirc^i D_Xy.$$
Using the above facts and Observation 7, we have the following:
\begin{align*}
 &h=h_0 \ot_y h_1 \ldots \ot_y h_{m_1-1} \stackrel{G}{\ot} h_{m_1}
\ot_{\tom y} \ldots \stackrel{G}{\ot} h_{m_2}\\ &\ot_{\tom^2 y}\ldots \stackrel{G}{\ot}
h_{m_N} \ot_{\tom^N y} \ldots \ot_{\tom^N y} h_{m_{N}+k}
\end{align*}
and
\begin{align*}
 & h'=h'_0 \ot_y h'_1 \ldots \ot_y h'_{m'_1-1} \stackrel{G}{\ot} h_{m'_1}
\ot_{\tom y} \ldots \stackrel{G}{\ot} h_{m'_2}\\ &\ot_{\tom^2 y}\ldots \stackrel{G}{\ot}
h_{m'_{N}}\ot_{\tom^N y} \ldots \ot_{\tom^N y} h_{m'_{N}+k'}.
\end{align*}

\vspace{1mm}

Recall that the original $X$-chains (being part of an $X$-path) have the property that $h_{m_N+k}=h_{m'_N+ k'}=inf_\preceq (h,h')$. The two $y$-chains above meet at $inf_\preceq (h,h')$ and so in that case we are done: the path from $h$ to $h'$ is a $y$-path.
\end{proof}

Taking all this together, $last$ is a surjective $p$-morphism from $\bM$ to $\mathcal{M}$.
\end{proof}

This completes the proof for  Theorem \ref{theorem:Repr}.
\end{proof}

\section{Proof for Fact \ref{fact:phi-model-c2}}\label{appendix:proof-for-fact:phi-model-c2}

\begin{proof}
Clauses V1 and V2 guarantee `Determinism' and `Dep-Reflexivity' respectively. We move to considering for `Dep-Transitivity'. Assume $\alpha\in
\|D_XY\|^{\dag}$ and $\alpha\in \|D_YZ\|^{\dag}$. To prove $\alpha\in \|D_XZ\|^{\dag}$, we consider all possible situations why $\alpha\in
\|D_XY\|^{\dag}$ and $\alpha\in \|D_YZ\|^{\dag}$ hold.

\vspace{1.5mm}

(1).  $\alpha\in \|D_XY\|^{\dag}$ holds by V1. Let us consider the cases for $\alpha\in \|D_YZ\|^{\dag}$.

\vspace{1mm}

(1.1). It holds by V1 or V2. Then, by V1, we have $\alpha\in \|D_XZ\|^{\dag}$ (one can check that the condition imposed in V1 is satisfied).

\vspace{1mm}

(1.2). It holds by V3. Then, $D_{Y'} Z\in\alpha$ for some $Y'\subseteq Y$. As $\alpha\in \|D_XY\|^{\dag}$ holds by V1,
$D_{\tom^m\mathbb{V}_{\Phi}} Y'\in \alpha$ for some $m$ with $\tom^m\mathbb{V}_{\Phi}\subseteq X$ and $m\le {\rm{min}}\{\td(y): y\in Y'\}$,
which gives us $D_{\tom^m\mathbb{V}_{\Phi}} Z\in \alpha$. Now, using V3, we obtain $\alpha\in\|D_XZ\|^{\dag}$.

\vspace{1mm}

(1.3). It holds by V4. Then, there are $Y'\subseteq Y$ and $m\le {\rm{min}}\{\td(z): z\in Z\}$ with $D_{Y'}\tom^m\mathbb{V}_{\Phi}\in\alpha$.
Now, one can check that $D_{\tom^n \mathbb{V}_{\Phi}}Y'\in\alpha$ for some $n\in\mathbb{N}$ s.t. $\tom^n \mathbb{V}_{\Phi}\subseteq X$ and
$n\le {\rm{min}}\{\td(y): y\in Y'\}$. Then, $D_{\tom^n \mathbb{V}_{\Phi}}\tom^m\mathbb{V}_{\Phi}\in\alpha$. By V4, it holds that
$\alpha\in\|D_XZ\|^{\dag}$.

\vspace{1.5mm}

(2). $\alpha\in \|D_XY\|^{\dag}$ holds by V2. Then, $Y\subseteq X$. Again, let us consider different cases for $\alpha\in \|D_YZ\|^{\dag}$.

\vspace{1mm}

(2.1). If it holds by V1 or V2, then by the same clause, we have $\alpha\in \|D_XZ\|^{\dag}$.

\vspace{1mm}

(2.2). If it holds by V3, then $D_{Y'}Z\in\alpha$ for some $Y'\subseteq Y$. From $Y\subseteq X$ and V3 it follows that
$\alpha\in\|D_XZ\|^{\dag}$.

\vspace{1mm}

(2.3). If it holds by V4, then there are $Y'$ and $m\in \mathbb{N}$ such that $Y'\subseteq Y$, $m\le {\rm{min}}\{\td(z): z\in Z\}$ and
$D_{Y'}\tom^m\mathbb{V}_{\Phi}\in\alpha$. Now, using $Y\subseteq X$ and V4, we can obtain $\alpha\in\|D_XZ\|^{\dag}$.

\vspace{1.5mm}

(3). $\alpha\in \|D_XY\|^{\dag}$ holds by V3. So, we have $D_{X'}Y\in \alpha$ for some $X'\subseteq X$. Now let us consider $\alpha\in
\|D_YZ\|^{\dag}$.

\vspace{1mm}

(3.1). If it holds by V1, then $\tom^m\mathbb{V}_{\Phi}\subseteq Y$ for some $m\le {\rm {min}}\{\td(z): z\in Z\}$. Then,
$D_{X'}\tom^m\mathbb{V}_{\Phi}\in\alpha$. Now, using V4, we have $\alpha\in \|D_XZ\|^{\dag}$.

\vspace{1mm}

(3.2).  If it holds by V2, then $Z\subseteq Y$. By $D_{X'}Y\in \alpha$, it holds that $D_{X'}Z\in \alpha$. Now, using V3 we can obtain
$\alpha\in \|D_XZ\|^{\dag}$.

\vspace{1mm}

(3.3). If it holds by V3, then  $D_{Y'}Z\in\alpha$ for some $Y'\subseteq Y$. Now, $D_{X'}Y'\in\alpha$, which gives us $D_{X'}Z\in\alpha$.
Recall $X'\subseteq X$, and so from V3 we get $\alpha\in\|D_XZ\|^{\dag}$.

\vspace{1mm}

(3.4). If it holds by V4, then we have $D_{Y'}\tom^m\mathbb{V}_{\Phi}\in\alpha$ for some $Y'\subseteq Y$ and $m\le {\rm{min}}\{\td(z): z\in
Z\}$. From $D_{X'}Y\in\alpha$, it follows that $D_{X'}Y'\in\alpha$, which then gives us $D_{X'}\tom^m\mathbb{V}_{\Phi}\in\alpha$. Now, using
V4, we get $\alpha\in\|D_XZ\|^{\dag}$.

\vspace{1.5mm}

(4). $\alpha\in \|D_XY\|^{\dag}$ holds by  V4. Then,   $D_{X'}\tom^m{\mathbb{V}_{\Phi}}\in\alpha$ for some $X'\subseteq X$ and
$m\le {\rm{min}}\{\td(y): y\in Y\}$. Now we move to analyzing $\alpha\in\|D_YZ\|^{\dag}$.

\vspace{1mm}

(4.1). It holds by V1. Then, there is some $n\in\mathbb{N}$ such that $\tom^n\mathbb{V}_{\Phi}\subseteq Y$ and $n\le {\rm {min}}\{\td(z): z\in
Z\}$. Also, $m\le n$. Now, using V4, we have $\alpha\in\|D_XZ\|^{\dag}$.

\vspace{1mm}

(4.2). It holds by V2. Then, we obtain $Z\subseteq Y$. So, ${\rm{min}}\{\td(z): z\in Z\}\ge {\rm{min}}\{\td(y): z\in Y\}$. From  $m\le
{\rm{min}}\{\td(y): y\in Y\}$, we know that $m\le {\rm{min}}\{\td(z): z\in Z\}$. By V4, $\alpha\in\|D_XZ\|^{\dag}$.

\vspace{1mm}

(4.3). It holds by V3. So, $D_{Y'}Z\in\alpha$ for some $Y'\subseteq Y$. Notice that we have $D_{\tom^{m}\mathbb{V}_{\Phi}}Y'\in\alpha$. Thus,
$D_{X'}Y'\in\alpha$, and so $D_{X'}Z\in\alpha$. As $X'\subseteq X$, using V3 we have $\alpha\in\|D_XZ\|^{\dag}$.

\vspace{1mm}

(4.4). It holds by V4. Then, we have $D_{Y'}\tom^n\mathbb{V}_{\Phi}\in\alpha$ for some $Y'\subseteq Y$ and $n\le {\rm{min}}\{\td(z): z\in Z\}$. Since
$m\le {\rm{min}}\{\td(y): y\in Y\}$ while also $Y'\subseteq Y$, it holds that $m\le {\rm{min}}\{\td(y): y\in Y'\}$. Now, from
$D_{X'}\tom^m{\mathbb{V}_{\Phi}}\in\alpha$, it follows that $D_{X'}Y'\in\alpha$. Therefore, $D_{X'}\tom^n\mathbb{V}_{\Phi}\in\alpha$.
Finally, using V4 we obtain $\alpha\in\|D_XZ\|^{\dag}$.
\end{proof}

\section{Proof for Fact \ref{fact:phi-model-c3}}\label{appendix:proof-for-fact:phi-model-c3}

\begin{proof}
When $\alpha=\beta=\emptyset$, it is not hard to check $\alpha\approx_Y \beta$ and $\beta\vDash D_XY$. In what follows, we show that for
$\alpha\not=\emptyset$. There are different cases.

\vspace{1.5mm}

(1). First, assume that $\alpha \approx_X \beta $ holds by E1. Then, $\td(X)\le\td(\alpha)$.  We consider the different reasons why
$\alpha\in\|D_XY\|^{\dag}$ holds.

\vspace{1mm}

(1.1). It holds by V1. Then, $\tom^m\mathbb{V}_{\Phi}\subseteq X$ for some $m\le {\rm{min}}\{\td(y): y\in Y\}$. As
$\td(X)\le\td(\alpha)$, we have $D_X\tom^m \mathbb{V}_{\Phi}\in\alpha$. Now, using E1, we can infer $D_X\tom^m \mathbb{V}_{\Phi}\in\beta$. From
V4, it follows that $\beta\vDash D_XY$. It remains to show $\alpha\approx_Y\beta$.

\vspace{1mm}

(1.1.1). Suppose $\td(Y)\le \td(\alpha)$. Then, from $\alpha\in\|D_XY\|^{\dag}$,  it follows that $D_XY\in\alpha$. Now, recall that
$\alpha\approx_X\beta$, and by Fact \ref{fact:observation-on-relations-of-phi-model}, it holds that $\alpha\approx_Y\beta$.

\vspace{1mm}

(1.1.2). Suppose $\td(Y)> \td(\alpha)$. Now, by the same reasoning as that in (1.1.1), but now using $\tom^m\mathbb{V}_{\Phi}$ in place of
$Y$, we can show $\alpha\approx_{\tom^m\mathbb{V}_{\Phi}}\beta$, where $m$ has already been specified. Now, using V2, we have
$\alpha\approx_Y\beta$.

\vspace{1mm}

Thus, when $\alpha\in\|D_XY\|^{\dag}$ holds by V1, we have $\alpha\approx_Y\beta$ and $\beta\vDash D_XY$.

\vspace{1mm}

(1.2). $\alpha\in\|D_XY\|^{\dag}$ holds by V2. Then, $Y\subseteq X$. So, $\td(Y)\le \td(\alpha)$. Then, $D_XY\in\alpha$. Since
$\alpha\approx_X\beta$, $D_XY\in\beta$. Now, by Fact \ref{fact:observation-on-relations-of-phi-model}, it holds that $\alpha\approx_Y\beta$.

\vspace{1mm}

(1.3). $\alpha\in\|D_XY\|^{\dag}$ holds by V3. Then, $D_{X'}Y\in\alpha$ for some $X'\subseteq X$. Recall that $\td(X)\le \td(\alpha)$. Thus,
$D_XX'\in\alpha$. Then, $D_XY\in\alpha$. Now, from Fact \ref{fact:observation-on-relations-of-phi-model} it follows that
$\alpha\approx_Y\beta$.

\vspace{1mm}

(1.4). It holds by V4. Then, it holds that  $D_{X'}\tom^n\mathbb{V}_{\Phi}\in\alpha$ for some  $X'\subseteq X$ and $n\le{\rm{min}}\{\td(y): y\in Y\}$. Using
$X'\subseteq X$ and $\td(X)\le\td(\alpha)$, we can infer $D_XX'\in\alpha$. Then, it holds that $D_X\tom^n\mathbb{V}_{\Phi}\in\alpha$. Recall
that $\alpha\approx_X\beta$ holds by E1. Then, by V4, it holds that $\beta\vDash D_XY$. Now, it suffices to show that
$\alpha\approx_{Y}\beta$. Again, there are different cases: $\td(Y)\le \td(\alpha)$ or $\td(Y)> \td(\alpha)$.

 When $\td(Y)\le \td(\alpha)$, we can show it by the same reasoning as that for (1.1.1). When $\td(Y)> \td(\alpha)$, it
suffices to prove $\alpha\approx_{\tom^n\mathbb{V}_{\Phi}}\beta$, which holds by Fact \ref{fact:observation-on-relations-of-phi-model}.

 \vspace{1mm}

So, when $\alpha\approx_X\beta$ holds by E1, both $\alpha\approx_Y\beta$ and $\beta\vDash D_XY$ hold.

\vspace{1.5mm}

 (2). Next, we consider the case that $\alpha\approx_X\beta$ holds by E2. So, $\td(X)>\td(\alpha)$. Also, there is some
 $m\le{\rm{min}}(\{\td(\alpha)\}\cup\{\td(x): x\in X\})$ such that $\alpha\approx_{\tom^m\mathbb{V}_{\Phi}}\beta$ holds by E1. We now consider
 the cases for the reason why $\alpha\in\|D_XY\|^{\dag}$ holds.

\vspace{1mm}

(2.1). $\alpha\in\|D_XY\|^{\dag}$ holds by V1. Then, it holds that $\tom^n\mathbb{V}_{\Phi}\subseteq X$ for some $n\le {\rm{min}}\{\td(y): y\in
Y\}$. Using V1 again, we can obtain $\beta\in\|D_XY\|^{\dag}$. Moreover, it is obvious that $m\le n$. We now proceed to prove $\alpha\approx_Y\beta$. There
are different situations: $\td(Y)\le\td(\alpha)$ or $\td(Y)>\td(\alpha)$.

\vspace{1mm}

(2.1.1). For the case that $\td(Y)\le\td(\alpha)$, we must prove $\alpha\approx_Y\beta$ with E1. This is  given by
Fact \ref{fact:observation-on-relations-of-phi-model}: to see this, observe that, crucially, $D_{\tom^m\mathbb{V}_{\Phi}}Y\in\alpha$.

\vspace{1mm}

(2.1.2). The case that $\td(Y)>\td(\alpha)$ is trivial: with the observation that $m\le {\rm{min}}\{\td(y): y\in Y\}$, it holds by the same reason
with that of $\alpha\approx_X\beta$.

\vspace{1mm}

(2.2). $\alpha\in\|D_XY\|^{\dag}$ holds by V2.  Then, $Y\subseteq X$. Immediately, by V2 we also have $\beta\in\|D_XY\|^{\dag}$. It remains
to prove $\alpha\approx_Y\beta$, and there are two cases: $\td(Y)> \td(\alpha)$ or $\td(Y)\le \td(\alpha)$. If the former  holds,
 $\alpha\approx_Y\beta$ holds by the same reason as  for $\alpha\approx_X\beta$. Let us move to the latter. Recall $m\le \td(\alpha)$.
As $Y\subseteq X$, we have $D_{\tom^m\mathbb{V}_{\Phi}}Y\in\alpha$. Now, from Fact
\ref{fact:observation-on-relations-of-phi-model}, it follows that $\alpha\approx_Y\beta$.

\vspace{1mm}

(2.3). $\alpha\in\|D_XY\|^{\dag}$ holds by V3. Then, $D_{X'}Y\in\alpha$ for some $X'\subseteq X$. As $m\le{\rm{min}}\{\td(x): x\in X\}$,
we have $D_{\tom^m\mathbb{V}_{\Phi}} X'\in\alpha$. Again, notice that $\alpha\approx_{\tom^m\mathbb{V}_{\Phi}}\beta$ holds by E1, and using
its further relevant clauses, we can  show that  $\beta\in\|D_XY\|^{\dag}$ and $\alpha\approx_Y\beta$.

\vspace{1mm}

(2.4). $\alpha\in\|D_XY\|^{\dag}$ holds by V4. Then,  $D_{X'}\tom^n\mathbb{V}_{\Phi}\in\alpha$ for some $X'\subseteq X$ and
$n\le{\rm{min}}\{\td(y): y\in Y\}$. Now, we also have $D_{\tom^m\mathbb{V}_{\Phi}}X'\in\alpha$. Since
$\alpha\approx_{\tom^m\mathbb{V}_{\Phi}}\beta$, by Fact \ref{fact:observation-on-relations-of-phi-model} it holds that
$\alpha\approx_{X'}\beta$, and then we have $D_{X'}\tom^n\mathbb{V}_{\Phi}\in\beta$, which  gives us $\beta\vDash D_XY$ (using V4).
Moreover, we can prove $\alpha\approx_{\tom^n\mathbb{V}_{\Phi}}\beta$, which implies $\alpha\approx_Y\beta$.

\vspace{1.5mm}

Thus, when $\alpha\approx_X\beta$ holds by E2, we also have $\alpha\approx_Y \beta$ and $\beta\vDash D_XY$.
\end{proof}

\section{Proof for Fact \ref{fact:phi-model-c6}}\label{appendix:proof-for-fact:phi-model-c6}

\begin{proof}
Suppose $\alpha\approx_{\tom X}\beta$ and $\td(\alpha)=i$. Then, it holds that $\td(\beta)=i$. So, $\td(G(\alpha))=\td(G(\beta))$. There are two situations: $1+\td(X)\le i$ or $1+\td(X)> i$.

\vspace{1.5mm}

(1). We begin with the first case. Then, we have $\td(X)\le \td(G(\alpha))$. Assume that $D_XY\in G(\alpha)$.

\vspace{1mm}

First, $D_XY\in G(\alpha)$ implies $\tom D_XY\in \alpha$. Then, by P4,   $\D_X\tom D_XY\in \Phi_i$. Also, by P3, $\tom \D_XD_XY\in
\Phi_i$. Now it is easy to see $\tom \D_XD_XY\in \alpha$, which implies $\D_{\tom X}\tom D_XY\in \alpha$. Obviously, $D_{\tom X}\tom X\in
\alpha$. Thus, from $\alpha\approx_{\tom X}\beta$, it follows that $\D_{\tom X}\tom D_XY\in \beta$. Now, it is simple to see that $D_XY\in
G(\beta)$.

\vspace{1mm}

Next, assume $\D_Y\varphi\in G(\alpha)$. Then, $\tom \D_Y\varphi\in \alpha$. By P4, $\D_{\tom Y} \tom \D_Y\varphi\in \Phi_i$. Then, we can
infer $\D_{\tom Y} \tom \D_Y\varphi\in \alpha$. Also, notice that $D_{\tom X}\tom Y\in\alpha$. So, $\alpha\approx_{\tom X}\beta$ implies
$\D_{\tom Y} \tom \D_Y\varphi\in \beta$. Thus, $\tom \D_Y\varphi\in \beta$, and so $\D_Y\varphi\in G(\beta)$. Hence, we conclude that
$G(\alpha)\approx_X G(\beta)$.

\vspace{1.5mm}

(2). Let us move to the second case. Then, $\td(X)> \td(G(\alpha))$. Now, we have $\alpha\approx_{\tom ^m\mathbb{V}_{\Phi}}\beta$ for some
$m\le
{\rm{min}}(\{\td(\alpha)\}\cup\{\td(\tom x): x\in X\})$.

\vspace{1mm}

When $m\ge 1$, the present case can be reduced to the earlier case (1) and we obtain that
$G(\alpha)\approx_{\tom^{m-1}\mathbb{V}_{\Phi}}G(\beta)$, which implies $G(\alpha)\approx_{X}G(\beta)$.

\vspace{1mm}

When $m= 0$ (i.e., $\alpha\approx_{\mathbb{V}_{\Phi}}\beta$), to show $G(\alpha)\approx_XG(\beta)$, it suffices to prove
$G(\alpha)\approx_{\mathbb{V}_{\Phi}}G(\beta)$. The case that $G(\alpha)=G(\beta)=\emptyset$ is trivial, and we merely consider the case that
they are not empty. Now, $\td(\alpha)=\td(\beta)\ge 1$. Also, it is simple to check that $\alpha\approx_{\tom\mathbb{V}_{\Phi}}\beta$. Again,
this can be reduced to the case (1) and we can get $G(\alpha)\approx_{\mathbb{V}_{\Phi}}G(\beta)$. This completes the proof.
\end{proof}




\end{document}